\documentclass[11pt]{article}
\usepackage{amsmath,amssymb, amsthm, hyperref}
\usepackage{algorithm}
\usepackage[algo2e]{algorithm2e} 

\usepackage{enumitem}
\providecommand{\keywords}[1]{\textbf{\textit{Keywords---}} #1}
\makeatletter
\renewcommand*\env@matrix[1][*\c@MaxMatrixCols c]{%
  \hskip -\arraycolsep
  \let\@ifnextchar\new@ifnextchar
  \array{#1}}
\makeatother

\usepackage{setspace}
\usepackage{amsmath}
\newcommand{\argmin}[1]{\underset{#1}{\mathrm{argmin}}}
\newcommand{\prox}{{\rm prox}}
\newcommand{\rprox}{{\rm rprox}}
\newcommand{\supp}{{\rm supp}}
\usepackage{algorithm, algpseudocode}
\usepackage[table]{xcolor}
\usepackage{graphicx}
\usepackage{epstopdf}
\usepackage{subfigure}
\usepackage{dsfont}
\usepackage{multicol,multirow}
\usepackage{amsmath}
\usepackage{physics}
\usepackage{empheq}

\usepackage{array}
\newcolumntype{P}[1]{>{\centering\arraybackslash}p{#1}}

\newcommand{\Var}{\mathrm{Var}}
\theoremstyle{plain}
\newtheorem{theorem}{Theorem}[section]
\newtheorem*{theorem*}{Theorem}

\newtheorem{proposition}[theorem]{Proposition}

\theoremstyle{definition}

\theoremstyle{remark}

\usepackage{authblk}

\title{Extracting structured dynamical systems using sparse optimization with very few samples}

\author[1]{Hayden Schaeffer\thanks{schaeffer@cmu.edu}}
\author[2]{Giang Tran\thanks{giang.tran@uwaterloo.ca}}
\author[3]{Rachel Ward\thanks{rward@math.utexas.edu}}
\author[1]{Linan Zhang\thanks{linanz@andrew.cmu.edu}}
\affil[1]{Department of Mathematical Sciences, Carnegie Mellon University, Pittsburgh, PA.}
\affil[2]{Department of Applied Mathematics, University of Waterloo, Waterloo, Ontario, Canada.}
\affil[3]{Department of Mathematics, The University of Texas at Austin, Austin, TX.}

\date{}
\begin{document}
\maketitle

\begin{abstract}
Learning governing equations allows for deeper understanding of the structure and dynamics of data. We present a random sampling method for learning structured dynamical systems from under-sampled and possibly noisy state-space measurements. The learning problem takes the form of a sparse least-squares fitting over a large set of candidate functions. Based on a Bernstein-like inequality for partly dependent random variables, we provide theoretical guarantees on the recovery rate of the sparse coefficients and the identification of the candidate functions for the corresponding problem. Computational results are demonstrated on datasets generated by the Lorenz 96 equation, the viscous Burgers' equation, and the two-component reaction-diffusion equations (which is challenging due to parameter sensitives in the model). This formulation has several advantages including ease of use, theoretical guarantees of success, and computational efficiency with respect to ambient dimension and number of candidate functions.
\end{abstract}

\keywords{High Dimensional Systems, Dynamical Systems, Sparsity, Model Selection, Exact Recovery, Cyclic Permutation, Coherence}

\section{Introduction}\label{sec:introduction}

Automated model selection is an important task for extracting useful information from observed data. One focus is to create methods and algorithms which allow for the data-based identification of governing equations that can then be used for more detailed analysis of the underlying structure and dynamics. The overall goal is to develop computational tools for reverse engineering equations from data. Automated model selection has several advantages over manual processing, since algorithmic approaches allow for the inclusion of a richer set of potential candidate functions, which thus allow for complex models and can be used to process and fit larger data sets. However, several factors restrict the computational efficiency; for example, as the dimension of the variables grows, the size of the set of possible candidate functions grows rapidly. In this work, we present some computational strategies for extracting governing equation when additional structural information of the data is known.

Data-driven methods for model selection have several recent advances. The authors of \cite{bongard2007automated, schmidt2009distilling} developed an approach for extracting physical laws (\textit{i.e.} equations of motion, energy, \textit{etc}.) from experimental data. The method uses a symbolic regression algorithm to fit the derivatives of the time-series data to the derivatives of candidate functions while taking into account accuracy versus simplicity. In \cite{brunton2016discovering}, the authors proposed a sparse approximation approach for selecting governing equations from data. One of the key ideas in \cite{brunton2016discovering} is to use a fixed (but large and possibly redundant) set of candidate functions in order to write the model selection problem as a linear system. The sparse solutions of this system are those that are likely to balance simplicity of the model while preserving accuracy. To find a sparse approximation, \cite{brunton2016discovering} uses a sequential least-squares thresholding algorithm which includes a thresholding sub-step (sparsification) and a least-squares sub-problem (fitting). Several sparsity-based methods were developed in order to solve the model selection problem. The authors of \cite{tran2017exact} proposed an optimization problem for extracting the governing equation from chaotic data with highly corrupted segments (of unknown location and length). Their approach uses the $\ell^{2,1}$-norm (often referred to as the group sparse or joint sparse penalty), in order to detect the location of the corruption, coupling each of the state-variables. In \cite{tran2017exact}, it was proven that, for chaotic systems, the solution to the optimization problem will locate the noise-free regions and thus extract the correct governing system.  In \cite{schaeffer2017learning}, the authors used a dictionary of partial derivatives along with a LASSO-based approach \cite{tibshirani1996regression} to extract the partial differential equation that governs some (possibly noisy) spatiotemporal dataset. An adaptive ridge-regression version of the method from \cite{brunton2016discovering} was proposed in \cite{rudy2017data} and applied to fit PDE to spatiotemporal data. One computational issue that arrises in these approaches is that noise on the state-variables are amplified by numerical differentiation. To lessen the effects of noise,  sparse learning can be applied to the integral formulation of the differential equation along with an integrated candidate set as done in \cite{schaeffer2017sparse}. The exact recovery of the governing equation can be guaranteed when there is sufficient randomness in the data, even in the under-sampling limit. In \cite{schaeffer2017extracting}, using random initial conditions, an $\ell^1$-penalized optimization problem was shown to recover the underlying differential equation with high probability, and several sampling strategies were discussed.  In order to allow for variations in the coefficients, a group-sparse recovery model was proposed in \cite{SchaefferTranWard17sep}. Using information criteria, \cite{mangan2017model} proposed a method to choose the ``optimal'' model learned from the algorithm in \cite{brunton2016discovering} as one varies the thresholding parameter.

There have been several recent methods using general sparse approximations techniques for learning governing dynamical systems, including: SINDy with control \cite{BruntonProctorKutz16IFAC}, the SINO method \cite{SorokinaSygletosTuritsyn16}, an extension of SINDy to stochastic dynamics \cite{BoninsegnaNuskeClementi17}, sparse identification of a predator-prey system \cite{DamBronsRasmussenNaulinHesthaven17}, SINDy with rational candidate functions \cite{ManganBruntonProctoKutz16}, rapid-SINDy  \cite{ QuadeAbelKutzBrunton18}, the unified sparse dynamics learning (USDL) algorithm which uses a weak formulation with the orthogonal matching pursuit (OMP) algorithm \cite{PantazisTsamardinos17}. Sparsity inducing and/or data-driven algorithms have been applied to other problems in scientific computing, including: sparse spectral methods for evolution equations \cite{schaeffer2013sparsepde,mackey2014compressive}, sparse penalties for obstacle problems \cite{tran20151}, sparse-low energy decomposition for conversation laws \cite{hou2015sparse}, sparse spatial approximations for PDE \cite{caflisch2015pdes}, sparse weighted interpolation of functions \cite{rauhut2016interpolation}, leveraging optimization algorithms in nonlinear PDE   \cite{ schaeffer2016accelerated}, sparse approximation for polynomial chaos \cite{peng2016polynomial}, high-dimensional function approximation using sparsity in lower sets \cite{adcock2017polynomial}, learning PDE through convolutional neural nets \cite{LongLuMaDong17}, modeling dynamics through Gaussian processes \cite{raissi2017numerical,RaissiKarniadakis18}, and constrained Galerkin methods \cite{LoiseauBrunton18}.

\subsection{Contributions of this Work.}

We present an approach for recovering governing equations from under-sampled measurements using the burst sampling methodology in \cite{schaeffer2017extracting}. In this work, we show that if the components of the governing equations are similar, \textit{i.e.} if each of the equations contains the same active terms (after permutation of the indices), then one can recover the coefficients and identify the active basis terms using fewer random samples than required in \cite{schaeffer2017extracting}. The problem statement and construction of the permuted data and dictionaries are detailed in Section \ref{section:problem}. In essence, after permutation, the dictionary matrix is still sufficiently random and maintains the necessary properties for exact and stable sparse recovery. Theoretical guarantees on the recovery rate of the coefficients and identification of the candidate functions (the support set of the coefficients) are provided in Section \ref{sec:theory}. The proofs rely on a Bernstein-like inequality for partially dependent measurements. The algorithm uses the Douglas-Rachford iteration to solve the $\ell^1$ penalized least-squares problem. In Section \ref{sec:numerics}, the algorithm and the data processing are explained\footnote{The code is available on \url{https://github.com/linanzhang/SparseCyclicRecovery}.}. In Section \ref{sec:computing}, several numerical results are presented, including learning the Lorenz 96 system, the viscous Burgers' equation, and a two-component reaction-diffusion system. These examples include third-order monomials, which extend the computational results of \cite{schaeffer2017extracting}. These problems are challenging due to their sensitivities of the dynamics to parameters, for example, shock locations or the structure of patterns. Our approach is able to recover the dynamics with high probability, and in some cases, we can recover the governing dynamical system from one-sample!

\section{Problem Statement}\label{section:problem}

Consider an evolution equation $\dot{u}=f(u)$, where $u(t)\in \mathbb{R}^n$ and the initial data is $u(t_0)=u_0$.  Assume that $f$ is a polynomial vector-valued equation in $u$. The evolution equation can be written component-wise as:
\begin{empheq}[left=\empheqlbrace]{align*}
\dot{u}_1 &= f_1(u_1,\ldots, u_n)\\
\dot{u}_2 &= f_2(u_1,\ldots, u_n)\\
&\ \ \vdots \\
\dot{u}_n&= f_n(u_1,\ldots, u_n).
\end{empheq}
From limited measurements on the state-space $u$, the objective is to extract the underlying model $f$. In \cite{schaeffer2017extracting}, this problem was investigated for general (sparse) polynomials $f$ using several random sampling strategies, including a burst construction. The data is assumed to be a collection of $K$-bursts, \textit{i.e.} a short-time trajectory:
\begin{align*}
\{ u(t_1; k), u(t_2; k), \ldots, u(t_{m-1}; k)\},
\end{align*}
associated with some initial data $u(t_0; k)$, where $u(\cdot\,;k)$ denotes the $k$-th burst, $1\leq k \leq K$. In addition, we assume that the time derivative associated with each of the measurements in a burst, denoted by $\{\dot{u}(t_0;k), \dot{u}(t_1;k), \dot{u}(t_2;k), \ldots, \dot{u}(t_{m-1};k)\}$, can be accurately approximated. Define the matrix $M$ as the collection of all monomials (stored column-wise), 
\begin{equation*}
M(k) = \begin{bmatrix} & M^{(0)}(k) &  | &  M^{(1)}(k) & | &  M^{(2)}(k) & | &  \cdots & \end{bmatrix},\label{eqn:dictionaryburstmonomial0}
\end{equation*}
where the sub-matrices are the collections of the constant, linear, quadratic terms and so on,
\begin{equation*}
M^{(0)}(k) =\begin{pmatrix} 1 \\ 1 \\ 1 \\ \vdots \\ 1 \end{pmatrix} \in\mathbb{R}^m,
\end{equation*} 
\begin{equation*}
M^{(1)}(k) =\begin{pmatrix}
 u_1(t_0;k) & u_2(t_0;k) & \cdots & u_n(t_0;k) \\
 u_1(t_1;k) & u_2(t_1;k) & \cdots & u_n(t_1;k)  \\
 u_1(t_2;k) & u_2(t_2;k) & \cdots & u_n(t_2;k)  \\
 \vdots & \vdots & \ddots & \vdots \\
 u_1(t_{m-1};k) & u_2(t_{m-1};k) & \cdots & u_n(t_{m-1};k)  
\end{pmatrix},
\end{equation*} 
and
\begin{equation*}
M^{(2)}(k) =\begin{pmatrix}
 u^2_1(t_0;k) & u_1(t_0;k)u_2(t_0;k) & \cdots &u^2_n(t_0;k)  \\
 u^2_1(t_1;k) & u_1(t_1;k)u_2(t_1;k) & \cdots &u^2_n(t_1;k)   \\
 u^2_1(t_2;k) & u_1(t_2;k)u_2(t_2;k) & \cdots &u^2_n(t_2;k)  \\
 \vdots & \vdots & \ddots & \vdots \\
u^2_1(t_{m-1};k) & u_1(t_{m-1};k)u_2(t_{m-1};k) & \cdots & u^2_n(t_{m-1};k)\end{pmatrix}.
\end{equation*}
Define the velocity matrix at each of the corresponding measurements: 
\begin{equation*}
v(k)=\begin{pmatrix}
\dot{u}_1(t_0;k)& \dot{u}_2(t_0;k) & \cdots &\dot{u}_n(t_0;k)\\
\dot{u}_1(t_1;k)& \dot{u}_2(t_1;k) & \cdots &\dot{u}_n(t_1;k)\\
\dot{u}_1(t_2;k)& \dot{u}_2(t_2;k) & \cdots &\dot{u}_n(t_2;k)\\
\vdots & \vdots & \ddots & \vdots \\
\dot{u}_1(t_{m-1};k)& \dot{u}_2(t_{m-1};k) & \cdots & \dot{u}_n(t_{m-1};k)
\end{pmatrix} .
\end{equation*}
Using this data, one would like to extract the coefficients for each of the components of $f$, \textit{i.e.} $f_j$, which is denoted by a column vector $c_j$. The collection of coefficients can be defined as:
\begin{equation}
C= \begin{pmatrix}
| & | &  & | \\ c_1 & c_2 & \cdots & c_n \\  | & | &  & |
\end{pmatrix}.
\label{eq:coefficientstotal}
\end{equation} 
Therefore, the problem of identifying the sparse polynomial coefficients associated with the model $f$ is equivalent to finding a sparse matrix $C$ such that $v(k) = M(k) C$ for all bursts $k$.

In \cite{schaeffer2017extracting}, it was shown that finding a sparse matrix $C$ from $v(k) = M(k) C$ with randomly sampled initial data was achievable using an $\ell^1$ optimization model. In particular, with probability $1-\epsilon$,  $C$ can be recovered exactly from limited samples as long as the number of samples $K$ satisfies $K\sim s \log(N)\log(\epsilon^{-1})$, where $s$ is the maximum sparsity level among the $n$ columns of $C$ and $N$ is the number of basis functions. In this work, using a coherence bound, we get a \textbf{sampling rate that scales like $n^{-1}\, \log(n)$ when the governing equation has a structural condition} relating each of the components of the model $f$.

\subsection{A Cyclic Condition}

When structural conditions on $f$ and $u$ can be assumed \textit{a priori}, one expects the number of initial samples needed for exact recovery to decrease. One common assumption is that the components of the model, $f_j$, are cyclic (or index-invariant), \textit{i.e.} for all $1\leq i,j \leq n$, we have:
\begin{equation*}
f_j(u_1,u_{2}, \ldots, u_{n})=f_i(u_{j-i+1},u_{j-i+2}, \ldots,u_n,u_1,\ldots, u_{j-i+n}),
\end{equation*} 
where $u_{n+q} = u_q$ and $u_{-q} = u_{n-q}$, for all $0\leq q\leq (n-1)$. In particular, all components $f_j$ can be obtained by determining just one component, say $f_1$, since:
\begin{equation*}
f_j(u_1,u_{2}, \ldots, u_{n}) = f_1(u_{j},u_{j+1}, \ldots,u_n,u_1,\ldots, u_{j-1+n}).
\end{equation*}
The goal is to determine $f$ (by learning $f_1$), given observations of $u$ and an (accurate) approximation of $\dot{u}$.

The physical meaning behind the cyclic condition relates to the invariance of a model to location/position. For example, the Lorenz 96 system in $n>3$ dimensions is given by \cite{lorenz1996predictability}:
\begin{equation*}
\dot{u}_j = -u_{j-2}\, u_{j-1} + u_{j-1} \, u_{j+1} - u_j + F, \quad j=1,2,\ldots,n,
\label{eqn:lorenz96}
\end{equation*}
for some constant $F$ (independent of $j$) and with periodic conditions, $u_{-1}=u_{n-1}$, $u_0=u_n$, and $u_{n+1}=u_1$. Each component of $f$ follows the same structure, and is invariant to the index $j$.
This is also the case with spatiotemporal dynamics which are not directly dependent on the space variable. For a simple example, consider the standard discretization of the heat equation in one spatial dimensional periodic domain:
\begin{equation*}
\dot{u}_j  =\frac{1}{h^2} \left(u_{j-1}-2u_{j}+u_{j+1}\right),
\end{equation*}
with grid size $h>0$ and periodic conditions, $u_0=u_n$ and $u_{n+1}=u_1$. The system is invariant to spatial translation, and thus satisfies the cyclic condition.

Extending the results from \cite{schaeffer2017extracting}, we show that recovering the governing equations from only one under-sampled measurement is tractable when the model $f$ has this cyclic structure. This is possible since one measurement of $u(t)$ will provide us with $n$-pieces of information for $f_1$ (and thus the entire model $f$). Under this assumption, the problem of determining the coefficient matrix $C$, defined by Equation~\eqref{eq:coefficientstotal}, reduces to the problem of determining the first column of $C$, \textit{i.e.} $c_1$. For simplicity, we can drop the subscript and look for a coefficient vector $c\in \mathbb{R}^N$ ($N$ is the number of candidate functions) that fits the dynamic data.

The construction of the optimization problem and computational method are detailed in the subsequent sections. To summarize, we define what cyclic permutations are and explain how to build the associated dictionary matrix. This construction is detailed for the one spatial dimensional case, since general spatial dimensions follow from a vectorization of the $n$-dimensional problem.  Following the second strategy in \cite{schaeffer2017extracting}, a subset of the domain is considered (via localization and restriction of the dictionary terms), which leads to a smaller, but still underdetermined, problem.  The dictionary is transformed to the tensorized Legendre basis in order to guarantee an incoherence principle on the system. Lastly, the coefficients of the governing equations are learned via an $\ell^1$ penalized basis pursuit problem with an inexact (noise-robust) constraint.

\subsection{Cyclic Permutations}

The (multi-sample) data matrix is computed using a set of cyclic permutations from very few samples. The set of cyclic permutations, $\mathcal{C}_n$, is a subset of all permutations of $[n]:=\{ 0, 1, \ldots, n-1\}$,  whose elements are shifted by a fixed amount. There are $n$ cyclic permutations out of the $n!$ possible permutations of $[n]$. In addition, the corresponding $n\times n$ permutation matrices of a cyclic permutation are all circulant.  For example, $\mathcal{C}_3$ contains three permutations of the set $\{0 ,1, 2\}$ (out of a total of six possible permutations), \textit{i.e.} $\{0, 1, 2\}$, $\{1, 2, 0\}$, and $\{2, 0, 1\}$ whose permutation matrices are:
\begin{equation*}
P_1 =\begin{pmatrix}
1 & 0 & 0 \\ 0 & 1 & 0\\ 0 & 0& 1 \\
\end{pmatrix},
\hspace{1cm}
P_2 =\begin{pmatrix}
0 & 1 & 0 \\ 0 & 0 & 1\\ 1 & 0& 0 \\
\end{pmatrix},
\hspace{1cm}
P_3 =\begin{pmatrix}
0 & 0 & 1 \\ 1 & 0 & 0\\ 0 & 1& 0 \\
\end{pmatrix}.
\end{equation*}

The importance of the cyclic permutations is that they preserve the natural ordering between elements, since the left and right neighbors of any element are fixed (with periodic conditions at the first and last elements).

\subsection{Dictionary for One Spatial Dimension}\label{section:oned}

Consider the sequence (indexed by $k$) of discrete measurements $\{ u(t_1;k), \, u(t_2;k), \ldots, \, u(t_{m-1};k)\}$, obtained through either $k$-simulations or $k$-observations. Assume that the data is a discretization of a system with one spatial variable $u(t;k) \in \mathbb{R}^n$ for $t\in\mathbb{R}$ and $k\in\mathbb{N}$. For general spatial dimensions, the variables are multidimensional arrays. In particular, after discretization and vectorization, the spatially dependent function $u(t,x)$ is converted to a 1D array (vector), $u(t,x,y)$ is converted to a 2D array (matrix), and $u(t,x,y,z)$ is converted to a 3D array, {\it etc}.. As in \cite{schaeffer2017extracting}, the total number of temporal samples, denoted by $m$, is small, and thus we refer to the short-time trajectory as a burst. Each of the bursts is initialized by data sampled from a random distribution.

Given one measurement  $u(t_0;k) \in \mathbb{R}^n$, we obtain ``multiple" measurements by considering the collection of all cyclic permutations of the data vector $u(t_0;k)$. In particular, we can construct the $n$-measurement matrix,
\begin{equation}
U(t_0;k)=\begin{pmatrix}[cccc]
 u_1(t_0;k) & u_2(t_0;k) & \cdots & u_n(t_0;k) \\
u_2(t_0;k) & u_3(t_0;k) & \cdots & u_1(t_0;k) \\
 u_3(t_0;k) & u_4(t_0;k) & \cdots & u_2(t_0;k)\\
\vdots & \vdots & \ddots & \vdots \\
 u_n(t_0;k) & u_1(t_0;k) & \cdots & u_{n-1}(t_0;k)\\
\end{pmatrix}.
\label{eqn:datamatrix}
\end{equation}
To build the dictionary matrix, we collect all monomials of $U$. The quadratic matrix, denoted by $U^2$, is defined as:
\begin{equation}
U^2(t_0;k)= \begin{pmatrix}[cccc]
 u^2_1(t_0;k) & u_1(t_0;k)u_2(t_0;k) & \cdots &u^2_n(t_0;k)\\
 u^2_2(t_0;k) & u_2(t_0;k)u_3(t_0;k) & \cdots &u^2_1(t_0;k)\\
 u^2_3(t_0;k) & u_3(t_0;k)u_4(t_0;k) & \cdots &u^2_2(t_0;k)\\
\vdots & \vdots & \ddots & \vdots \\
 u^2_n(t_0;k) & u_n(t_0;k)u_1(t_0;k) & \cdots &u^2_{n-1}(t_0;k)\\
\end{pmatrix},
\label{eqn:quadmatrix}
\end{equation}
and the cubic matrix, denoted by $U^3$, is defined as:
\begin{equation}
U^3(t_0;k)= \begin{pmatrix}[cccccc]
 u^3_1(t_0;k) & u^2_1(t_0;k)u_2(t_0;k) & \cdots & u_1(t_0;k)u_2(t_0;k)u_3(t_0;k) & \cdots &u^3_n(t_0;k)\\
 u^3_2(t_0;k) & u^2_2(t_0;k)u_3(t_0;k) & \cdots & u_2(t_0;k)u_3(t_0;k)u_4(t_0;k) & \cdots &u^3_{1}(t_0;k)\\
  u^3_3(t_0;k) & u^2_3(t_0;k)u_4(t_0;k) & \cdots & u_3(t_0;k)u_4(t_0;k)u_5(t_0;k) & \cdots &u^3_2(t_0;k)\\
\vdots & \vdots & \ddots & \vdots & \ddots & \vdots \\
  u^3_n(t_0;k) & u^2_n(t_0;k)u_1(t_0;k) & \cdots & u_{n}(t_0;k)u_{1}(t_0;k)u_{2}(t_0;k) & \cdots &u^3_{n-1}(t_0;k)\\
\end{pmatrix}.
\label{eqn:cubicmatrix}
\end{equation} 
The process continues this way for any higher-order monomial term. The $n \times N$  dictionary matrix (where $N={n+p \choose p}$ is the number of monomials of degree at most $p$) is given by:
\begin{equation}
A(t_0; k) = \begin{bmatrix} & 1 & | & U(t_0;k) & | & U^2(t_0;k) & | & U^3(t_0;k) & | &  \cdots & \end{bmatrix},
\label{eqn:dictionaryburstmonomial}
\end{equation} 
where one augments the matrix from the right by the additional monomial terms. For simplicity, we will consider the cubic case for all examples and results. Note that when $n=150$, the number of candidate functions $N$ exceeds half a million for the cubic case and over 22 million for the quartic case.

The velocity for the $k$th burst is given by:
\begin{equation*}
V(t_0;k)=\begin{pmatrix}
\dot{u}_1(t_0;k)\\
\dot{u}_2(t_0;k)\\
\dot{u}_3(t_0;k)\\
 \vdots    \\
\dot{u}_n(t_0;k)
\end{pmatrix} .
\end{equation*} 
Let $c$ be the vector of coefficients, $c=\begin{pmatrix}c_1, c_2, \ldots, c_N\end{pmatrix}^T$. If we use multiple bursts, say $k$ from $1,\ldots, K$ and/or multiple snapshots (\textit{i.e.} $m>1$), then we concatenate the matrices and vectors row-wise as follows,  
\begin{equation}
V = \begin{pmatrix}
V(t_0;1) \\ V(t_1;1) \\ \vdots    \\ V(t_{m-1};1) \\
V(t_0;2) \\ V(t_1;2) \\ \vdots    \\ V(t_{m-1};2) \\
\vdots    \\
V(t_0;K) \\ V(t_1;K) \\ \vdots    \\ V(t_{m-1};K) \\
\end{pmatrix}
\hspace{0.3in} \text{and} \hspace{0.3in}
A = \begin{pmatrix}
A(t_0;1) \\ A(t_1;1) \\ \vdots    \\ A(t_{m-1};1) \\
A(t_0;2) \\ A(t_1;2) \\ \vdots    \\ A(t_{m-1};2) \\
\vdots    \\
A(t_0;K) \\ A(t_1;K) \\ \vdots \\ A(t_{m-1};K) \\
\end{pmatrix}.\label{eqn:dictionaryvelocity}
\end{equation} 
Thus, the linear inverse problem is to find $c$ such that $V = A c$. The size of the dictionary matrix $A$ is $mnK \times N$. Therefore, for small $K$ and $m$, this problem will be underdetermined.

We assume that the number of samples is very small. Thus, the initial data provides a large portion of the information in $A$. The burst is (mainly) used to obtain an approximation to the velocity, so that $V$ is relatively accurate.

\subsection{Dictionary for Higher Spatial Dimension}\label{section:higher}

To generalize to higher spatial dimension, the cyclic permutations must be defined for multiple indices. 
Given the vectorization of a two-dimensional array, $w = \text{vec}(W)$, where $W = [W_{i,j}]$ for $1\leq i,j \leq n$, we must permute $w$ with respect to cyclic permutations of the two-dimensional array $W$. A permutation of an array preserves the cyclic structural condition if it is a cyclic permutation of both the rows and the columns. In particular, a cyclic permutation of $W \in \mathbb{R}^{n \times n}$ is equivalent to sending each element $W_{i,j}$ to $W_{\gamma(i),\tau(j)}$ for $\gamma, \tau \in \mathcal{C}_n$. In order to combine the $n^2$-permuted arrays, each permutation is vectorized and stored row-wise:
\begin{equation}
U(t_0;k)= [ \text{vec}(u_{\gamma(i),\tau(j)})(t_0;k)].
\label{eqn:2ddatamatrix}
\end{equation}
As an example, consider $u(t_0;k) \in \mathbb{R}^{3 \times 3}$, where:
\begin{equation*}
u(t_0;k)=\begin{pmatrix}[ccc]
 u_{1,1}(t_0;k) & u_{1,2}(t_0;k) &  u_{1,3}(t_0;k) \\
u_{2,1}(t_0;k) & u_{2,2}(t_0;k) &  u_{2,3}(t_0;k) \\
 u_{3,1}(t_0;k) & u_{3,2}(t_0;k) &  u_{3,3}(t_0;k)\\
\end{pmatrix}.
\label{eqn:2ddat}
\end{equation*}
One cyclic permutation of $u(t_0;k)$ is to take rows $\{ 1,2,3\}$ to $\{ 2,3,1\}$  and columns $\{ 1,2,3\}$ to $\{ 3,1,2\}$,
\begin{equation*}
\widetilde{u}(t_0;k)=\begin{pmatrix}[ccc]
 u_{2,3}(t_0;k) & u_{2,1}(t_0;k) &  u_{2,2}(t_0;k) \\
u_{3,3}(t_0;k) & u_{3,1}(t_0;k) &  u_{3,2}(t_0;k) \\
 u_{1,3}(t_0;k) & u_{1,1}(t_0;k) &  u_{1,2}(t_0;k)\\
\end{pmatrix}.
\label{eqn:permuted2ddat}
\end{equation*}

This construction has the additional benefit of not repeating elements in each row or column. The higher-order monomials and the corresponding dictionary and velocity matrices ($A$ and $V$, respectively, defined by Equation~\eqref{eqn:dictionaryvelocity}) are built as before using Equation~\eqref{eqn:2ddatamatrix} as the input data. As we increase the spatial dimension, the number of candidate functions grows; for example for  $n=15$, the number of cubic candidate functions in two spatial dimensions is nearly two million!

For a general spatial dimension $n$, the process above is repeated, where one constructs all permutations of the $n$-dimensional array $u(t_0;k)$ by applying cyclic permutations to each of the coordinates separately. Each of the permuted $n$-dimensional arrays are vectorized and collected (row-wise) as is done in Equation~\eqref{eqn:2ddatamatrix}. The dictionary and velocity matrices are constructed as above.

\subsection{Restriction of the Data and Localization of the Dictionary Matrix}\label{section:local}

The dictionary matrix construction in the previous sections relies on the cyclic permutation of the input. One may restrict the learning algorithm to a subset of the data and also localize the basis to a patch in the domain. This is advantageous, for example, when only a subset of the data is known to be accurate enough to approximate the velocity or when the initial data is only sufficiently random in a subset of the domain.

The process of restricting the data and localizing the basis are technically distinct. The restriction to a subdomain will always be slightly larger than the localization of the basis terms. To illustrate, consider the one-dimensional system with $n>9$ points. We localize the basis by assuming that the equation for the $j$th component, say $u_j$, only depends on monomial terms $u_i$ for $i\in[j-2,j+2]$. Therefore, the data matrix $U(t_0;k)$ defined by Equation~\eqref{eqn:datamatrix} becomes a five point approximation:
\begin{equation*}
U(t_0;k)|_{5-pnts}=\begin{pmatrix}[ccccc]
u_1(t_0;k) & u_2(t_0;k) & u_3(t_0;k) & u_{n-1}(t_0;k) & u_n(t_0;k)\\
u_2(t_0;k) & u_3(t_0;k) & u_4(t_0;k)  & u_{n}(t_0;k) & u_1(t_0;k)\\
u_3(t_0;k) & u_4(t_0;k) & u_5(t_0;k)  & u_{1}(t_0;k) & u_2(t_0;k)\\
& &  \vdots &  \\
u_{n-1}(t_0;k) & u_n(t_0;k) & u_1(t_0;k)  & u_{n-3}(t_0;k) & u_{n-2}(t_0;k)\\
u_n(t_0;k) & u_1(t_0;k) & u_2(t_0;k)  & u_{n-2}(t_0;k) & u_{n-1}(t_0;k)\\
\end{pmatrix}.
\label{eqn:datamatrixlocal}
\end{equation*}
Note that $U(t_0;k)|_{5-pnts}$ is of size $n\times 5$. The first two and last two rows assume that the data is periodic, since information crosses the boundary between indices $n$ and $1$. Next, the restriction of the data onto a subdomain is done by removing all rows that include points outside of the subdomain. For example, the restriction onto the subdomain indexed by $\{3,4,5,6,7\}$ yields:
\begin{equation}
U(t_0;k)|_{5-pnts, restricted}=\begin{pmatrix}[ccccc]
u_3(t_0;k) & u_4(t_0;k) & u_5(t_0;k)  & u_{1}(t_0;k) & u_2(t_0;k)\\
u_4(t_0;k) & u_5(t_0;k) & u_6(t_0;k)  & u_{2}(t_0;k) & u_3(t_0;k)\\
u_5(t_0;k) & u_6(t_0;k) & u_7(t_0;k)  & u_{3}(t_0;k) & u_4(t_0;k)\\
u_6(t_0;k) & u_7(t_0;k) & u_8(t_0;k)  & u_{4}(t_0;k) & u_5(t_0;k)\\
u_7(t_0;k) & u_8(t_0;k) & u_9(t_0;k)  & u_{5}(t_0;k) & u_6(t_0;k)\\
\end{pmatrix},
\label{eqn:datamatrixres}
\end{equation}
which reduces the matrix to size $5\times 5$ -- the loss of additional rows are required so that all cyclic permutations remain within the domain. It is important to note that the localized and restricted data matrix \textit{no longer requires periodic data as long we are sufficiently away from the boundary}. The localized and restricted dictionary matrix is built by repeating the process in Equations~\eqref{eqn:quadmatrix}-\eqref{eqn:dictionaryburstmonomial}, but using the localized and restricted data matrix described above (see Equation~\eqref{eqn:datamatrixres}).

Localizing the dictionary elements provide additional benefits. For many dynamical systems, information at a particular spatial point (or index) only interacts with information at its neighboring points (for example, all neighbors within a prescribed distance). Thus, localization may remove unnecessary complexities in the dictionary. The second is that the number of unknowns is severely reduced when considering a subset of the candidate functions. This was observed in \cite{schaeffer2017extracting} where localization reduced the inverse problem to a smaller (but still under-sampled) system and makes the sampling rate nearly independent of the ambient dimension $n$. Lastly, the accuracy of the approximation to the time derivative controls the error bound in our recovery problem. Thus, if the dynamics are only accurate in a small region, it is better to restrict the learning to that region. More data is usually beneficial; however, adding noisy and inaccurate measurements does not increase the likelihood of recovering the correct governing model.

\subsection{Bounded Orthogonal Dictionary}\label{section:bod}

The recovery of the coefficient vector $c$ from data $V$ is better conditioned with respect to a dictionary built from bounded orthogonal terms. For simplicity, we will detail this construction for data $u\in \mathbb{R}^n$, \textit{i.e.} one spatial dimension with $n$-nodes. Consider a subset of the domain, ${\cal D} \subset \mathbb{R}^n$, endowed with a probability measure $\mu$.  Suppose that $\{ \phi_1, \phi_2, \dots, \phi_N \}$ is a (possibly complex-valued) orthonormal system  on ${\cal D}$,
\begin{align*}
\int_{{\cal D}} \phi_j(u) \overline{ \phi_k(u)} d\mu(u) = \delta_{j,k} &=  \left\{ \begin{array}{ll} 0 & \text{ if}\ j \neq k \\ 1 & \text{ if}\ j=k \\ \end{array} \right\}.
\end{align*}
The collection $\{ \phi_1, \phi_2, \dots, \phi_N \}$ is called a \emph{bounded orthonormal system} with constant $K_b \geq 1$ if:
\begin{equation}
\label{eq:BOS}
\| \phi_j \|_{\infty} := \sup_{u \in {\cal D}} | \phi_j(u) | \leq K_b \quad \text{for all}\ j \in [N].
\end{equation}
Suppose that $u^{(1)}, u^{(2)}, \dots, u^{(m)} \in {\cal D}$ are sampling points which are drawn i.i.d. with respect to the orthogonalization measure $\mu$, and consider the sampling matrix: 
\begin{align*}
A_{\ell,k} = \phi_k(u^{(\ell)}), \quad \ell \in [m],\ k \in [N].
\end{align*}
An important example of a bounded orthonormal system is the Legendre polynomials. In high dimensional systems, we will use the tensorized Legendre basis in place of their corresponding monomials. We denote $A_L$ the dictionary matrix corresponding to the tensorized Legendre basis. For example, if we consider the initial data samples $u(t_0)$ drawn i.i.d. from the uniform distribution $[-1,1]^n$, then the Legendre polynomials (orthogonalization with respect to $\dd{\mu}=\frac{1}{2} \dd{x}$) up to degree three are:
\begin{align*}1,\ \ \sqrt{3} u_i, \ \ \frac{\sqrt{5}}{2} (3 u_i^2-1), \ \ 3u_i u_j, \ \ \frac{\sqrt{7}}{2} (5u_i^3-3u_i), \  \ \frac{\sqrt{15}}{2} (3u_i^2-1)u_j,  \ \ \sqrt{27} u_i u_j u_k.
\end{align*}

If a function is $s$-sparse with respect to the standard quadratic basis, it will be $(s+1)$-sparse with respect to the Legendre basis, since the quadratic Legendre term, $\frac{\sqrt{5}}{2} (3 u_i^2-1)$, can add at most a constant to the representation. If a function is $s$-sparse with respect to the standard cubic basis, it will be $(2s)$-sparse with respect to the Legendre basis, since the terms  $\frac{\sqrt{7}}{2} (5u_i^3-3u_i)$ and $\frac{\sqrt{15}}{2} (3u_i^2-1)u_j$ each add an additional $s$ terms (in the worst-case scenario). We assume that $s$ is sufficiently small so that, for example, a $(2s)$-sparse system is still relatively sparse.

For the examples presented here, we focus on dynamical systems with (at most) cubic nonlinearity. The procedure above is not limited to this case. In fact, generalizing this construction to systems which are sparse with respect to alternative bounded orthogonal system is fairly direct. With high probability, a random matrix formed from bounded orthogonal terms will lead to a well-conditioned inverse problem $V= A_Lc$ if $c$ is sufficiently sparse (see  \cite{candes2011probabilistic, foucart2013mathematical}).

\section{Sparse Optimization and Recovery Guarantee}\label{sec:theory}

Let $A_L$ be the dictionary in the Legendre basis up to third order. The size of $A_L$ is $mnK \times N$, where $N$ is the number of basis terms. The linear system $V= A_Lc$ is underdetermined if we assume that $m$ and $K$ are small and fixed. 
To ``invert" this system, we impose that the vector of coefficients $c$ is sparse, \textit{i.e.,} $c$ has only a few non-zero elements. This can be written formally as a non-convex optimization problem:
\begin{equation*}
\min_c \ ||c||_{0}  \ \ \  \text{subject to} \ \ A_Lc=V,
\end{equation*}
where $||c||_0=\text{card}(\supp(c))$ is the $\ell^0$ penalty which measures the number of non-zero elements of $c$. In practice, the constraint is not exact since $V$ is computed and contains some errors. The noise-robust problem is:
\begin{equation*}
\min_c \ ||c||_{0}  \ \ \  \text{subject to} \  ||A_Lc-V||_2 \leq \sigma,
\end{equation*}
where $\sigma>0$ is a noise parameter determined by the user or approximated from the data. The general noise-robust $\ell^0$ problem is known to be NP hard \cite{foucart2013mathematical}, and is thus relaxed to the $\ell^1$ regularized, noise-robust problem:
\begin{equation}
\min_{c' } \ ||c'||_{1}  \ \ \  \text{subject to} \ \ \|A_Lc'-V\|_2 \leq \sigma, \tag{L-BP$_\sigma$}
 \label{eq:lbp}
 \end{equation}
which we refer to as the Legendre basis pursuit (L-BP) (for a general matrix $A$ this is known as the $\ell^1$ basis pursuit). Note that $c'$ is the coefficients in the Legendre basis and $c$ is the coefficients in the standard monomial basis. If the system is sufficiently sparse with respect to the standard monomial basis representation, then it will be sparse with respect to the Legendre basis representation, and thus the formulation is consistent. The parameter $\sigma$ is independent of the basis we use in the dictionary matrix. In the ideal case, $\sigma$ is the $\ell^2$ error between the computed velocity and true velocity. In practice, it must be estimated from the trajectories.

\subsection{Recovery Guarantee and Error Bounds}\label{sec:guarantee}

To guarantee the recovery of sparse solution to the underdetermined linear inverse problem, we use several results from random matrix theory. In general, it is difficult to recover $c\in \mathbb{R}^N$ from $Ac=V$, when $V\in\mathbb{R}^{\widetilde{m}}$, $A \in \mathbb{R}^{\widetilde{m} \times N}$, and  $\widetilde{m} \ll N$. In our setting, we know that the system is well-approximated by an $s$-term polynomial (for small $s$), and thus the size of the support set of $c$ is relatively small. However, the locations of the nonzero elements (the indices of the support set) are unknown. If the matrix $A$ is incoherent and $\widetilde{m}\sim s \log(N)$, then the recovery of the sparse vector $c$ from the $\ell^1$ basis pursuit problem is possible. In particular, by leveraging the sparsity of the solution $c$ and the structure of $A$, compressive sensing is able to overcome the curse of dimensionality by requiring far fewer samples than the ambient dimension of the problem. This approach also yields tractable methods for high-dimensional problems.

We provide a theoretical result on the exact and stable recovery of high-dimensional orthogonal polynomial systems with the cyclic condition via a probabilistic bound on the coherence of the dictionary matrix.  

\begin{theorem}\label{thrm:31}
If $A_{j_1}$ and $A_{j_2}$ are two columns from the cyclic Legendre sampling matrix of order $p$ generated by a vector $u\in \mathbb{R}^n$ with i.i.d. uniformly distributed entries in $[-1,1]$ and $2p^2 \leq n$, then with probability exceeding $1 - \left(\frac{e}{p} + \frac{e}{2p^2}\right)^{2p} \, n^{-2p/11}$, the following holds:
\begin{enumerate}
  \item $| \left< A_{j_1},  A_{j_2} \right>| \leq 12 \,p^3 \,3^{p} \, \sqrt{n\,  \log n}$ for all $j_1 \neq j_2 \in \{1,2, \dots, n\}$,
  \item  $| \| A_{j_1} \|^2  - n | \leq 12\, p^3 \,3^{p}\,  \sqrt{n\, \log n}$ for all $j_1 \in \{1,2, \dots, n\}$. 
  \end{enumerate}
 
\end{theorem}

\begin{proof}
Given a vector $u = (u_1, \dots, u_n)\in \mathbb{R}^n$ with i.i.d. uniformly distributed entries in $[-1,1]$, let $A \in \mathbb{R}^{n\times N}$, with $N= {{n+p}\choose{p}}$, be the associated Legendre sampling matrix of order $p$, that is the matrix formed by transforming the matrix in Equation~\eqref{eqn:dictionaryburstmonomial} with $k=1$ to the Legendre system. In particular, the matrix is defined as:
\begin{equation*}
A: = \begin{bmatrix} & U^0_L & | & U^1_L & | & U^2_L  & | & \cdots & | & U^p_L & \end{bmatrix},
\end{equation*} 
where $U_L^q$ is a matrix generated from the tensorized Legendre basis of order $q$ for $0\leq q\leq p$. For examples, $U_L^0\in \mathbb{R}^{n}$ is a vector of all ones,
\begin{equation*}
U_L^1:=\begin{pmatrix}[cccc]
 \sqrt{3}\, u_1 &  \sqrt{3}\,u_2 & \cdots &  \sqrt{3}\,u_n \\
 \sqrt{3}\,u_2 &  \sqrt{3}\,u_3 & \cdots &  \sqrt{3}\,u_1 \\
\vdots & \vdots &  & \vdots \\
  \sqrt{3}\,u_n &  \sqrt{3}\,u_1 & \cdots &  \sqrt{3}\,u_{n-1}\\
\end{pmatrix},
\end{equation*}
and 
\begin{equation*}
U_L^2:= \begin{pmatrix}[cccc]
 \frac{\sqrt{5}}{2} \,(3 u_1^2-1) & 3\,u_1\, u_2 & \cdots &\frac{\sqrt{5}}{2}\, (3 u_n^2-1)\\
 \frac{\sqrt{5}}{2} \,(3 u_2^2-1) & 3\,u_2\, u_3 & \cdots &\frac{\sqrt{5}}{2} \,(3 u_1^2-1)\\
\vdots & \vdots &  & \vdots \\
\frac{\sqrt{5}}{2} \,(3 u_n^2-1)& 3\,u_n\, u_1 & \cdots &\frac{\sqrt{5}}{2} \,(3 u_{n-1}^2-1)\\
\end{pmatrix}.
\end{equation*}
Consider the random variable $Y_{j_1,j_2} = \left< A_{\cdot,j_1},  A_{\cdot,j_2} \right>$  which is the inner product between the columns $j_1$ and $j_2$ of $A$, where $A=[A_{i,j}]$ for $ 1\leq  i\leq n$ and $ 1\leq  j\leq N$. Denote the components of the sum by:
$$Y_i := A_{i,j_1} \,  A_{i,j_2},$$ 
so that we can write the inner product as:
\begin{equation}
Y_{j_1,j_2}= \left< A_{\cdot,j_1},  A_{\cdot,j_2} \right> = \sum\limits_{i=1}^n \, A_{i,j_1} \,  A_{i,j_2} = \sum\limits_{i=1}^n Y_{i}.
\end{equation}

The components $Y_i$ have several useful properties. The components of $Y_{j_1,j_2}$ are uncorrelated, in particular, they satisfy $\mathbb{E}[Y_{i}] = 0$ if $j_1 \neq j_2$ and  $\mathbb{E}[Y_{i}] = 1$ if $j_1 \neq j_2$,  when one normalizes the columns. For fixed $j_1$ and $j_2$, the elements $Y_i = A_{i,j_1} \,  A_{i,j_2}$ follow the same distribution for all $1\leq  i \leq n$. This is a consequence of the cyclic structure of $A$, since each product $A_{i,j_1} \,  A_{i,j_2}$ has the same functional form applied to different permutations of the data $u$. 

Note that the $L^2(d\mu)$-normalized Legendre system of order $p$ (\textit{i.e.} the tensor product of univariate Legendre polynomials up to order $p$) is a bounded orthonormal system with respect to $d\mu = \frac{1}{2}dx$. In particular, each basis term is bounded in $L^\infty([-1,1]^n)$ by $K_b = 3^{p/2}$ (which can be achieved by  $3^{p/2} u_{i_1} \cdots u_{i_p} $ at the boundary of the domain). Therefore, $| Y_{i} | \leq K_b^2 = 3^p$, $\abs{Y_i - \mathbb{E}[Y_{i}]} = \abs{Y_i} \leq 3^p$, and $\Var(Y_i)\leq \mathbb{E}(Y_i^2) \leq 9^p$.

Applying Theorem~\ref{thrm:dependent} from the Appendix, which is a rephrasing of Theorem 2.5 from \cite{janson2004large}, with $\triangle = 4p^2$, $M = 3^p$, and $\text{Var}(Y_{i}) \leq 9^p$, yields the following bound:  
\begin{align}
P\left( \abs{Y_{j_1,j_2}  -  \mathbb{E} Y_{j_1,j_2}} \geq \tau\right) \leq 2\exp\left( - \frac{\tau^2 ( 1- p^2/n)}{8p^2 (9^p\, n  + 3^{p-1}\, \tau) } \right).
\end{align}
By assumption we have $\dfrac{p^2}{n}\leq \dfrac{1}{2}$, which happens, for example, when the maximal degree $p$ is small and the ambient dimension $n$ is much larger than $p$. By setting $\tau = 12\,p^3\,  3^{p}\, \sqrt{\,n\, \log n}$ and using $(1-\frac{p^2}{n} )\geq \frac{1}{2}$ and $\log n\leq n$, we have:
\begin{align}
P\left( |Y_{j_1,j_2}-  \mathbb{E} Y_{j_1,j_2}| \geq 12\,p^3\,  3^{p}\, \sqrt{n\, \log n}\right) & \leq 2\exp\left( - \dfrac{\tau^2}{16p^2 (9^p\, n  + 3^{p-1}\, \tau)} \right) \nonumber \\
&\leq 2\exp\left(-\dfrac{9 p^4 \, n\, \log n}{n+4 \, p^3\,\sqrt{n\, \log n}}\right) \nonumber \\
& \leq 2\exp\left(-\dfrac{9 p^4 n\log n}{n+4\, p^3 n}\right) \nonumber \\
& \leq 2\exp\left(-\dfrac{9 p^4}{1+4\, p^3} \, \log n\right).
\label{eq:inequalitydepend} 
\end{align}
Equation~\eqref{eq:inequalitydepend} holds for all pairs $(j_1,j_2)$, therefore taking a union bound over all $N(N-1)/2$ pairs and using the inequality: 
$$N = {{n+p}\choose p} \leq \left(\dfrac{e(n+p)}{p}\right)^p\leq  \left(n \left(\frac{e}{p} + \frac{e}{2p^2}\right)\right)^p=  n^p \left(\frac{e}{p} + \frac{e}{2p^2}\right)^p,$$
for $p\geq1$ where $e = \exp(1)$, we have:
\begin{align*}
P\left( \exists (j_1,j_2): |Y_{j_1,j_2} - \mathbb{E} Y_{j_1,j_2}| \geq 12\,p^3\,  3^{p}\, \sqrt{n\, \log n} \right) 
&\leq N^2 \exp\left(-\dfrac{9 p^4}{1+4\, p^3} \, \log n\right)  \\
&\leq n^{2p}  \left(\frac{e}{p} + \frac{e}{2p^2}\right)^{2p} \exp\left(-\dfrac{9p^4}{1+4\, p^3} \, \log n\right)\\
&\leq\left(\frac{e}{p} + \frac{e}{2p^2}\right)^{2p}  \, \exp\left(\left(2p-\dfrac{9p^4}{1+4\, p^3} \right) \log n\right)\\
&\leq \left(\frac{e}{p} + \frac{e}{2p^2}\right)^{2p} \, \exp\left(-\dfrac{p^4-2p}{4p^3+1}\log n\right)\\
&\leq \left(\frac{e}{p} + \frac{e}{2p^2}\right)^{2p}  \, n^{-2p/11},\quad \text{for}\quad p\geq 2.
\end{align*}

\end{proof}

Theorem \ref{thrm:31} provides an estimate on the coherence of the sampling matrix. We recall the coherence-based sparse recovery result from \cite{gribonval2003sparse, donoho2003optimally,foucart2013mathematical} below.

\begin{theorem} [Related to Theorem 5.15 from \cite{foucart2013mathematical}]\label{thrm:coherence}
Let $A$ be an $m \times N$ matrix with $\ell_2$-normalized columns.  If: 
\begin{equation}
\label{coherence_condition}
\max_{j \neq k} | \langle{ A_j,  A_k \rangle}| < \frac{1}{2s - 1},
\end{equation}
then for any $s$-sparse vector $c \in \mathbb{C}^N$ satisfying $v=Ac+e$ with $||e||_2 \leq \sigma$, a minimizer $c^\#$ of L-BP$_\sigma$ approximates the vector $c$ with the error bound:
$$\|c -c^\# \|_1 \leq d \, s\, \sigma,$$
where $d>0$ is a universal constant.
\end{theorem}

Using Theorem~\ref{thrm:coherence}, we can show the exact recovery for the case were $A$ is a cyclic Legendre sampling matrix of order $p$.

\begin{theorem}  \label{thrm:samplingcyclic} 
Let $A\in \mathbb{R}^{n \times N}$ be the Legendre sampling matrix of order $p$ generated by a vector $u\in \mathbb{R}^n$ with i.i.d. uniformly distributed entries in $[-1,1]$, then with probability exceeding \\$1 - \left(\frac{e}{p} + \frac{e}{2p^2}\right)^{2p}  \, n^{-2p/11}$, an $s$-sparse vector $c \in \mathbb{C}^N$ satisfying $v=Ac+e$ with $||e||_2 \leq \sigma$ can be recovered by $c^\#$, the solution of L-BP$_\sigma$, with the error bound:
$$\|c -c^\# \|_1 \leq d \, s\, \sigma,$$
for some universal constant $d>0$ as long as:
$$\frac{n}{\log n} \geq 144 \, p^6 \, 9^p \, s^2.$$
In addition, if $A$ is generated from $K$ samples $u(k)\in \mathbb{R}^n$ with i.i.d. uniformly distributed entries in $[-1,1]$ for $1\leq k\leq K$ and: 
$$K \geq \frac{ 144 \, p^6 \, 9^p \, s^2\, \log n}{n},$$
then with probability exceeding $1 - \left(\frac{e}{p} + \frac{e}{2p^2}\right)^{2p} \, n^{-2p/11}$, an $s$-sparse vector $c \in \mathbb{C}^N$ satisfying $v=Ac+e$ with $||e||_2 \leq \sigma$ can be recovered by $c^\#$, the solution of L-BP$_\sigma$, with the error bound:
$$\|c -c^\# \|_1 \leq d \, s\, \sigma.$$
\end{theorem}

\begin{proof}
The normalized matrix can be written as:
$$\bar{A} = A  D,$$
where $D$ is a diagonal matrix with $D_j={\| A_{\cdot,j} \|_2^2}$, \textit{i.e.} the diagonal contains the squared norm of the columns of $A$.  Then $\bar{A}$ has $\ell_2$-normalized columns, and Equation~\eqref{coherence_condition} is satisfied with probability exceeding $1 - \left(\frac{e}{p} + \frac{e}{2p^2}\right)^{2p} \, n^{-2p/11}$ as long as:

$$
\frac{n}{\log n} \geq 144 \,p^6 \,9^p\, s^2.
$$
Thus by Theorem~\ref{thrm:coherence}, we have the corresponding $\ell^1$ error bound. The extension of this result to multiple samples $u(k) = (u(k)_1, \dots, u(k)_n)$ for $1\leq k\leq K$, follows directly as long as:
$$K \geq \frac{ 144\, p^6 \,9^p \,s^2 \,\log n}{n}.$$
\end{proof}

The results in Theorem~\ref{thrm:samplingcyclic} are important on their own. In particular, the theorem shows that for large enough dimension, one can recovery cyclic polynomial systems from only a few samples. 

\medskip

Returning to the problem of model selection for structured dynamical systems, we can apply Theorem~\ref{thrm:samplingcyclic} to obtain the following recovery guarantee.

\begin{theorem}
\label{thm:main}
Let $\left\{u(t_0;k), \dots, u(t_{m-1}; k)\right\}$ and $\left\{\dot{u}(t_0; k), \dots, \dot{u}(t_{m-1};k)\right\}$ be the state-space and velocity measurements, respectively, for $1\leq k\leq K$ bursts of the $n$-dimensional evolution equation $\dot{u}=f(u)$.   Assume that the components, $f_j$, satisfy the cyclic structural condition and that they have at most $s$ non-zero coefficients with respect to the Legendre basis.  Assume that for each $k$, the initial data $u(t_0;k)$ is randomly sampled from the uniform distribution in  $[-1,1]^n$ (thus each component of the initial vector are  i.i.d.). 
Also, assume that the total number of bursts, $K$, satisfies:
\begin{align}
\label{eq:numburst}
K \geq \frac{ 144\, p^6 \,9^p \,s^2 \,\log n}{n}.
\end{align}
 Then with probability exceeding $1 - \left(\frac{e}{p} + \frac{e}{2p^2}\right)^{2p}  \, n^{-2p/11}$, the vector $c$ can be recovered exactly by the unique solution to Problem \eqref{eq:lbp}.
\bigskip
\noindent In addition, under the same assumptions as above, if the time derivative is approximated within $\eta$-accuracy in the scaled $\ell^2$ norm, \textit{i.e.} if $\widetilde{\dot{u}}(t_0;k)$ is the approximation to the time derivative and:
\begin{align*}
\sqrt{\frac{1}{K}\sum\limits_{k=1}^K \left|\widetilde{\dot{u}}(t_0;k)-{\dot{u}}(t_0;k)\right|^2} \leq \eta,
\end{align*}
then by setting $\sigma = \sqrt{K} \eta$ and using the submatrix of $A_L$ consisting of only the initial data, any particular vector $c$ is approximated by a minimizer $c^{\#}$ of  Problem \eqref{eq:lbp} with the following error bound:
\begin{align}
\|c -c^\# \|_1 \leq d \, s\, \sigma, \label{eqn:errorbound}
\end{align}
where $d$ is a universal constant.
\end{theorem}

Theorem~\ref{thm:main} provides an $\ell^1$ error bound between the learned coefficients and the true sparse coefficients. If the nonzero elements of $c^\#$ are sufficiently large with respect to the error bound, then the support set containing the $s$-largest coefficients coincides with the true support set.

\begin{proposition}
\label{prop2}
Assume that the conditions of Theorem~\ref{thm:main}  hold. Let $S$ be the support set of the true coefficients $c$, \textit{i.e.} $S:=\supp(c)$, and let $S^\#$ be the support set of the $s$-largest (in magnitude) of $c^\#$, a minimizer of Problem \eqref{eq:lbp}. If 
\begin{align}
\sigma < \frac{\min\limits_{j\in S} |c_j|}{2\, d \, s}, \label{eq:prop2condi}
\end{align} 
where $d$ is the same universal constant as in Equation~\eqref{eqn:errorbound}, then $S^\#=S$. 
\end{proposition}

\begin{proof}
This proposition is a consequence of the recovery bound in Equation \eqref{eqn:errorbound}:
\begin{align*}
\| c^{\#} - c \|_1   \leq d \, s\, \sigma.
\end{align*} 
By assumption, $\sigma$ satisfies Equation \eqref{eq:prop2condi},
then the maximum difference between the true and approximate coefficients is:
\begin{align*}
\max_j | c_j-c_j^\#| \leq \| c-c^\# \|_1 \leq \,  d \, s\, \sigma < \, \frac{1}{2} \, \min\limits_{j\in S} |c_j|.
\end{align*} 
Thus, for any $j \in S$, we have $|c_j^\#|>\, \frac{1}{2} \, \min\limits_{j\in S} |c_j|$, and  for any $j \in S^c$, we have $|c_j^\#|\le \, \frac{1}{2} \, \min\limits_{j\in S} |c_j|$. Therefore, $S^\#$ corresponds to the support set of $|c_j^\#|>\, \frac{1}{2} \, \min\limits_{j\in S} |c_j|$, which is identically $S$.
\end{proof}

Proposition~\ref{prop2} provides validation for post-processing the coefficients of Problem \eqref{eq:lbp}, in particular, if the noise is small enough, we could remove all but the $s$ largest (in magnitude) coefficients in $c^{\#}$.

It is worth noting that it is possible to recover the system from one time-step. This is more probable as the dimension $n$ of the problem grows. The sampling bound improves as $n$ grows, since for large $n$, we have $n\gg s^2\, \log n$. Thus, for large enough $n$, one random sample is sufficient. Furthermore, if $s^2\ll n$, we can recover the system from only one time step and from only a subset $\widetilde{n}<n$ of the coordinate equation, where $\widetilde{n}\sim s^2$. Therefore, one just needs to have $\widetilde{n}$ accurate estimations of velocities. 

Theorem~\ref{thm:main} also highlights an important aspect of the scaling. Without any additional assumptions, one is limited to lower-order polynomials, since the numbers of samples required may be too large (since $K_b$ grows rapidly). However, with additional assumptions, for example the cyclic structural condition, the recovery becomes nearly dimension-free, which as a side-effect, allows for higher-order polynomials more easily.

 Note that if the initial data follows another random distribution, then one can construct the corresponding orthogonal polynomial basis. For example, we could assume that the initial data has i.i.d. components sampled from the Chebyshev measure on $[-1,1]^n$ or an interpolating measure between the uniform measure and the Chebyshev measure \cite{rauhut2012sparse}.

\section{Numerical Method}\label{sec:numerics}

The constrained optimization problem \eqref{eq:lbp} can be solved using the Douglas-Rachford algorithm \cite{lions1979splitting, combettes2011proximal}. To do so, we first define the auxiliary variable $w$ with the constraints:
\begin{align*}
(w,c)\in \mathcal{K}:=\left\{ (w,c) | \ w=Ac \right\} \quad  \text{and} \quad w\in B_\sigma(V) := \{ w \, | \ \|w-V\|_2 \leq \sigma \},
\end{align*}
Equation~\eqref{eq:lbp} can be rewritten as an unconstrained minimization problem:
\begin{align}
&\min_{(w,c)}\  F_1(w,c)+F_2(w,c), \label{eqn:DRobjectivefunction}
\end{align}
where the auxiliary functions $F_1$ and $F_2$ are defined as:
\begin{align*}
 F_1(w,c) :=  \| c \|_{1} + \mathbb{I}_{B_\sigma(V)}(w), \quad \text{and} \quad  F_2(w,c) := \mathbb{I}_{\mathcal{K}}(w,c).
  \end{align*}
Here $\mathbb{I}_\mathcal{S}$ denotes the indicator function over a set $\mathcal{S}$, \textit{i.e.},
\begin{align*}
\mathbb{I}_{\mathcal{S}}(w):=
\begin{cases}
&\ \, 0, \ \ \text{if} \ \ w \in \mathcal{S},\\
 & \infty, \ \ \text{if} \ \ w \notin \mathcal{S}.
 \end{cases}
  \end{align*}
The utility of writing the optimization problem in this form is that both auxiliary functions have a simple and explicit proximal operators, which will be used in the iterative method. The proximal operator for a function $F(x)$ is defined as:
\begin{align*}
 \prox_{\gamma F}(x) :=  \argmin{y}  \left\{\frac{1}{2} \| x-y\|^2 + \gamma \, F(y)\right\},
 \end{align*}
where $\gamma>0$ (to be specified later). The proximal operator of $F_1(w,c)$ is:
\begin{align*}
 \prox_{\gamma F_1}(w,c) &=  \argmin{(y,d)} \left\{ \frac{1}{2} \| w-y\|^2  + \frac{1}{2} \| c-d\|^2+\gamma \| d \|_{1} + \gamma \mathbb{I}_{B_\sigma(V)}(w) \right\} \\
 &=  \left( \argmin{y} \left\{ \frac{1}{2} \| w-y\|^2+\mathbb{I}_{B_\sigma(V)}(w) \right\},\  \argmin{d} \left\{ \frac{1}{2} \| c-d\|^2+\gamma \| d \|_{1} \right\} \right)\\
 &= \left(  \text{proj}_{B_\sigma(V)}(w), \ S_\gamma(c) \right),
\end{align*}
where the projection onto the ball can be computed by:
\begin{align*}
\text{proj}_{B_\sigma(V)}(w) :=
\begin{cases}
 \hspace{2.25cm} w, \ & \text{if} \ \ w \in B_\sigma(V),\\
  V +  \sigma\, \dfrac{w-V}{\|w-V\|_2}, \ & \text{if} \ \ w \notin B_\sigma(V).
 \end{cases} 
 \end{align*}
 and the soft-thresholding function $S$ with parameter $\gamma$ is defined (component-wise) as:
 \begin{align*}
[S_\gamma(c) ]_{j} =
\begin{cases}
 & c_j -  \gamma\, \dfrac{c_j}{|c_j|}, \ \ \text{if} \ \ |c_j|>\gamma,\\
 &\hspace{1.4cm} 0, \ \ \text{if} \ \ |c_j| \leq \gamma.
 \end{cases} 
 \end{align*}
Similarly, the proximal operator for $F_2$ is:
\begin{align*}
 \prox_{\gamma F_2}(w,c) &=  \argmin{(y,d)} \left\{ \frac{1}{2} \| w-y\|^2  + \frac{1}{2} \| c-d\|^2+\mathbb{I}_{\mathcal{K}}(w,c) \right\} \\
 &= \left( \  A(I + A^TA)^{-1}(c+A^T w), \ (I + A^TA)^{-1}(c+A^T w)\ \right).
\end{align*}
To implement the proximal operator for $F_2$, we pre-compute the Cholesky factorization $(I + A^TA)=LL^T$ and use forward and back substitution to compute the inverse at each iteration. This lowers the computational cost of each of the iterations. The iteration step for the Douglas-Rachford method is:
\begin{equation}
\begin{aligned}
(\widetilde{w}^{k+1},\widetilde{c}^{k+1})&=\left(1-\frac{\mu}{2}\right)(\widetilde{w}^{k},\widetilde{c}^{k})+\frac{\mu}{2}\ \rprox_{\gamma F_2} \left(  \rprox_{\gamma F_1}\left(\widetilde{w}^{k},\widetilde{c}^{k}\right) \right),\\
({w}^{k+1},{c}^{k+1})&=\prox_{\gamma F_1} (\widetilde{w}^{k+1},\widetilde{c}^{k+1}),
\end{aligned}\label{eq:DRalg}
\end{equation}
where $\rprox_{\gamma F_i}(x) := 2 \prox_{\gamma F_i}(x)-x$ for $i=1,2$. The second step of Equation~\eqref{eq:DRalg} can be computed at the last iteration and does not need to be included within the main iterative loop. The approximation $({w}^{k},{c}^{k})$ converges to the minimizer of Problem~\eqref{eqn:DRobjectivefunction} for any $\gamma>0$ and $\mu \in [0,2]$.

\begin{algorithm}[t!]
\KwData{Given: $u(t;k)\in \mathbb{R}^n$ for $t=t_0$ and $t=t_1$. The number of bursts $k$ is either equal to 1 or very small. The number of nodes $n$ does not need to be large. }
\KwResult{Coefficients of the governing equation $c\in \mathbb{R}^N$. }

 	\textbf{Step 1}: Construct data matrix $U$ as in Sections~\ref{section:oned}, \ref{section:higher}, and \ref{section:local}.\;
 
 	\textbf{Step 2 (optional)}: Add Gaussian noise to $U$, \textit{i.e.} $U\mapsto U+\eta$, where $\eta \sim \mathcal{N}(0,\text{var})$.
 
  	\textbf{Step 3}: Construct the velocity vector $V$ from using $U$ from the previous step.
  
  	\textbf{Step 4}: Transform $U\mapsto aU+b$ so that each elements is valued in $[-1,1]$.\;
  
  	\textbf{Step 5}: Construct the dictionary matrix $A_L$ using $U$ from Step 4; see Section~\ref{section:bod}.\;
   
     	\textbf{Step 6}: Normalize each column of $A_L$ to have unit $\ell^2$-norm.\;
   
  	\textbf{Step 7}: Apply the Douglas-Rachford algorithm to solve Problem \eqref{eq:lbp}. 
  
\indent Input: Set $\sigma>0$. Compute the Cholesky decomposition of $(I+A_L^TA_L)$. Initialize $\widetilde{w}^{0}$ and $\widetilde{c}^{0}$.\;
  
\indent   \While{the sequence $\{\widetilde{c}^{k}\}$ does not converge}{
 {
\indent  Update $\widetilde{w}^{k+1}$ and $\widetilde{c}^{k+1}$ based on Equation~\eqref{eq:DRalg}.     }
 }
\indent  Output: $c_L:=c^k$\;
 
     \textbf{Step 8}: Map the coefficients $c_L$ obtained from Step 5 to the coefficients with respect to the standard monomial basis on the original $U$ as constructed in Step 1. 
       
      \textbf{Step 9 (optional)}: The coefficients can be ``debiased" by solving the system $A|_S \widetilde{c}=V$, where $A|_S$ is the submatrix of $A$ consisting of columns of $A$ indexed by $S:=\supp(c)$ (from Step 6; see also Proposition~\ref{prop2}).   
 \caption{Learning Sparse Dynamics}
 \label{algorithm}
\end{algorithm}

\begin{figure}[b!]
\centering
\subfigure[$u$ at $t_0$]{
	\includegraphics[width = 2 in]{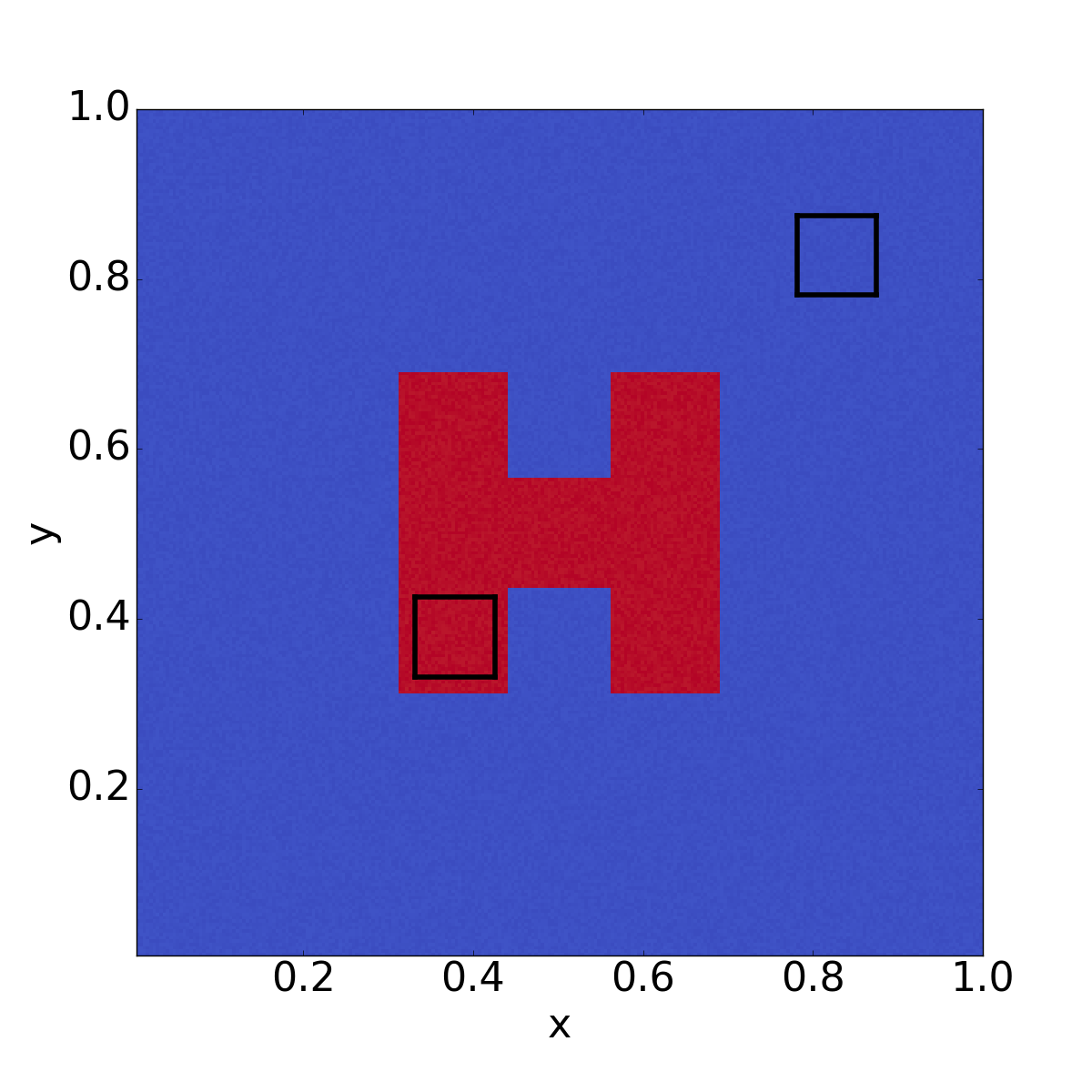}
	}
\subfigure[$u$ at $t_1$]{
	\includegraphics[width = 2 in]{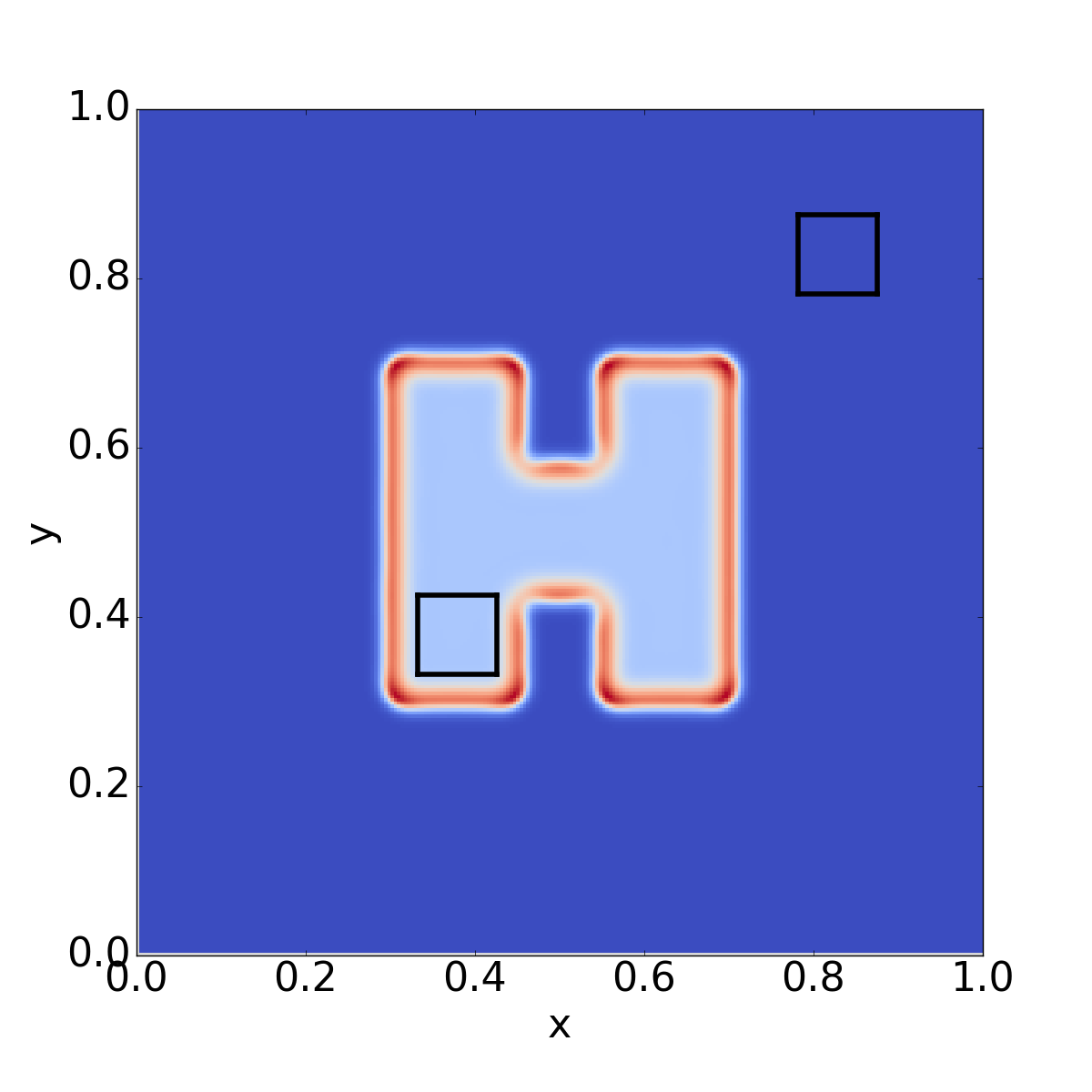}
	}\\
\subfigure[Sub-block 1 at $t_0$]{
	\includegraphics[width = 2 in]{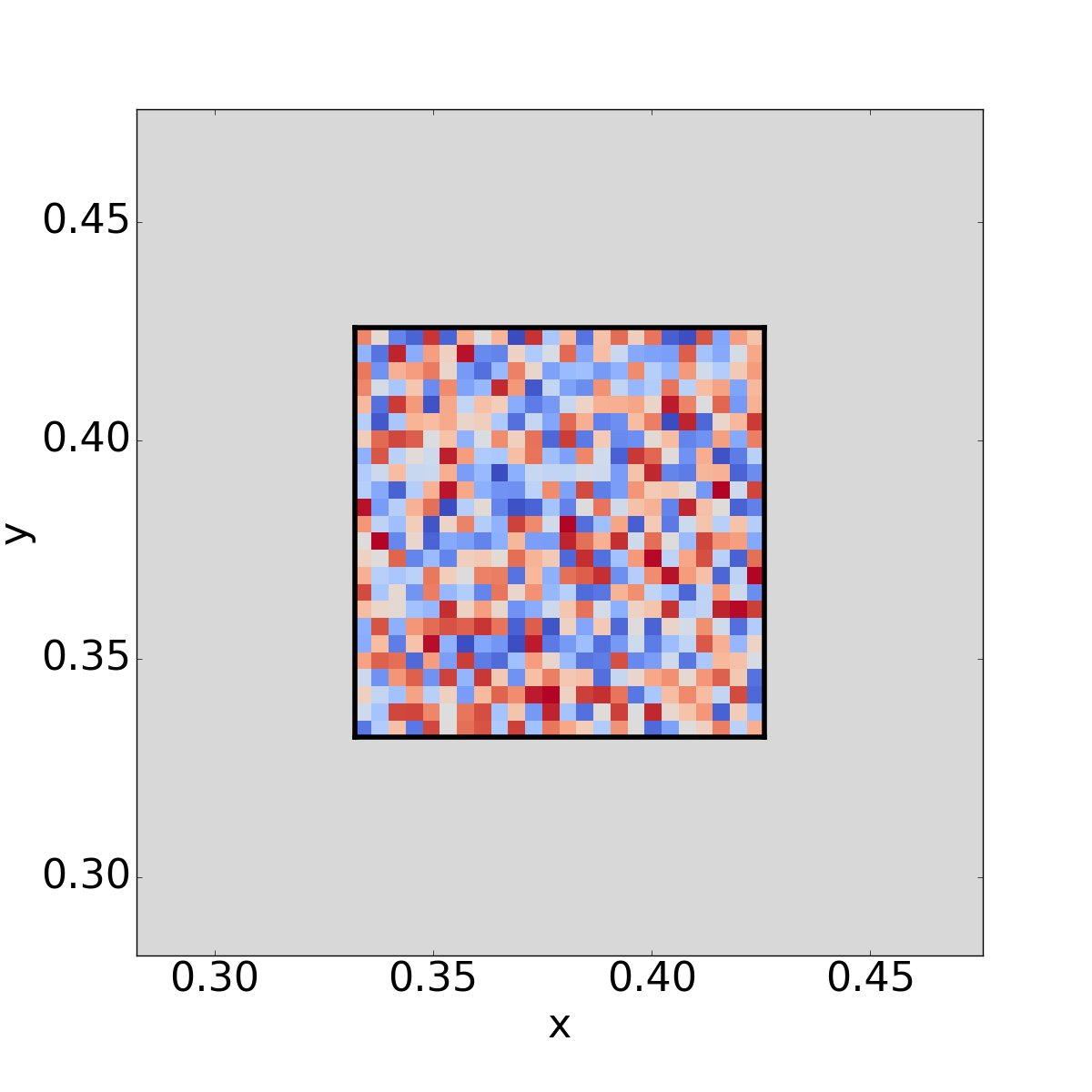}
	}
\subfigure[Sub-block 1 at $t_1$]{	
	\includegraphics[width = 2 in]{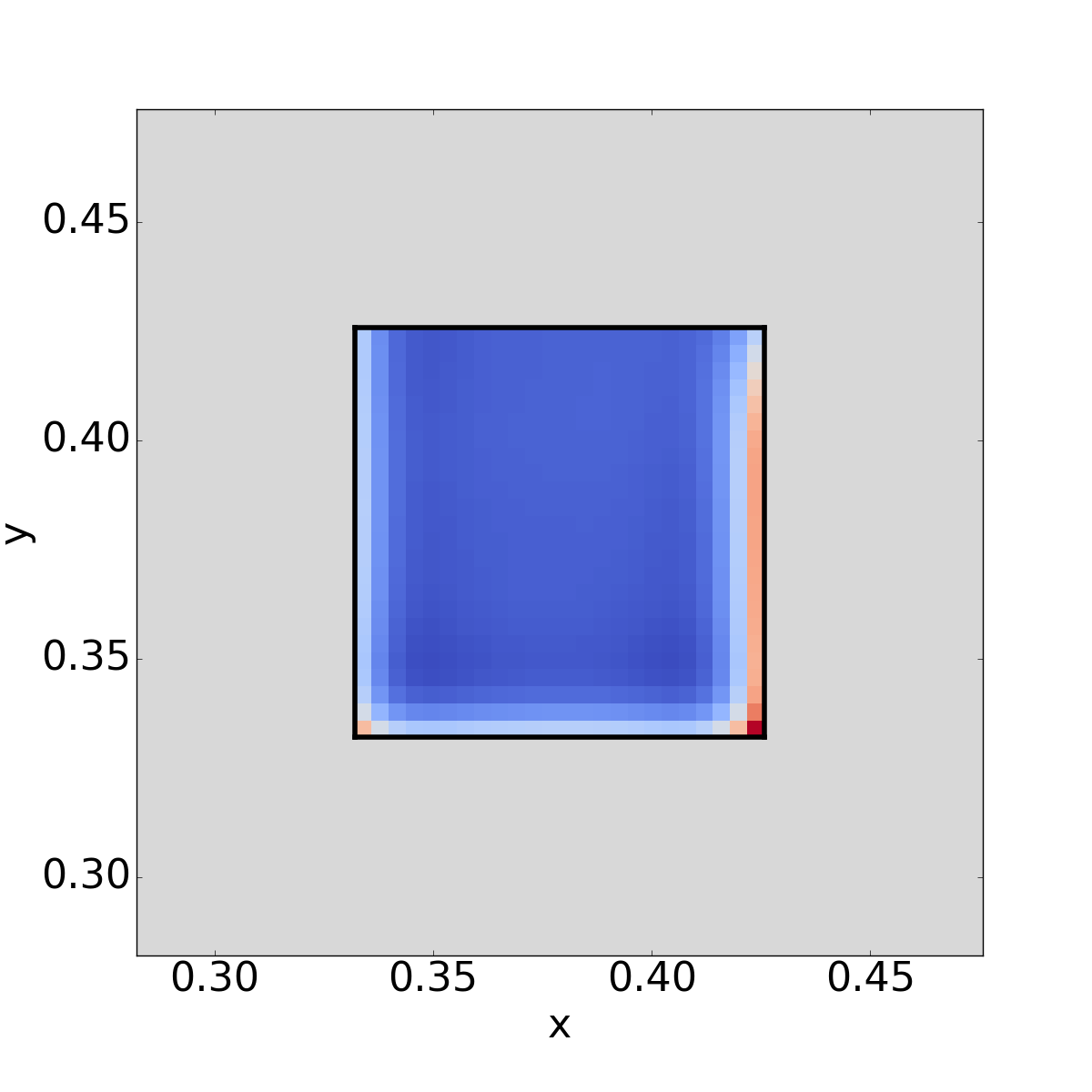}
	}\\
\subfigure[Sub-block 2 at $t_0$]{    
	\includegraphics[width = 2 in]{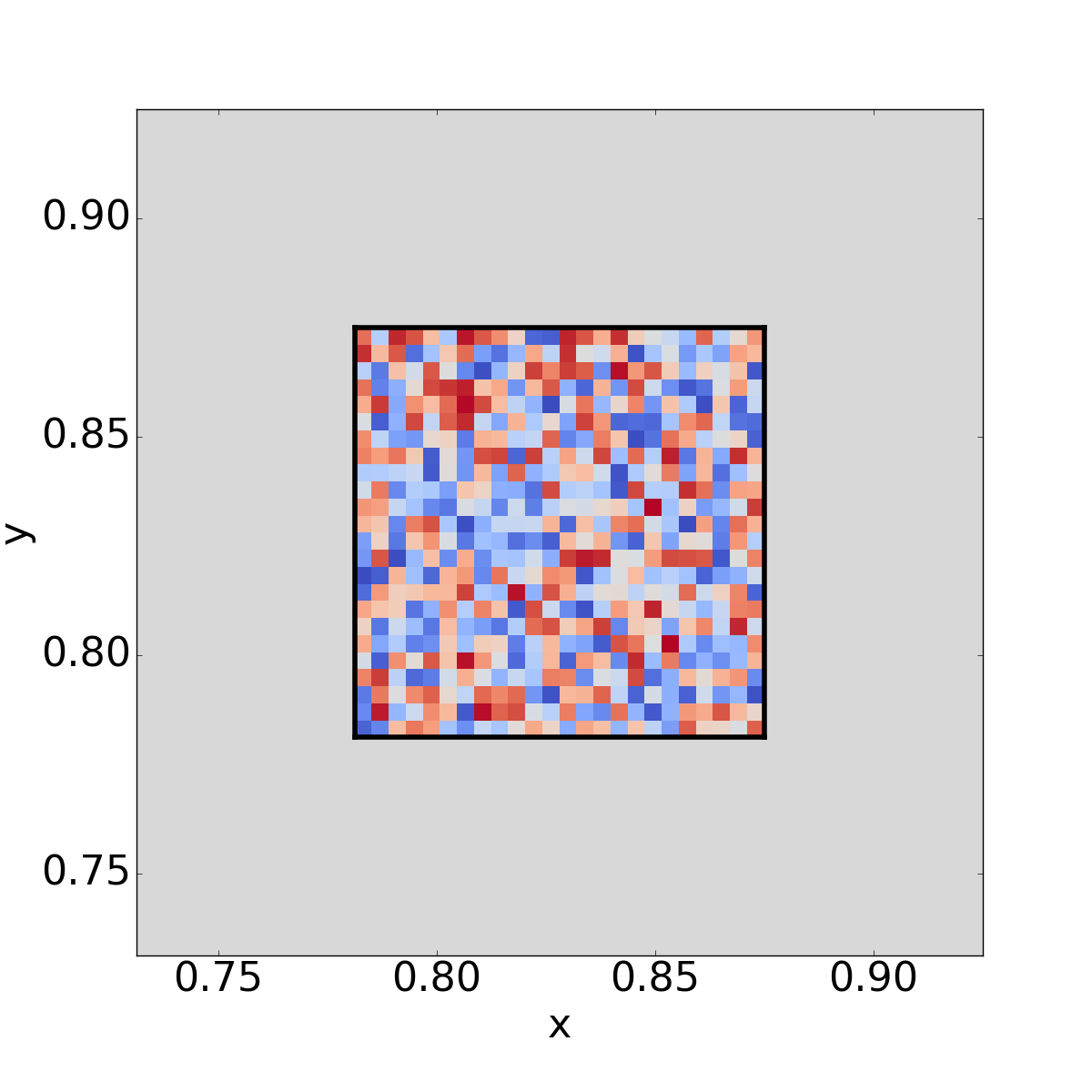}
	}
\subfigure[Sub-block 2 at $t_1$]{
	\includegraphics[width = 2 in]{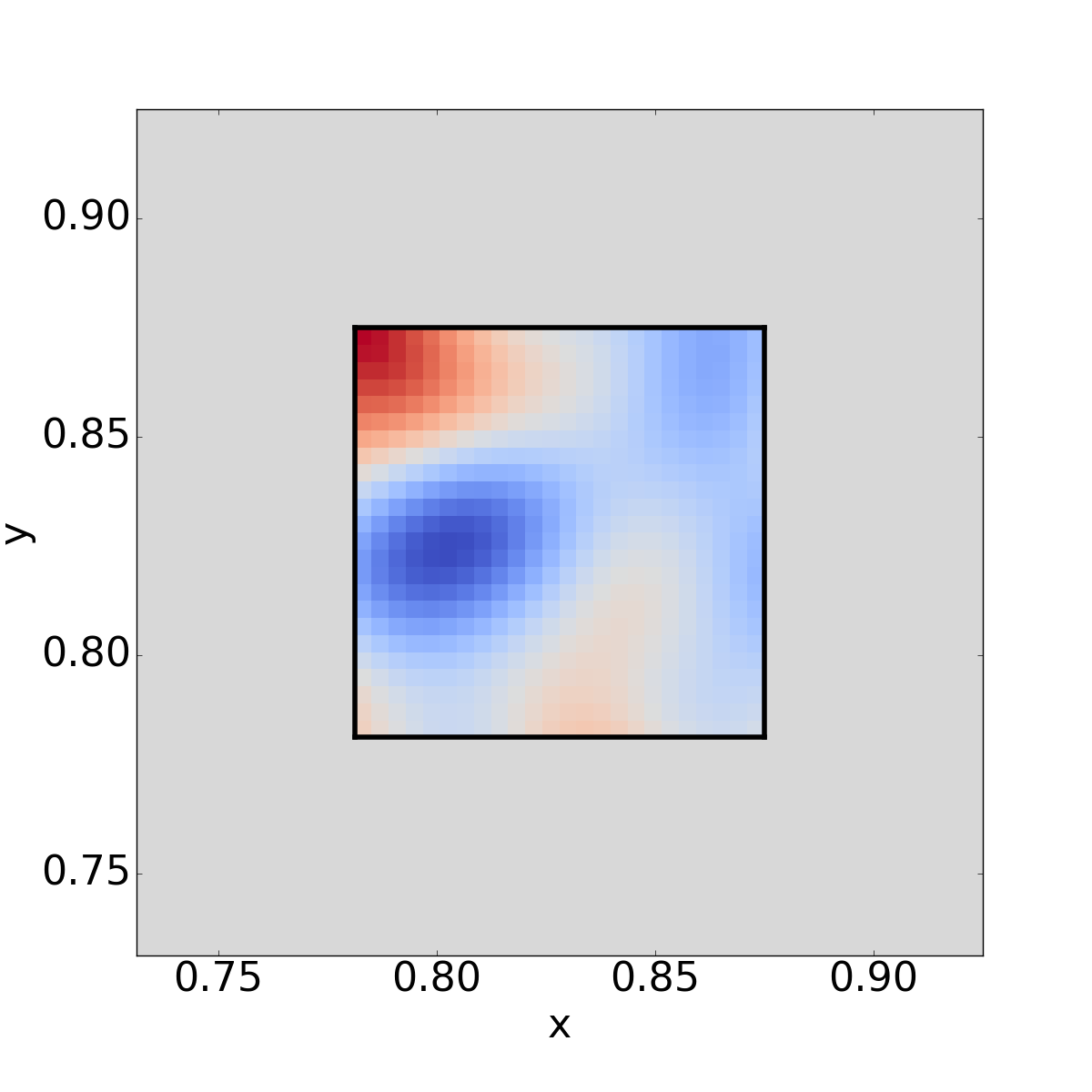}
	}
\caption{This figure includes a visual description of what the algorithm sees as its input. The first column corresponds to the system at $t_0$ and the second column corresponds to the system at $t_1$. The first row is the full state, which is not known to the user; the two highlighted blocks are what is actually given. In the second row, the first block and its evolution are shown and in the third row the second block and its evolution is shown. }
 \label{fig:num1}
 \end{figure}

An outline of the numerical method is provided Algorithm~\ref{algorithm}. The data, $u(t;k)\in \mathbb{R}^n$, is given at two consecutive time-steps $t=t_0$ and $t=t_1$, and each component of $u(t_0;k)$ is i.i.d. uniform. The number of samples must satisfy Equation~\eqref{eq:numburst}. First, the data must be arranged into the data matrix $U$ using the cyclic permutation construction, as detailed in Sections~\ref{section:oned}, \ref{section:higher}, and \ref{section:local}. Then, the data matrix is transformed so that each element is ranged in the interval $[-1,1]$. Using the transformed data matrix,  the Legendre dictionary matrix $A_L$ is computed using the basis described in Section~\ref{section:bod} and is normalized so that each column has unit $\ell^2$-norm. The coefficients with respect to the normalized Legendre dictionary is computed by solving Problem \eqref{eq:lbp} via the Douglas-Rachford method. The last step is to map the coefficients with respect to the normalized Legendre dictionary to the standard monomial basis. As an optional step, the problem $Ac=V$ with respect to the monomial dictionary can be re-solved by restricting it to the support set computed from the main algorithm. In particular, let $c$ be the output from Algorithm~\ref{algorithm} and $S=\supp(c)$, then the solution can be refined by solving the reduced system $A|_S \widetilde{c}=V$ (see also Proposition~\ref{prop2}).

\section{Computational Results}\label{sec:computing}

The method and algorithm are validated on a high-dimensional ODE as well as two finite dimensional evolution equations that arise as the discretization of nonlinear PDEs with synthetic data. In each case, the initial data is perturbed by a small amount of uniform noise. For the 2D examples, it is assumed that there exists a block of size $n\times n$ of the data which is nearly uniformly distributed in $[-1,1]^{n\times n}$ (possibly up to translation and rescaling). Similarly for the high-dimensional ODE case, one can restrict to a subset of the components.   Therefore, the input data to Problem \eqref{eq:lbp} is restricted to the block (see Figure \ref{fig:num1}; the restriction is described in Section~\ref{section:local}). It is important to note that the data restricted onto the blocks are not necessarily uniformly random; they may contain some slope. However, we assume that the dominate statistics are close to the uniform measure. In each of the examples, we apply the Douglas-Rachford algorithm described in Section~\ref{sec:numerics}, with the parameter $\sigma>0$ determined beforehand.

\subsection{The Lorenz 96 Equation} \label{sec: lorenz}

For the first example, we consider the Lorenz 96 equation:
\begin{equation*}
\dot{u}_j = -u_{j-2}\, u_{j-1} + u_{j-1} \, u_{j+1} - u_j + F, \quad j=1,2,\ldots,n,
\label{eqn:lorenz96}
\end{equation*}
for $j=1,\dots,n$ with periodic conditions $u_{-1}=u_{n-1}$, $u_0=u_n$, and $u_{n+1}=u_1$. We simulate the data using the forward Euler method with $n=128$ and $F=8$. The simulation is performed with a finer time-step $dt=5\times 10^{-5}$, but we only record the solution at the two time-stamps, the initial time $t_0 = 0$ and the final time $t_1=10^{-2}$. Let 
\begin{align*}
u(t) = \begin{pmatrix}u_1(t), u_2(t), \ldots, u_n(t) \end{pmatrix}^T\in\mathbb{R}^n,
\end{align*}
and set the initial data to be $u(0) = \nu$, where $\nu$ is sampled from the uniform distribution in $[-1,1]^n$. Assume that the input data is corrupted by additive Gaussian noise, and the resulting measurements are denoted by $\widetilde u$, \text{i.e.},
\begin{align*}
\widetilde u = u + \eta, \quad \eta \sim \mathcal{N}(0,\text{var}).
\end{align*}
To construct the velocity vector $V$, we use the following approximation of $\dot u$:
\begin{align*} 
\dot u_i(t_0) := 
\dfrac{\widetilde u_i(t_1)-\widetilde u_i(t_0)}{dt}, \quad i=1,2,\ldots,n.
\end{align*}
In this example, we vary the variance of the additive noise and the size of the dictionary, and compare the accuracy of the recovery under different noise levels and dictionary sizes. The results are provided in Section \ref{sec:resultsdiscuss}.

\subsection{Viscous Burgers' Equation} \label{sec: burgers}

Consider a 2D variant of the viscous Burgers' Equation:
\begin{align*}
u_t =\alpha\Delta u+ u\, u_x+u\, u_y,
\end{align*}
where $\Delta$ is the Laplacian operator and is defined by $\Delta u = u_{xx} + u_{yy}$, and $\alpha>0$ is the viscosity. The equation is spatially invariant and well-posed, and thus there exists a discretization that yields a finite dimensional system that satisfies the cyclic structural condition. In particular, we simulate the data using the finite dimensional semi-discrete system:
\begin{align*}
\dot{u}_{i,j} = \alpha\dfrac{u_{i+1,j}+u_{i-1,j}+u_{i,j+1}+u_{i,j-1}-4u_{i,j}}{h^2} + \dfrac{\left(u_{i+1,j}\right)^2-\left(u_{i-1,j}\right)^2}{4h} + \dfrac{\left(u_{i,j+1}\right)^2-\left(u_{i,j-1}\right)^2}{4h},
\end{align*}
for $i,j=1,2,\dots,n$, where $h$ is the grid spacing of the spacial domain, and $n=1/h$. For $\alpha$ large enough (relative to $h$), this semi-discrete system is convergent. Note that this nonlinear evolution equation is $9$-sparse with respect to the standard monomial basis in terms of $u_{i,j}$.  We simulate the data using the discrete system above with a $128\times 128$ grid, \textit{i.e.} $h=1/128$, and $\alpha=10^{-2}$. This choice of $\alpha$ allows for both nonlinear and diffusive phenomena over the time-scale that we are sampling. The simulation is performed with a finer time-step $dt=5\times 10^{-8}$, but the solution is only recorded at the two time-stamps, the initial time $t_0 = 0$ and the final time $t_1=10^{-5}$.

The results appear in Figure~\ref{fig:burgers}. The initial data is plotted in Figures~\ref{fig: burger_initial_withblock}-\ref{fig: burger_initial}, and is given by:
\begin{align*}
u_0(x,y) = 50\sin(8\pi (x-0.5)) \ \exp\left(-\frac{(x-0.5)^2+(y-0.5)^2}{0.05}\right) + \nu,
\end{align*}
where $\nu$ is sampled from the uniform distribution in $[-1,1]^{128\times128}$. To construct the velocity vector $V$, we use the following approximation of $\dot u$:
\begin{align} 
\dot u_{i,j}(t_0) := 
\dfrac{u_i(t_1)-u_i(t_0)}{dt}, \quad i,j=1,2,\ldots,n. \label{eq: approximation udot}
\end{align}
The input to the algorithm is a block of size $7\times 7$. For display purpose, we mark in Figure~\ref{fig: burger_initial_withblock} the location of the block which is used as the input. The learned equation is given by:
\begin{align}
\dot{u}_{i,j} &= - 655.9404u_{i,j} + 163.3892u_{i+1,j} + 163.5089u_{i-1,j} + 163.4859u_{i,j+1} +163.5551u_{i,j-1} \nonumber \\
&\quad + 31.9211\left(u_{i+1,j}\right)^2 - 31.7654\left(u_{i-1,j}\right)^2 + 31.7716\left(u_{i,j+1}\right)^2 - 31.8849\left(u_{i,j-1}\right)^2, \label{eq: burger learned}
\end{align}
compared to the exact equation:
\begin{align}
\dot{u}_{i,j} &= - 655.36u_{i,j} + 163.84u_{i+1,j} + 163.84u_{i-1,j} + 163.84u_{i,j+1} +163.84u_{i,j-1} \nonumber \\ 
&\quad + 32\left(u_{i+1,j}\right)^2 - 32\left(u_{i-1,j}\right)^2 + 32\left(u_{i,j+1}\right)^2 - 32\left(u_{i,j-1}\right)^2. \label{eq: burger exact}
\end{align}
The correct $9$-terms are selected from the $351$ possible candidate functions.  To compare between the learned and true evolutions, we simulate the two systems up to the time of the shock formation, which is well beyond the interval of learning. Note that the qualitative difference between the two shocks is small.

\begin{figure}[b!]
\centering
\subfigure[Initial data $u_0$]{\label{fig: burger_initial_withblock}
	\includegraphics[width = 2.3 in]{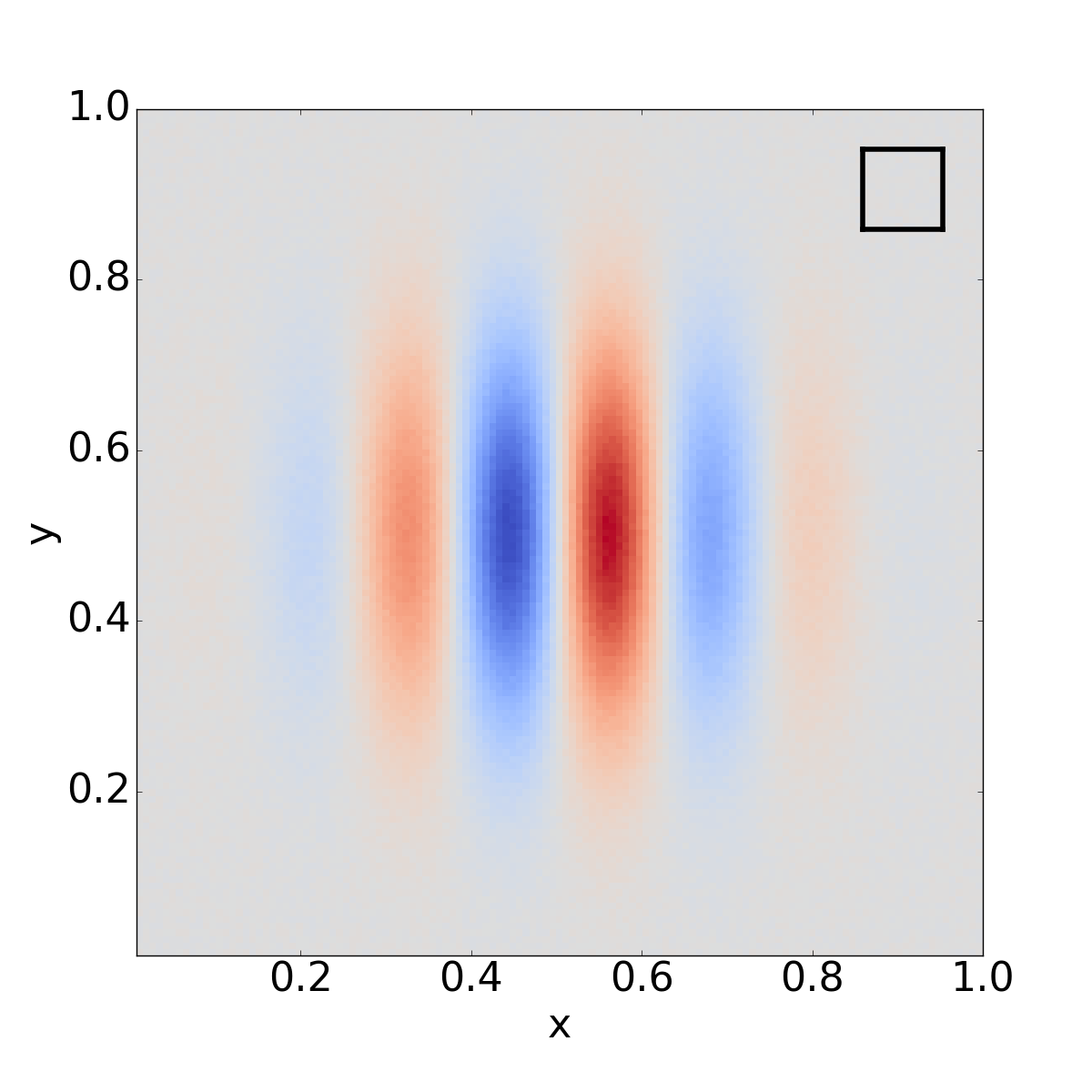}
	}
\subfigure[Initial data $u_0$]{\label{fig: burger_initial}	
	\includegraphics[width = 2.5 in]{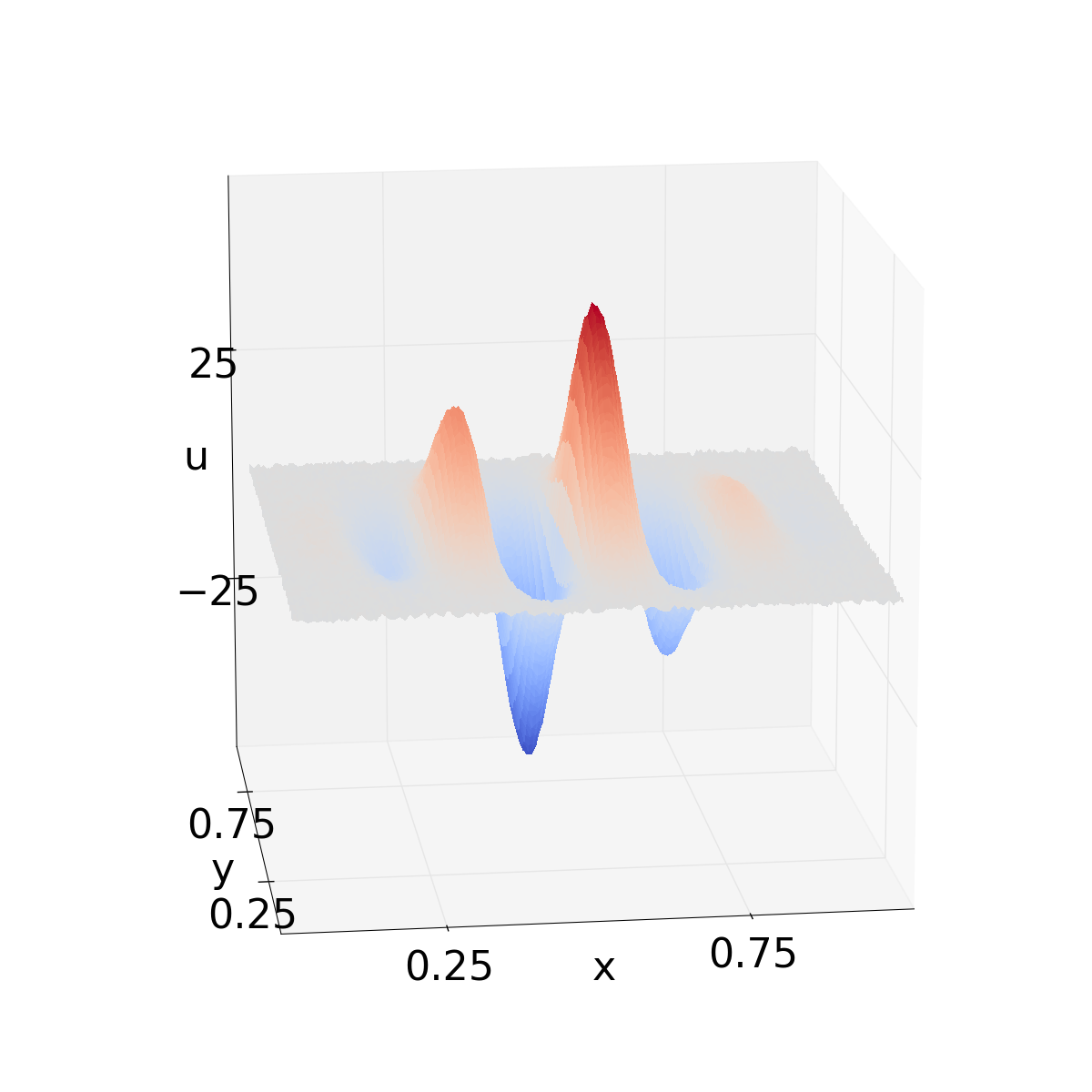}
	}
\subfigure[True evolution]{\label{fig: burger_realevolution}    
	\includegraphics[width = 2.5 in]{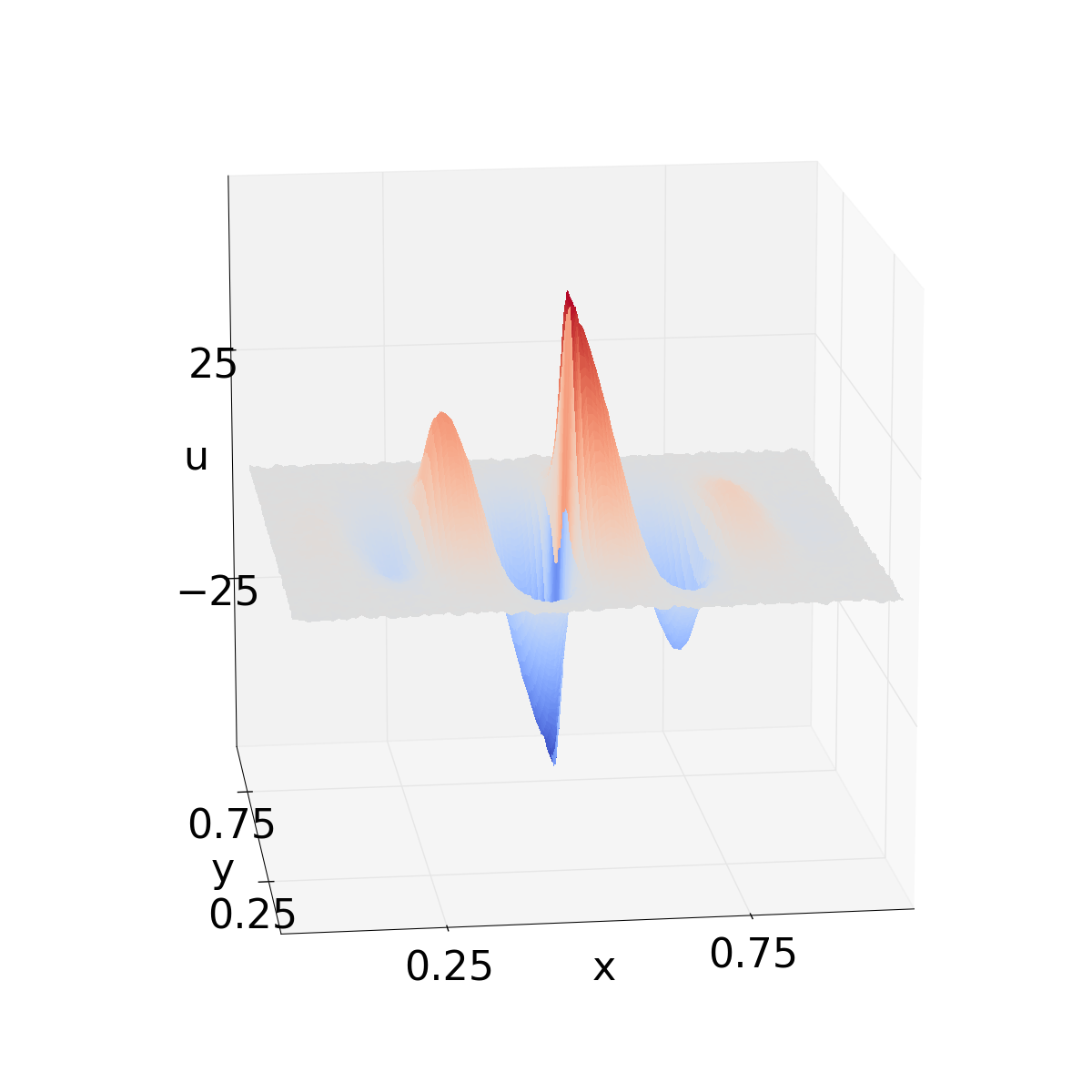}
	}
\subfigure[Learned evolution]{\label{fig: burger_learnedevolution}    
	\includegraphics[width = 2.5 in]{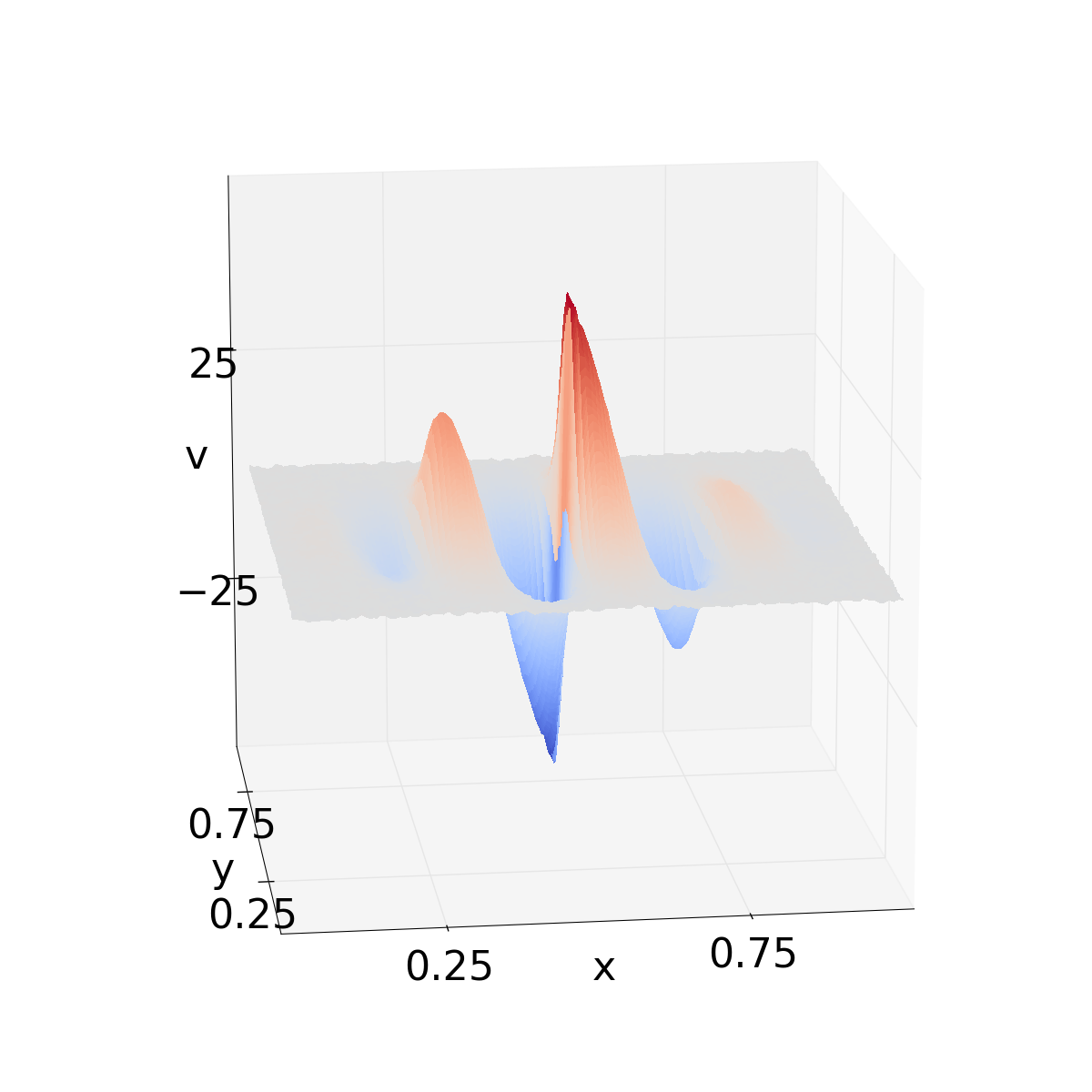}
	}
 \caption{\textbf{The Burgers's Equation}: (a) The initial data $u_0$ in a planar view; the sub-block in the boxed region is used as the input to the algorithm. (b) The initial data $u_0$ in a 3D view. (c) The true evolution at $T=10^{-3}$ using Equation \eqref{eq: burger exact}. (d) The learned evolution at $T=10^{-3}$ using Equation \eqref{eq: burger learned}. }
\label{fig:burgers}
\end{figure}

\subsection{Two Component Cubic Reaction-Diffusion Systems} \label{sec: grayscott}

Consider the 2D Gray-Scott Equation, which models a reaction-diffusion system:
\begin{align*}
u_t &= r_u\Delta u - uv^2 + f(1-u), \\
v_t &= r_v\Delta v + uv^2 - (f+k)v,
\end{align*}
where $r_u$ and $r_v$ are the diffusion rates of $u$ and $v$, respectively, $f$ is the processing rate of $u$, and $k$ represents the rate of conversion of $v$. We simulate the data using the finite dimensional semi-discrete system:
\begin{align*}
\dot{u}_{i,j} &= r_u\Delta_{h,9}u_{i,j} - u_{i,j}\left(v_{i,j}\right)^2 + f\left(1-u_{i,j}\right), \\
\dot{v}_{i,j} &= r_v\Delta_{h,9}v_{i,j} + u_{i,j}\left(v_{i,j}\right)^2 - (f+k)v_{i,j},
\end{align*}
for $i,j=1,2,\dots,n$, where $h$ is the grid spacing of the spacial domain, $n=1/h$, and $\Delta_{h,9}$ denotes the nine-point discrete Laplacian operator which is defined by:
\begin{align*}
\Delta_{h,9} u_{i,j}&= \dfrac{2}{3h^2}\left(u_{i+1,j}+u_{i-1,j}+u_{i,j+1}+u_{i,j-1}-5u_{i,j}\right) \\
&\quad + \dfrac{1}{6h^2}\left(u_{i+1,j+1}+u_{i-1,j+1}+u_{i-1,j+1}+u_{i-1,j-1}\right).
\end{align*}
Note that this nonlinear evolution equation is 12-sparse with respect to the standard monomial basis in terms of $u_{i,j}$ and is 11-sparse in terms of $v_{i,j}$.

We first present the implementation details for constructing Problem~\eqref{eq:lbp} in this setting (a system of PDEs). Given the initial data $u(t_0;k), v(t_0;k)\in\mathbb{R}^{n\times n}$, construct the data matrix $W(t_0;k)$ as follows:
\begin{align*}
W(t_0;k) = \begin{pmatrix}
 u_{1,1}(t_0;k) & u_{1,2}(t_0;k) & \cdots & u_{n,n}(t_0;k) & v_{1,1}(t_0;k) & v_{1,2}(t_0;k) & \cdots & v_{n,n}(t_0;k)\\
u_{1,2}(t_0;k) & u_{1,3}(t_0;k) & \cdots & u_{n,1}(t_0;k) & v_{1,2}(t_0;k) & v_{1,3}(t_0;k) & \cdots & v_{n,1}(t_0;k) \\
u_{1,3}(t_0;k) & u_{1,4}(t_0;k) & \cdots & u_{n,2}(t_0;k) & v_{1,3}(t_0;k) & v_{1,4}(t_0;k) & \cdots & v_{n,2}(t_0;k)\\
\vdots & \vdots &  \ddots & \vdots & \vdots & \vdots &  \ddots & \vdots \\
 u_{n,n}(t_0;k) & u_{n,1}(t_0;k) & \cdots & u_{n-1,n-1}(t_0;k) & v_{n,n}(t_0;k) & v_{n,1}(t_0;k) & \cdots & v_{n-1,n-1}(t_0;k)\\
\end{pmatrix}.
\end{align*}
Localization and restriction of $W(t_0;k)$ are performed with respect to both $u$ and $v$ independently. For example, with $n>7$, the restriction onto the indices $(i,j)\in\{3,4,5\}^2$ is given by:
\begin{align}
W(t_0;k)|_{9-pnts, restricted} = \begin{bmatrix} U(t_0;k)|_{9-pnts, restricted} & | & V(t_0;k)|_{9-pnts, restricted}\end{bmatrix}, \label{eqn: datamatrix2variables}
\end{align}
where $U(t_0;k)|_{9-pnts, restricted}$ is given by:
\begin{align*}
\begin{pmatrix}
u_{3,3}(t_0;k) & u_{3,4}(t_0;k) & u_{3,2}(t_0;k) & u_{4,3}(t_0;k) & u_{4,4}(t_0;k) & u_{4,2}(t_0;k) & u_{2,3}(t_0;k) & u_{2,4}(t_0;k) & u_{2,2}(t_0;k) \\
u_{3,4}(t_0;k) & u_{3,2}(t_0;k) & u_{3,3}(t_0;k) & u_{4,4}(t_0;k) & u_{4,2}(t_0;k) & u_{4,3}(t_0;k) & u_{2,4}(t_0;k) & u_{2,2}(t_0;k) & u_{2,3}(t_0;k) \\
& & & & \vdots & & & & \\
u_{5,5}(t_0;k) & u_{5,6}(t_0;k) & u_{5,4}(t_0;k) & u_{6,5}(t_0;k) & u_{6,6}(t_0;k) & u_{6,4}(t_0;k) & u_{4,5}(t_0;k) & u_{4,6}(t_0;k) & u_{4,4}(t_0;k)
\end{pmatrix},
\end{align*}
and $V(t_0;k)|_{9-pnts, restricted}$ is defined in the same way using the information of $v(t_0;k)$. Thus, we have reduced the size of the data matrix from $n\times (2n)$ to $9\times 18$. The localized and restricted dictionary matrix is then built by repeating the process in Equations~\eqref{eqn:quadmatrix}-\eqref{eqn:dictionaryburstmonomial}, but using the localized and restricted data matrix described above (see Equation~\eqref{eqn: datamatrix2variables}). The velocity vectors, $V_u$ for $\dot{u}_{i,j}$ and $V_v$ for $\dot{v}_{i,j}$, are constructed as in Equation~\eqref{eqn:dictionaryvelocity}, and $\dot{u}_{i,j}$ and $\dot{v}_{i,j}$ are approximated using Equation \eqref{eq: approximation udot}. Let $A_L$ be the (localized and restricted) dictionary in the Legendre basis. With the given system of PDEs, we then need to solve two basis pursuit problems:
\begin{align*}
\min_{c'_u } \ ||c'_u||_{1}  \ \ \  \text{subject to} \ \ \|A_Lc'_u-V_u\|_2 \leq \sigma,
\end{align*}
and,
\begin{align*}
\min_{c'_v } \ ||c'_v||_{1}  \ \ \  \text{subject to} \ \ \|A_Lc'_v-V_v\|_2 \leq \sigma,
\end{align*}
where $c'_u$ and $c'_v$ are the coefficients for the governing equations for $\dot{u}_{i,j}$ and $\dot{v}_{i,j}$, respectively, in the Legendre basis. Note that $A_L$ is the same between each of the basis pursuit problems above since each equation depends on both $u$ and $v$, but the outputs ($c'_u$ and $c'_v$) are cyclic independently. Its worth noting that this example extends beyond the theoretical results, since the entire governing equation is not cyclic, but it is cyclic in the components $(u,v)$.

We simulate the data using the discrete system above with a $128\times 128$ grid, \textit{i.e.} $h=1/128$, and parameters $r_u = 0.3$, $r_v = 0.15$. We consider three different parameter sets for the Gray-Scott model, by varying the values of $f$ and $k$. The simulation is performed with a finer time-step $dt=10^{-6}$, but the solution is only recorded at two time-stamps, the initial time $t_0 = 0$ and the final time $t_1=10^{-5}$.

The initial data is shown in Figure \ref{fig:grayscottsin}, and is given by:
\begin{align*}
u_0(x,y) &= 1 + 0.2\nu, \quad v_0(x,y) = \mathbb{I}_H(x,y) + 0.02\nu,
\end{align*}
where $\nu$ is sampled from the uniform distribution in $[-1,1]^{128\times128}$, and $H\subset[0,1]^2$ represented the H-shaped region in Figure~\ref{fig: gs_initial_v}. The input to the algorithm is a block of $u$ and the corresponding block of $v$, each of size $7\times7$. For display purposes, we mark the block's location in each of Figures~\ref{fig: gs_initial_u} and \ref{fig: gs_initial_v}.

\begin{figure}[b!]
\centering
\subfigure[Initial data $u_0$]{\label{fig: gs_initial_u}
	\includegraphics[width = 2 in]{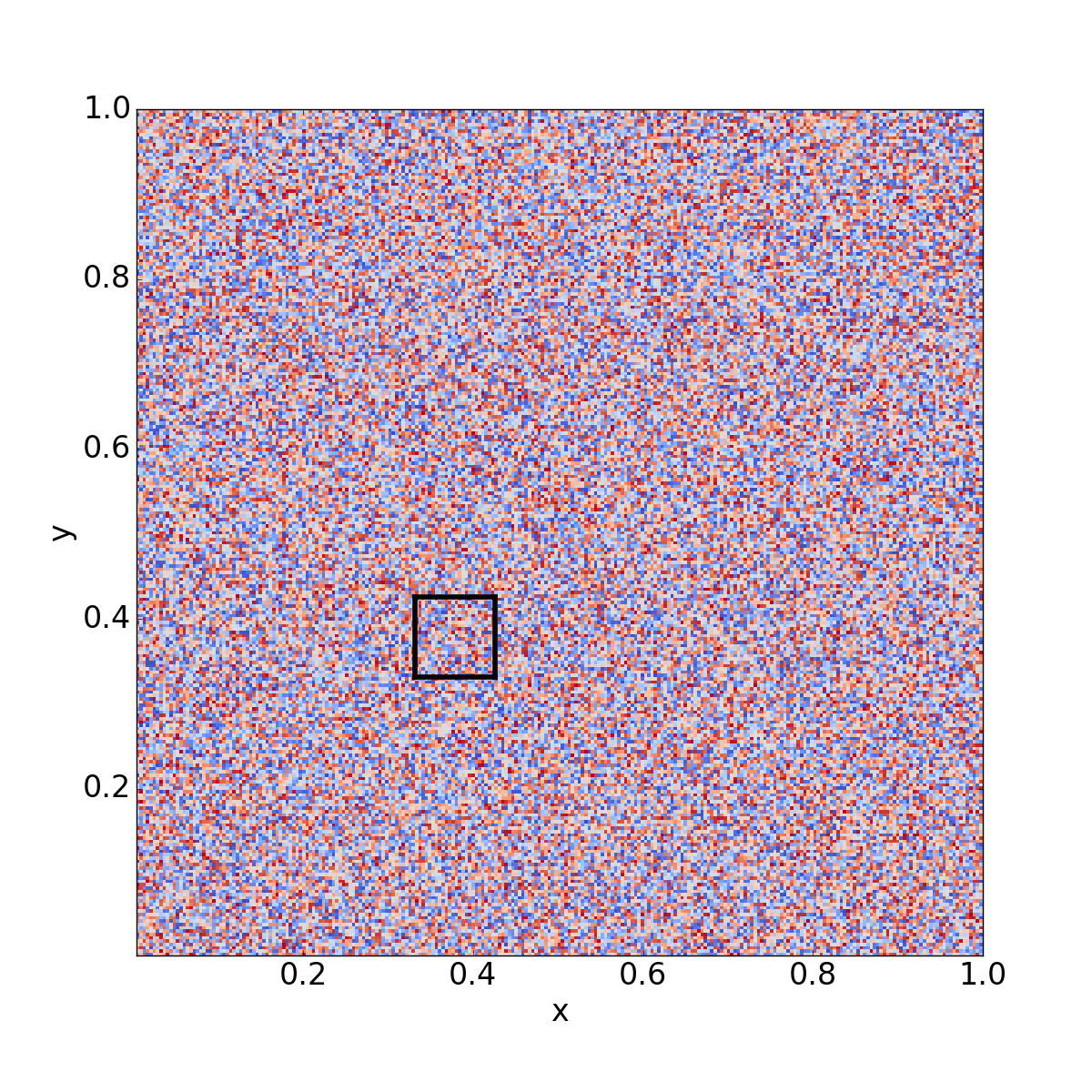}
	}
\subfigure[Sub-block of $u$ at $t_0$]{\label{fig: gs_u_block_t0}	
	\includegraphics[width = 2 in]{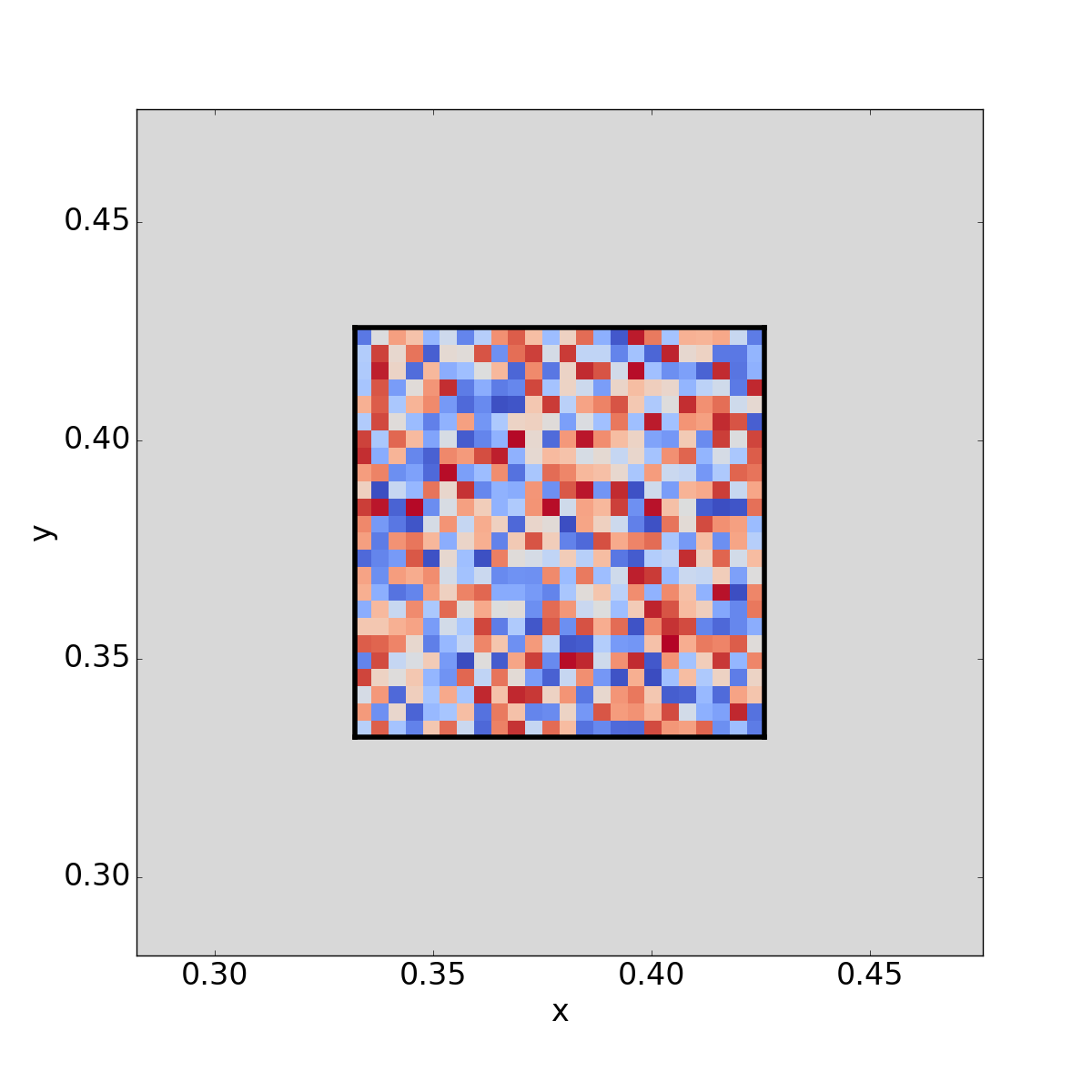}
	}
\subfigure[Sub-block of $u$ at $t_1$]{\label{fig: gs_u_block_t1}
	\includegraphics[width = 2 in]{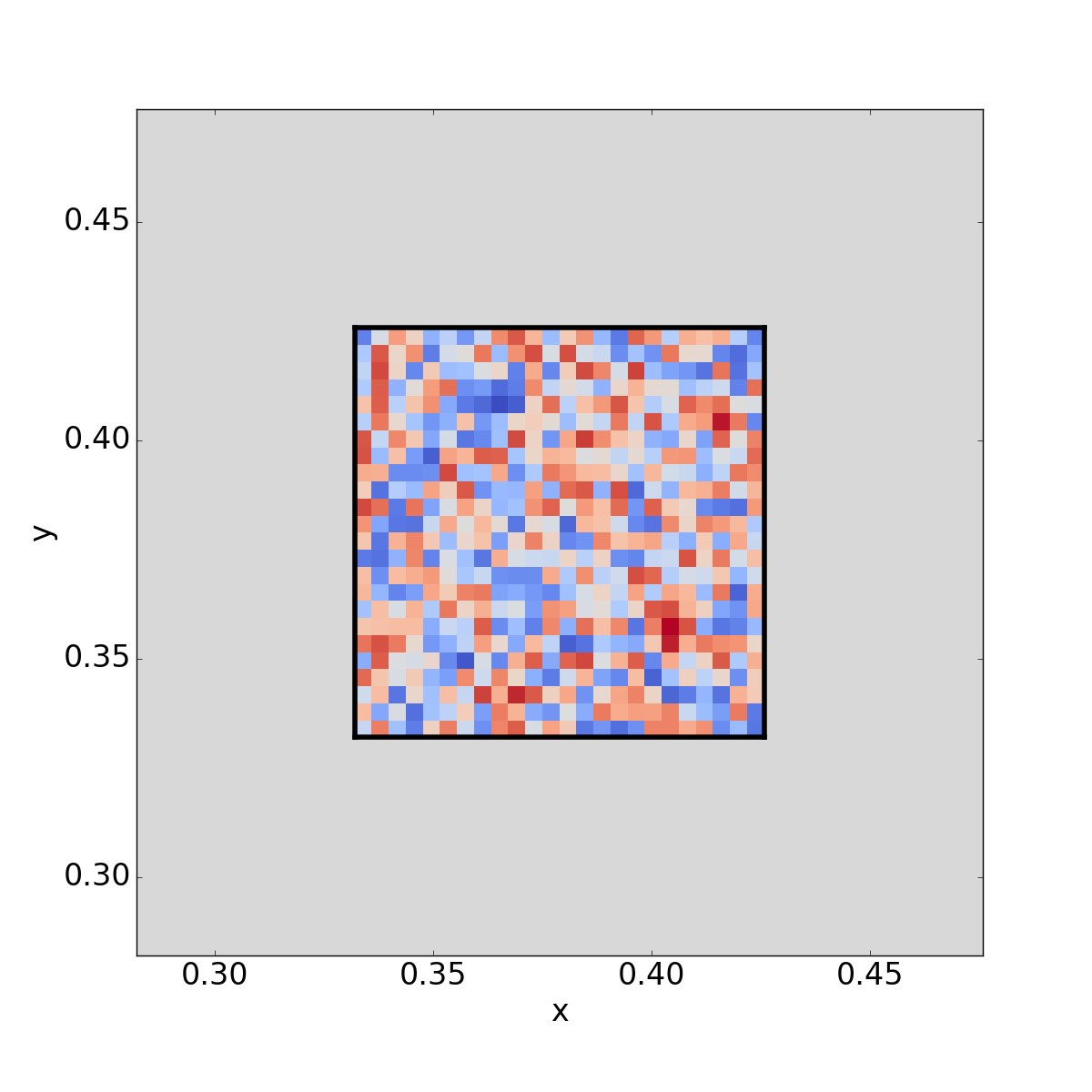}
	}
\subfigure[Initial data $v_0$]{\label{fig: gs_initial_v}	
	\includegraphics[width = 2 in]{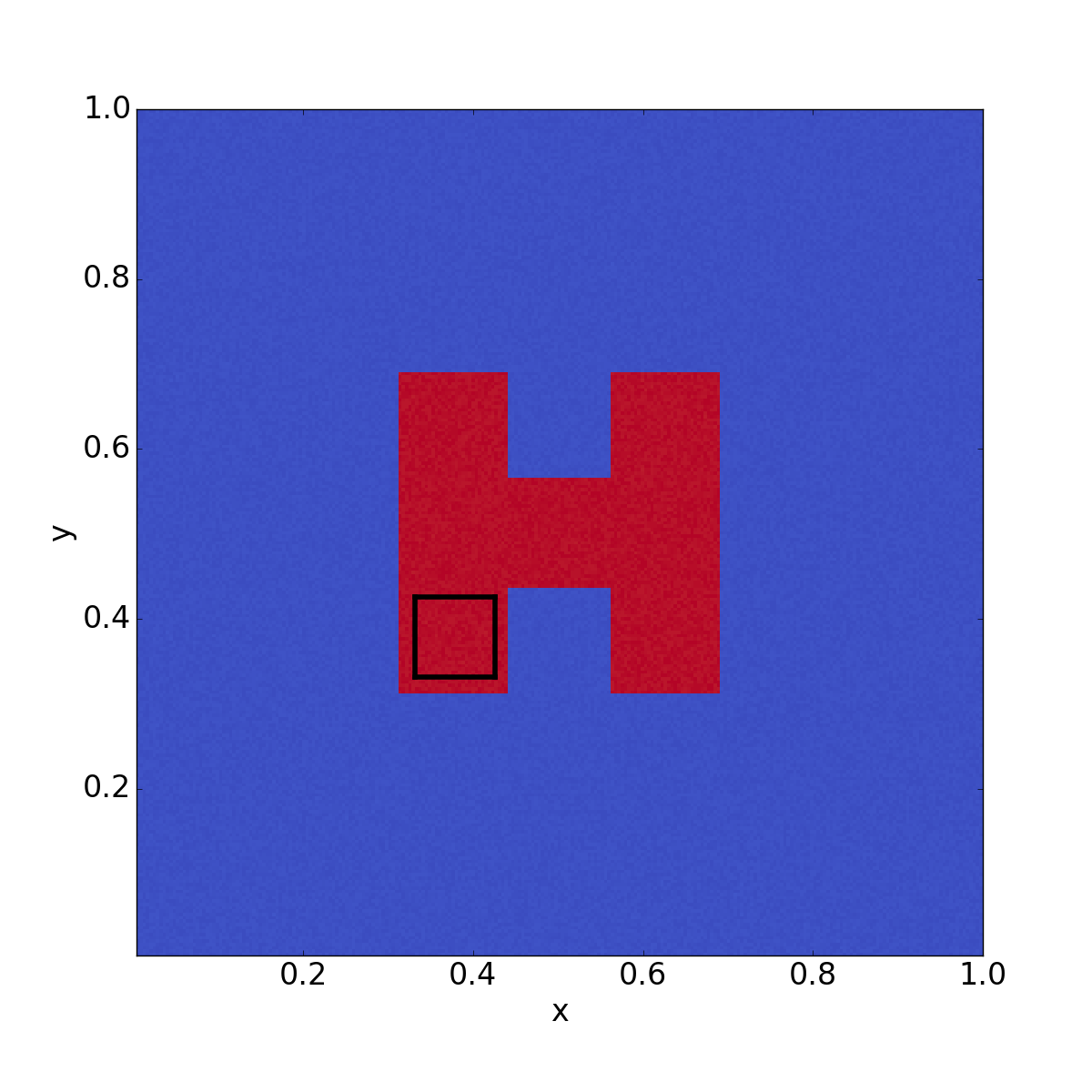}
	}
\subfigure[Sub-block of $v$ at $t_0$]{\label{fig: gs_v_block_t0}	
	\includegraphics[width = 2 in]{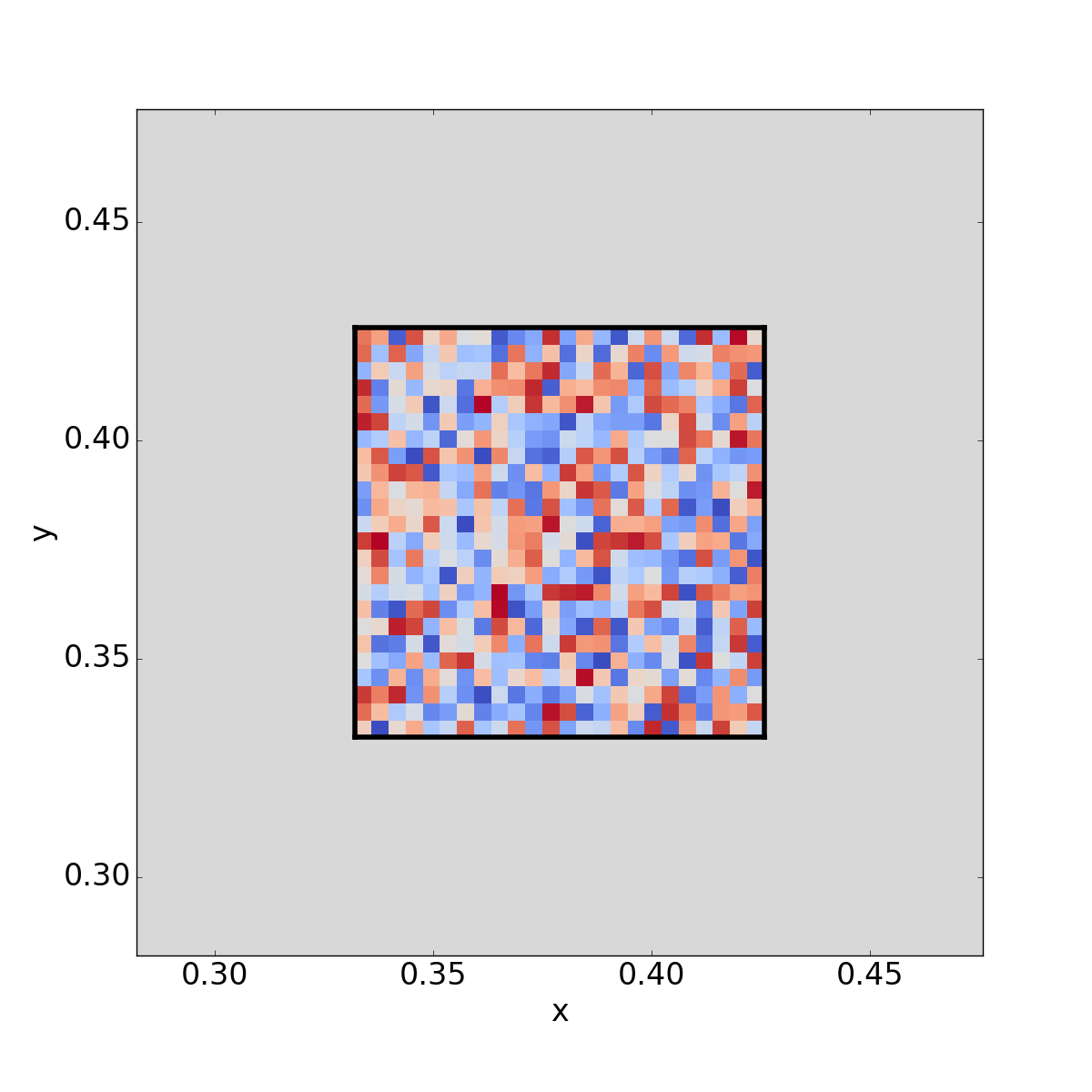}
	}
\subfigure[Sub-block of $v$ at $t_1$]{\label{fig: gs_v_block_t1}
	\includegraphics[width = 2 in]{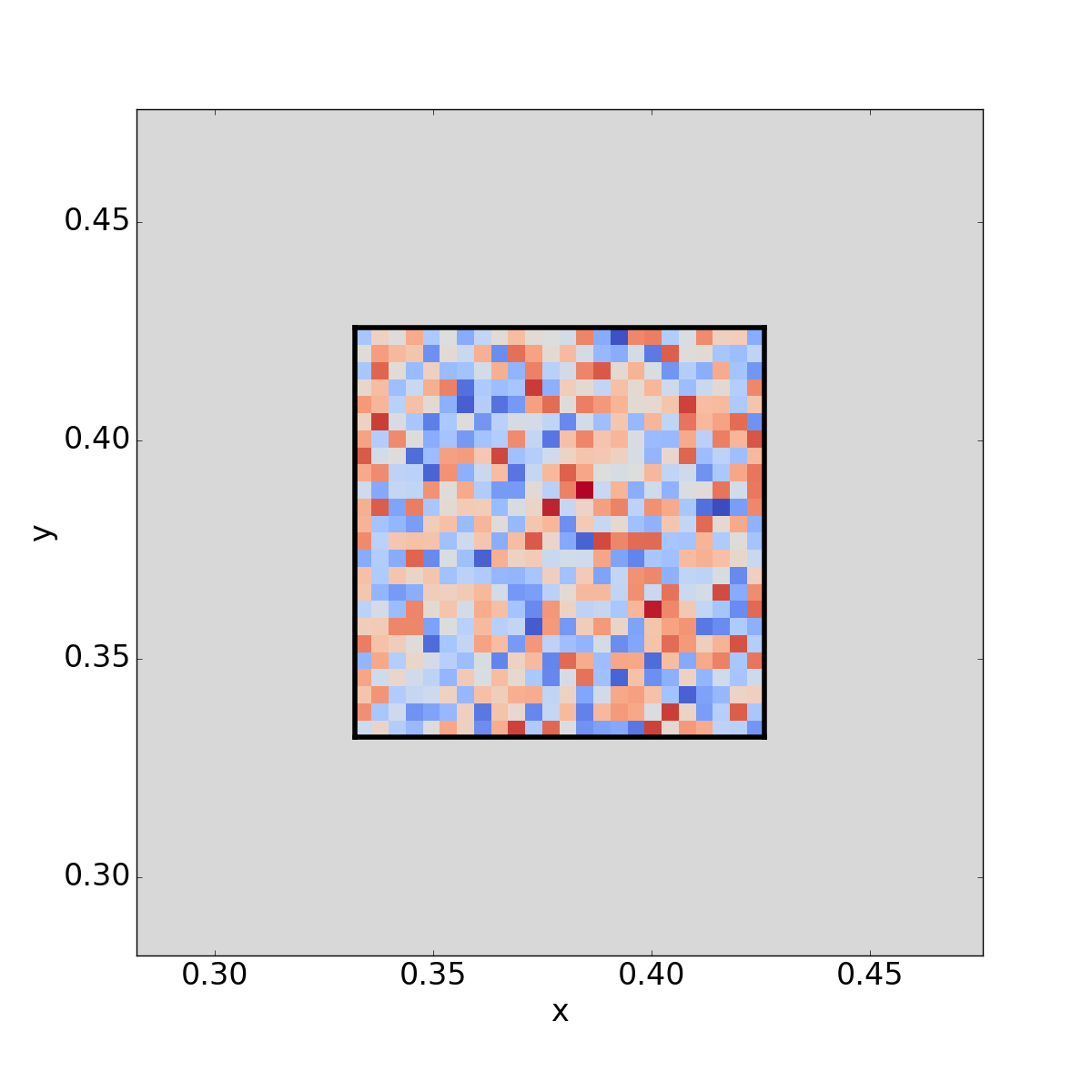}
	}
 \caption{\textbf{Initial Data for the Gray-Scott Equation}: (a)(d) The initial data $u_0$ and $v_0$; the sub-blocks in the boxed regions are used as the input to the algorithm. (b)(c) The sub-block of $u$ at time-stamps $t_0$ and $t_1$, whose measurements are used to compute $\dot{u}_{i,j}$. (e)(f) The sub-block of $v$ at time-stamps $t_0$ and $t_1$, whose measurements are used to compute $\dot{v}_{i,j}$.}
 \label{fig:grayscottsin}
 \end{figure}

For the first example, we use the parameters $f = 0.055$ and $k = 0.063$, which creates a ``coral'' pattern. The visual results are given in Figure~\ref{fig:grayscotts1}. The learned equations are:
\begin{equation}
\begin{split}
u_t &= 0.30000\Delta u - 1.00000uv^2 - 1.05500u + 0.05501, \\
v_t &= 0.15000\Delta v + 1.00000uv^2 - 0.61801 v - 0.00001,
\end{split} \label{eq: gs1 learned}
\end{equation}
compared to the exact equations:
\begin{equation}
\begin{split}
u_t &= 0.3\Delta u - uv^2 - 1.055u + 0.055, \\
v_t &= 0.15\Delta v + uv^2 - 0.618v.
\end{split} \label{eq: gs1 exact}
\end{equation}
To compare between the learned and true evolutions, we simulate the two systems up to time-stamp $T=5000$, well past the interval of learning. It is worth nothing that two evolutions are close (see Section~\ref{sec:resultsdiscuss} for errors).

\begin{figure}[b!]
\centering
\subfigure[True evolution of $u$]{\label{fig: gs1_u_realevolution}
	\includegraphics[width = 2 in]{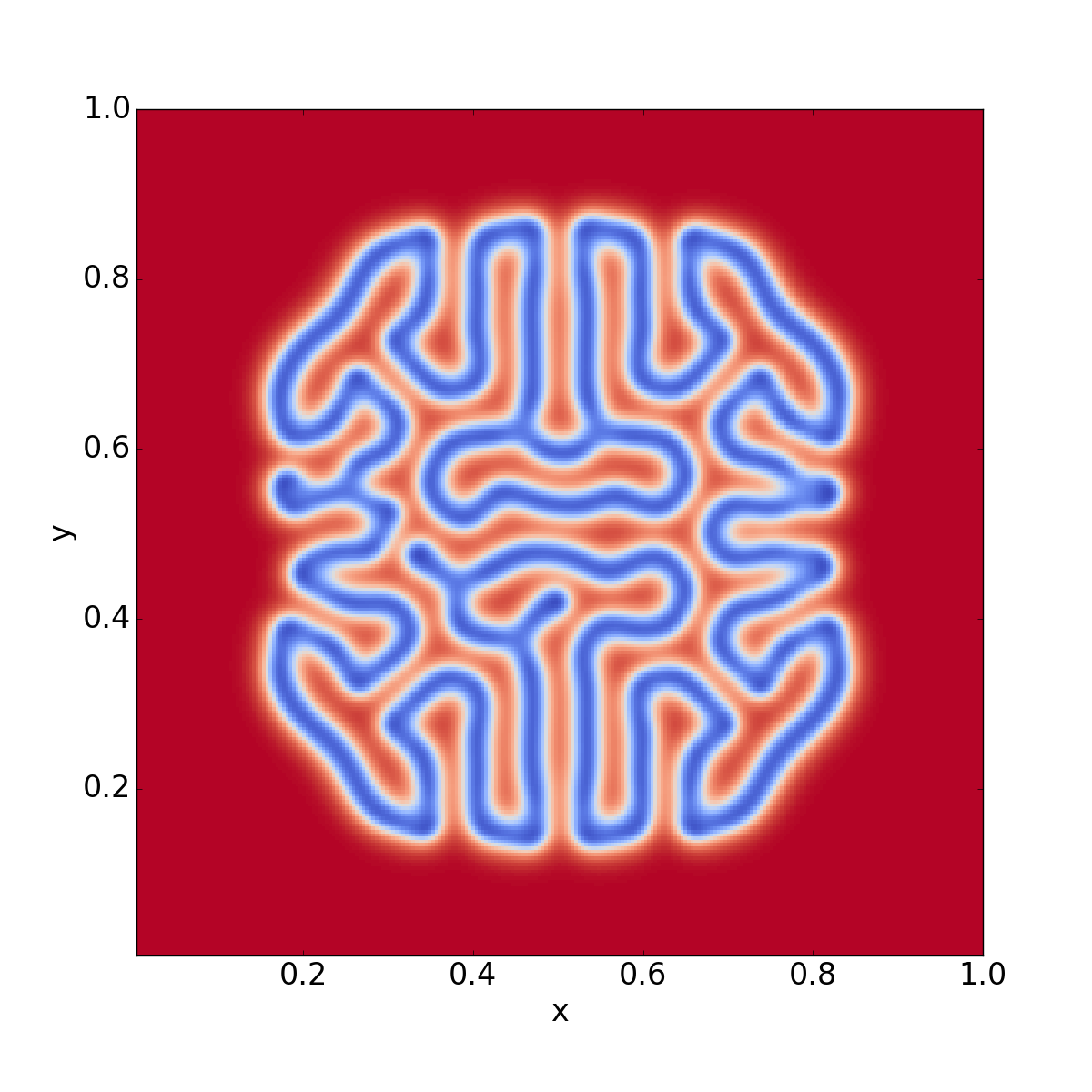}
	}
\subfigure[Learned evolution of $u$]{\label{fig: gs1_u_learnedevolution}	
	\includegraphics[width = 2 in]{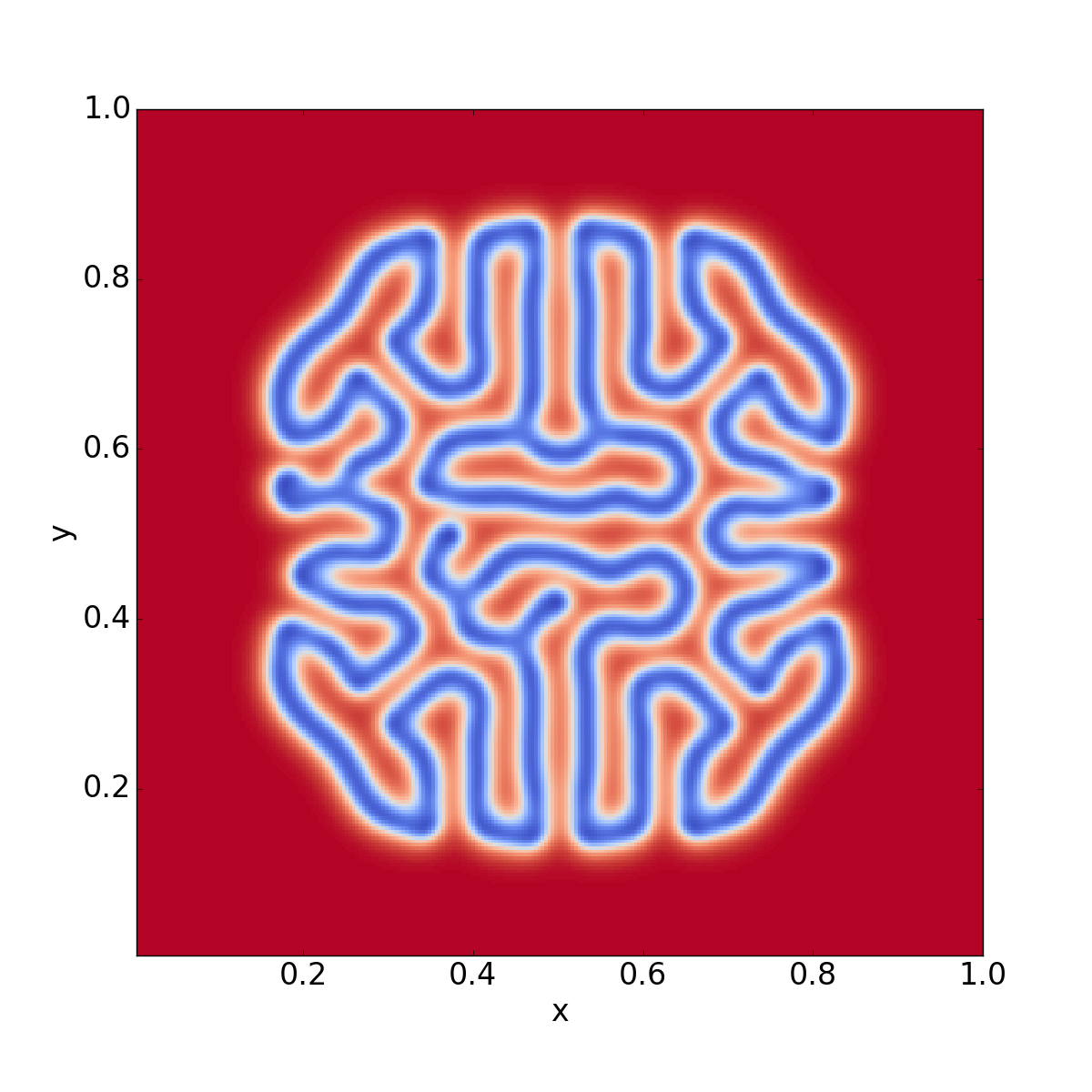}
	}
\subfigure[Difference in $u$]{\label{fig: gs1_u_difference}	
	\includegraphics[width = 2 in]{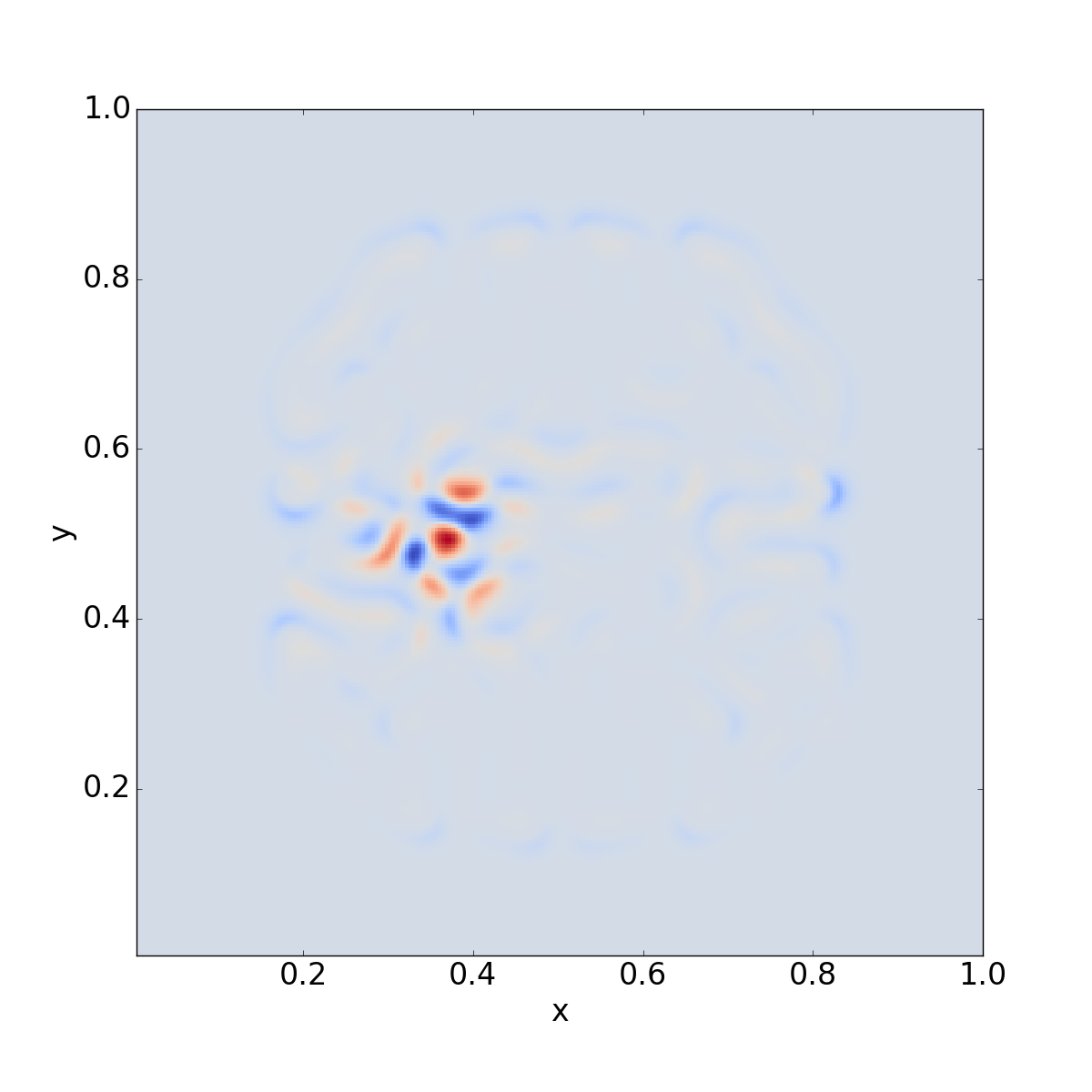}
	}
\subfigure[True evolution of $v$]{\label{fig: gs1_v_realevolution}	
	\includegraphics[width = 2 in]{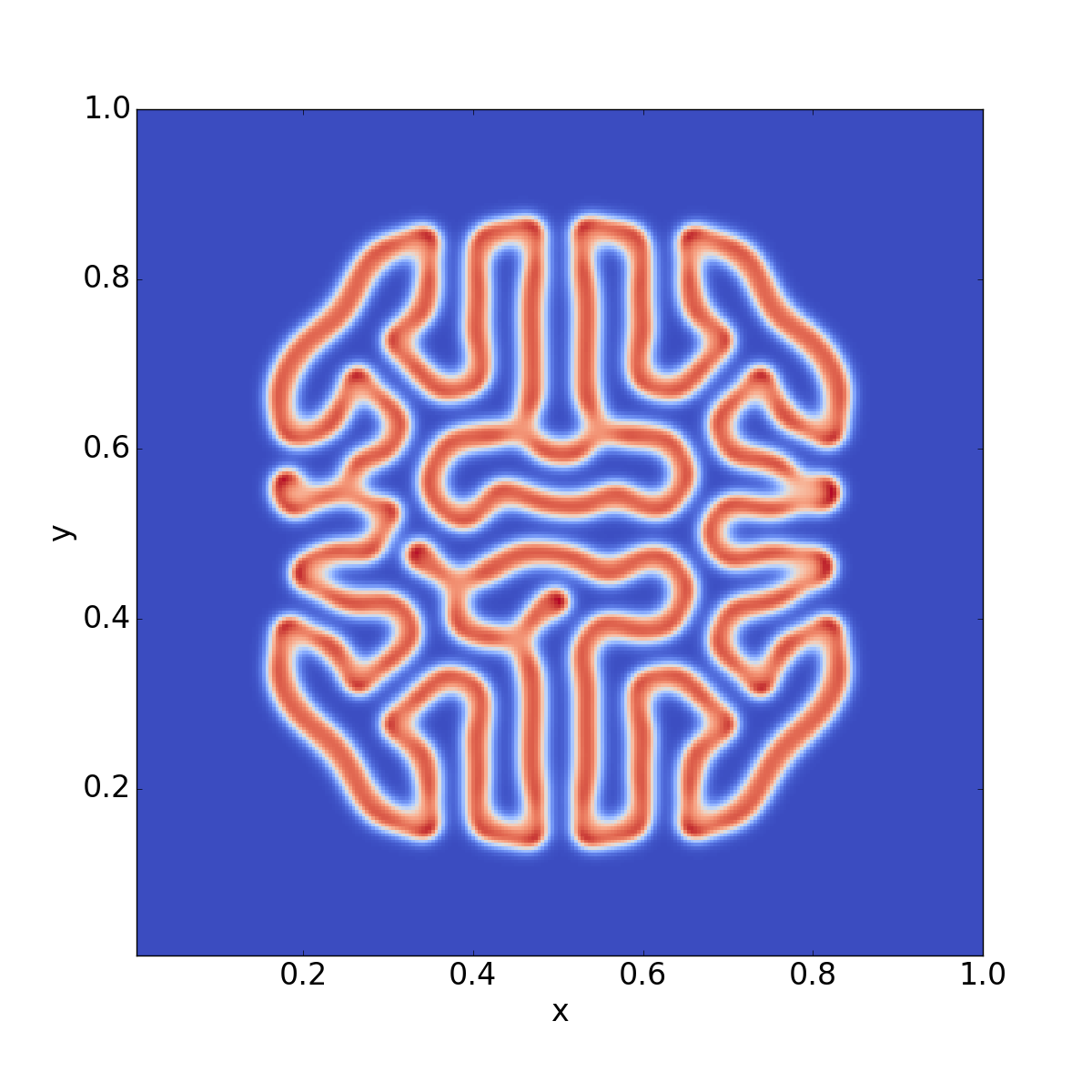}
	}
\subfigure[Learned evolution of $v$]{\label{fig: gs1_v_learnedevolution}	
	\includegraphics[width = 2 in]{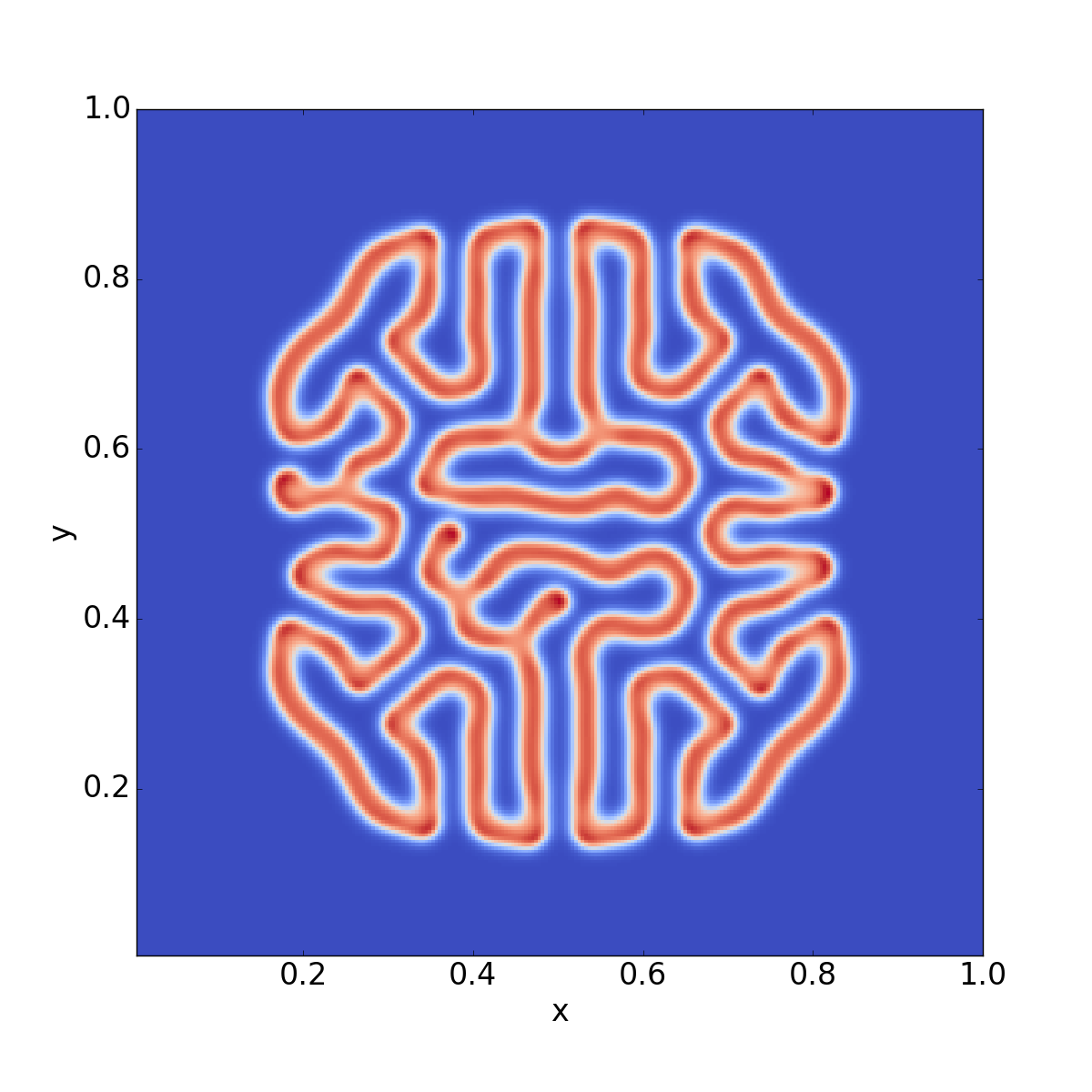}
	}
\subfigure[Difference in $v$]{\label{fig: gs1_v_difference}		
	\includegraphics[width = 2 in]{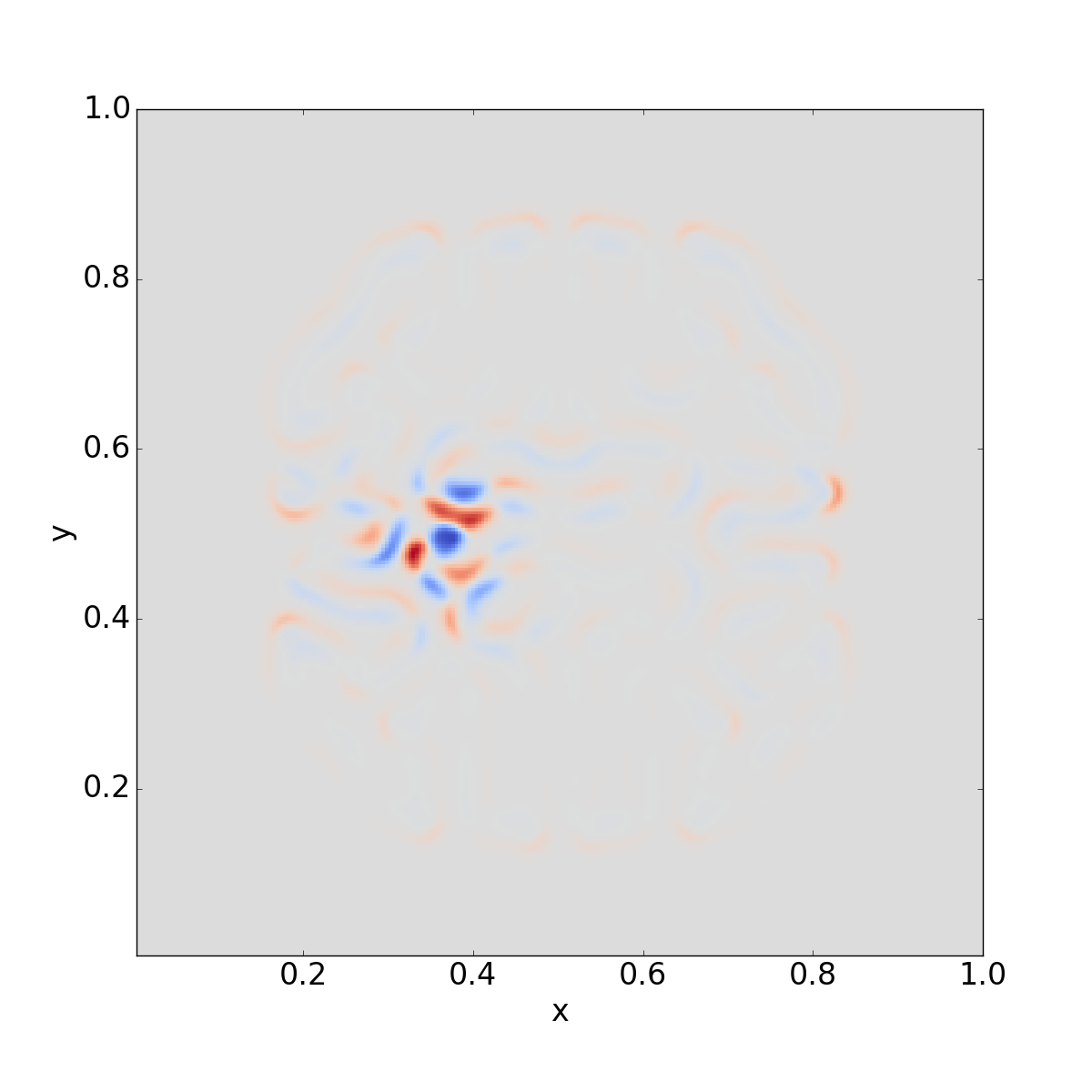}
	}
 \caption{\textbf{The Gray-Scott Equation, Example 1}: (a)(d) The true evolution at $T=5000$ using Equation \eqref{eq: gs1 exact}. (b)(e) The learned evolution at $T=5000$ using Equation \eqref{eq: gs1 learned}. (c)(f) The difference between the true evolution and the learned evolution. Note that the patterns are very similar except at a small region near the center.}
\label{fig:grayscotts1}
\end{figure}

In the second example, we use the parameters $f = 0.026$ and $k = 0.053$, which yields a hexagonal pattern. The visual results are given in Figure~\ref{fig:grayscotts2}. The learned equations is:
\begin{equation}
\begin{split}
u_t &= 0.30000\Delta u - 1.00000uv^2 - 1.02600u + 0.02601, \\
v_t &= 0.15000\Delta v + 1.00001uv^2 - 0.57901 v - 0.00001,
\end{split} \label{eq: gs2 learned}
\end{equation}
compared to the exact equations:
\begin{equation}
\begin{split}
u_t &= 0.3\Delta u - uv^2 - 1.026u + 0.026, \\
v_t &= 0.15\Delta v + uv^2 - 0.579v.
\end{split} \label{eq: gs2 exact}
\end{equation}
As before, to compare between the learned and true evolutions, we simulate the two systems up to time-stamp $T=2500$, beyond the learning interval. Small errors in the coefficient lead to some error in the pattern formulation; however, visually the simulations are similar.

\begin{figure}[b!]
\centering
\subfigure[True evolution of $u$]{\label{fig: gs2_u_realevolution}
	\includegraphics[width = 2 in]{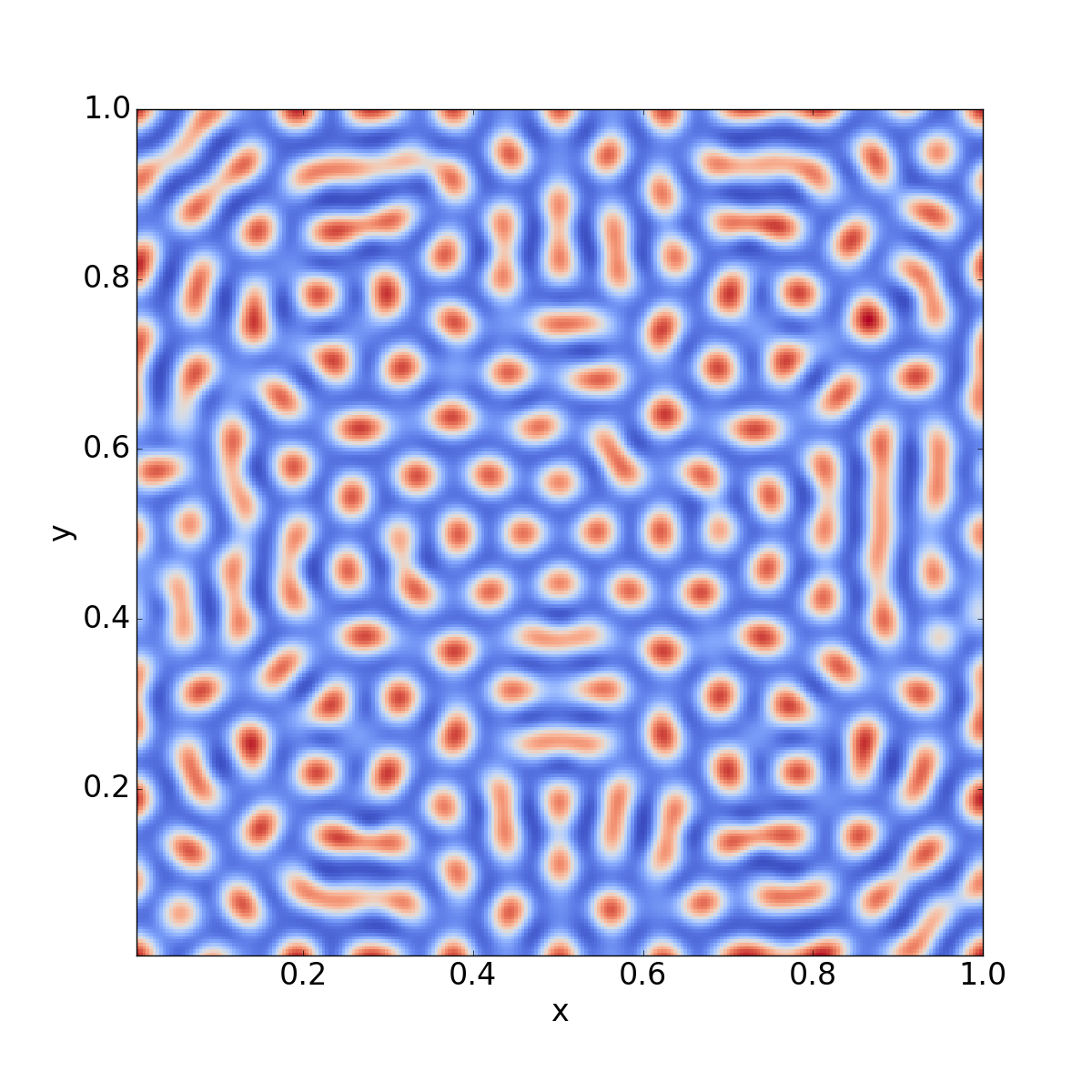}
	}
\subfigure[Learned evolution of $u$]{\label{fig: gs2_u_learnedevolution}	
	\includegraphics[width = 2 in]{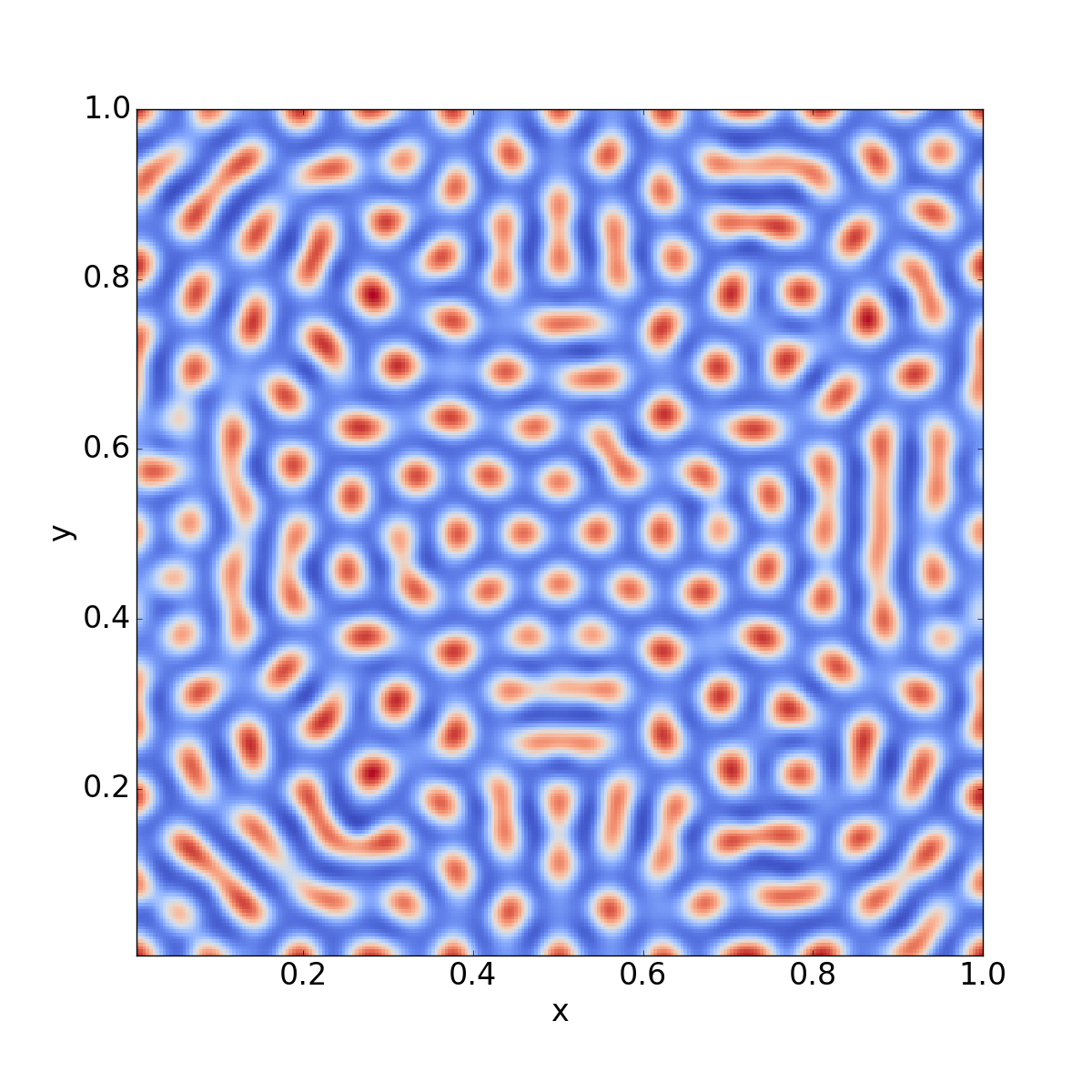}
	}
\subfigure[Difference in $u$]{\label{fig: gs2_u_difference}	
	\includegraphics[width = 2 in]{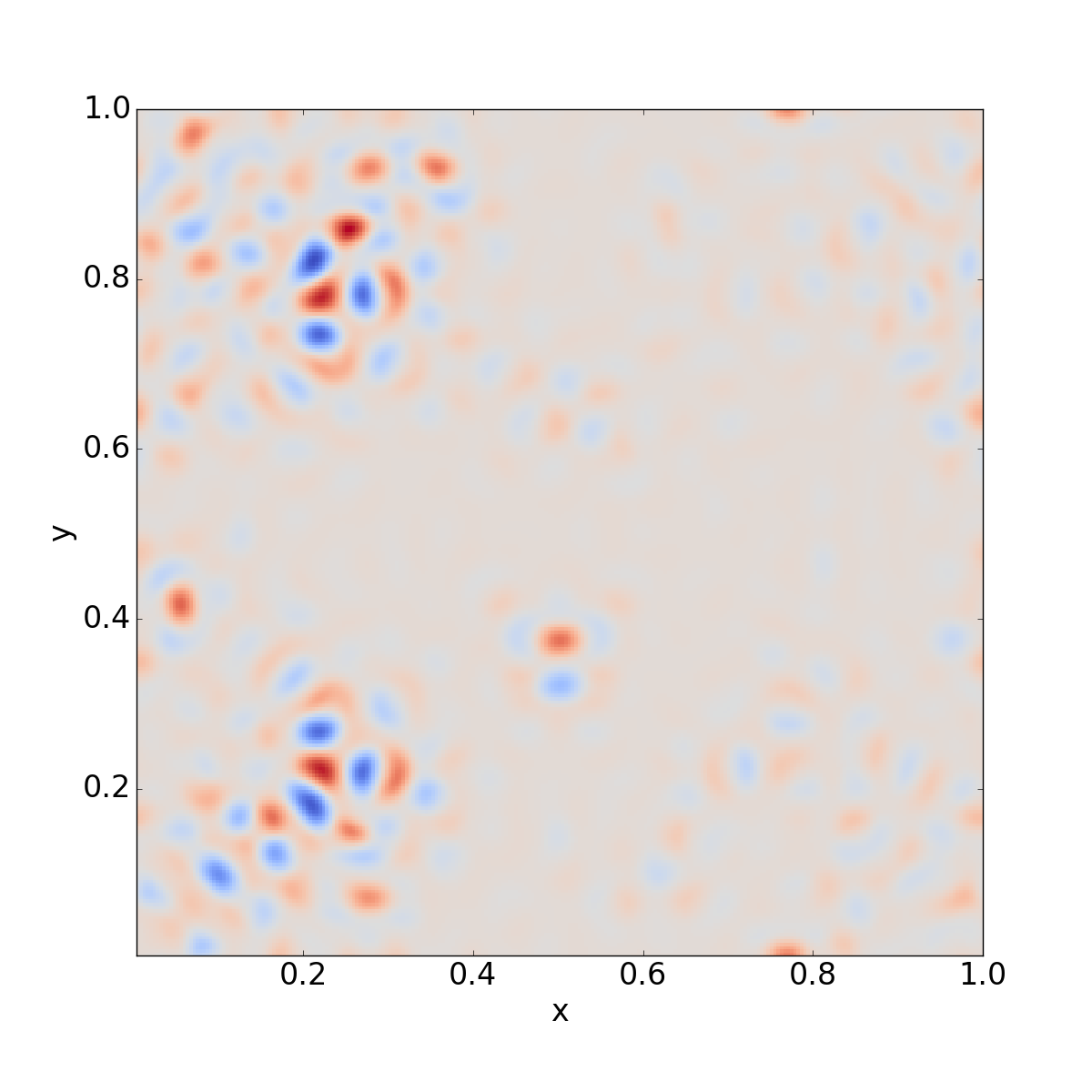}
	}
\subfigure[True evolution of $v$]{\label{fig: gs2_v_realevolution}	
	\includegraphics[width = 2 in]{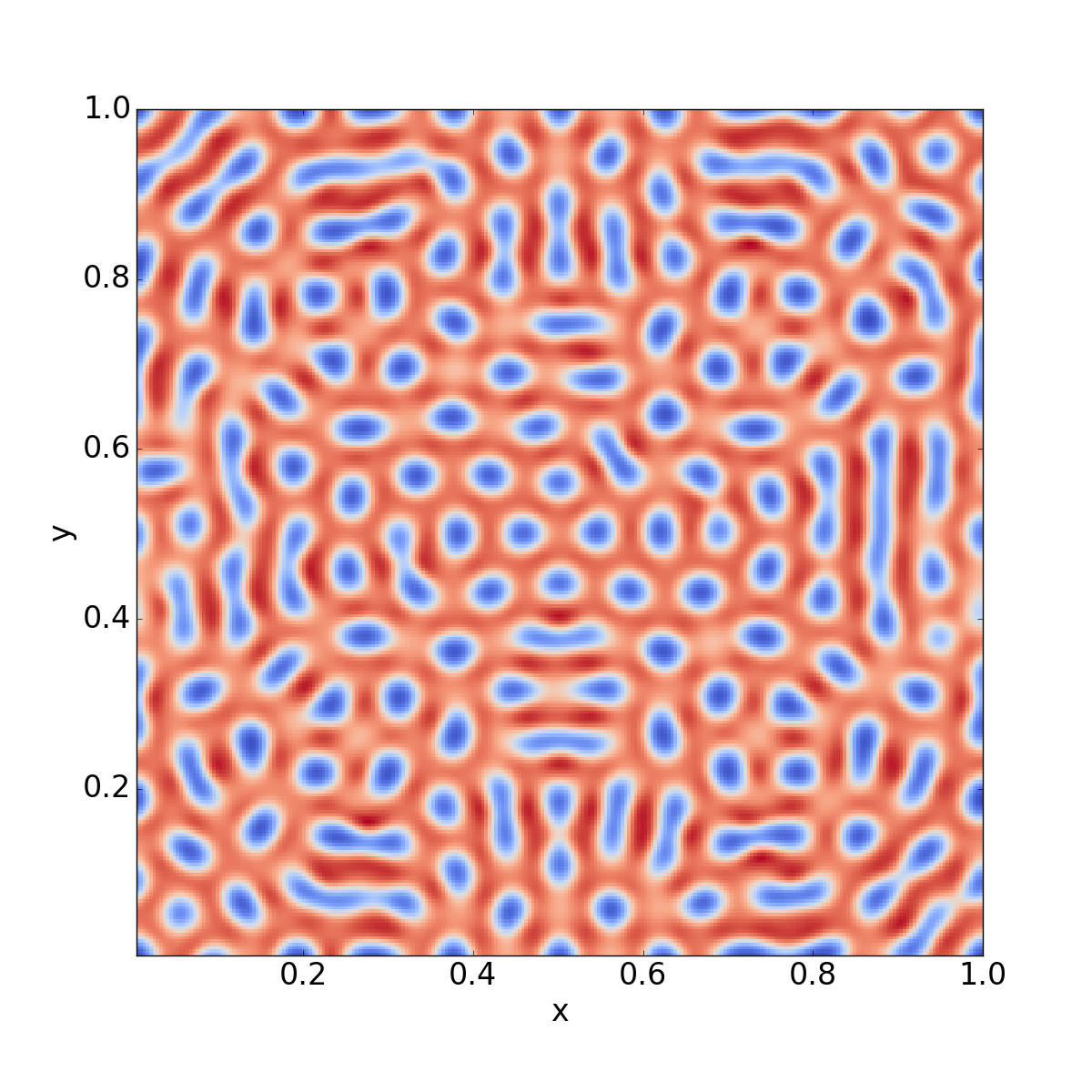}
	}
\subfigure[Learned evolution of $v$]{\label{fig: gs2_v_learnedevolution}	
	\includegraphics[width = 2 in]{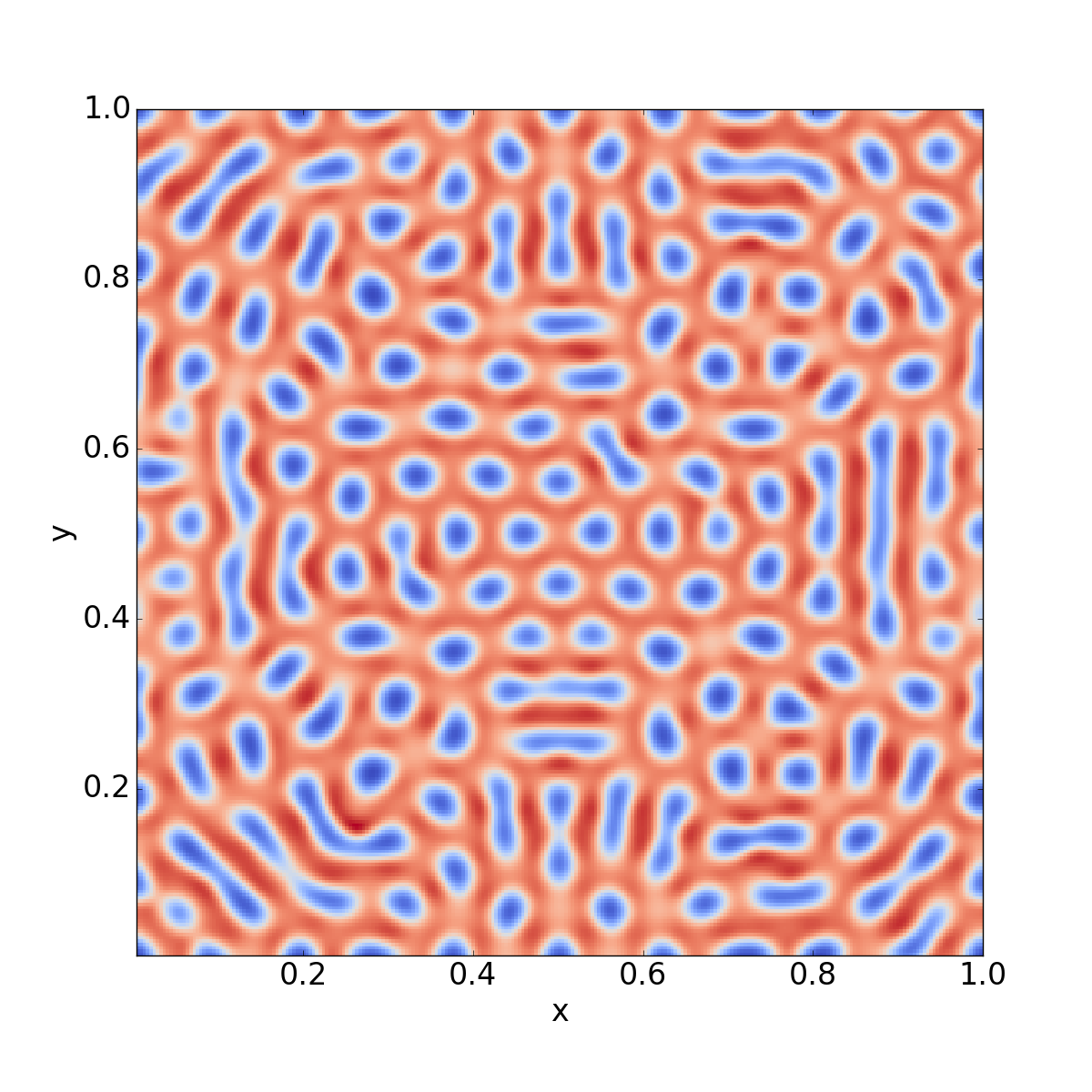}
	}
\subfigure[Difference in $v$]{\label{fig: gs2_v_difference}		
	\includegraphics[width = 2 in]{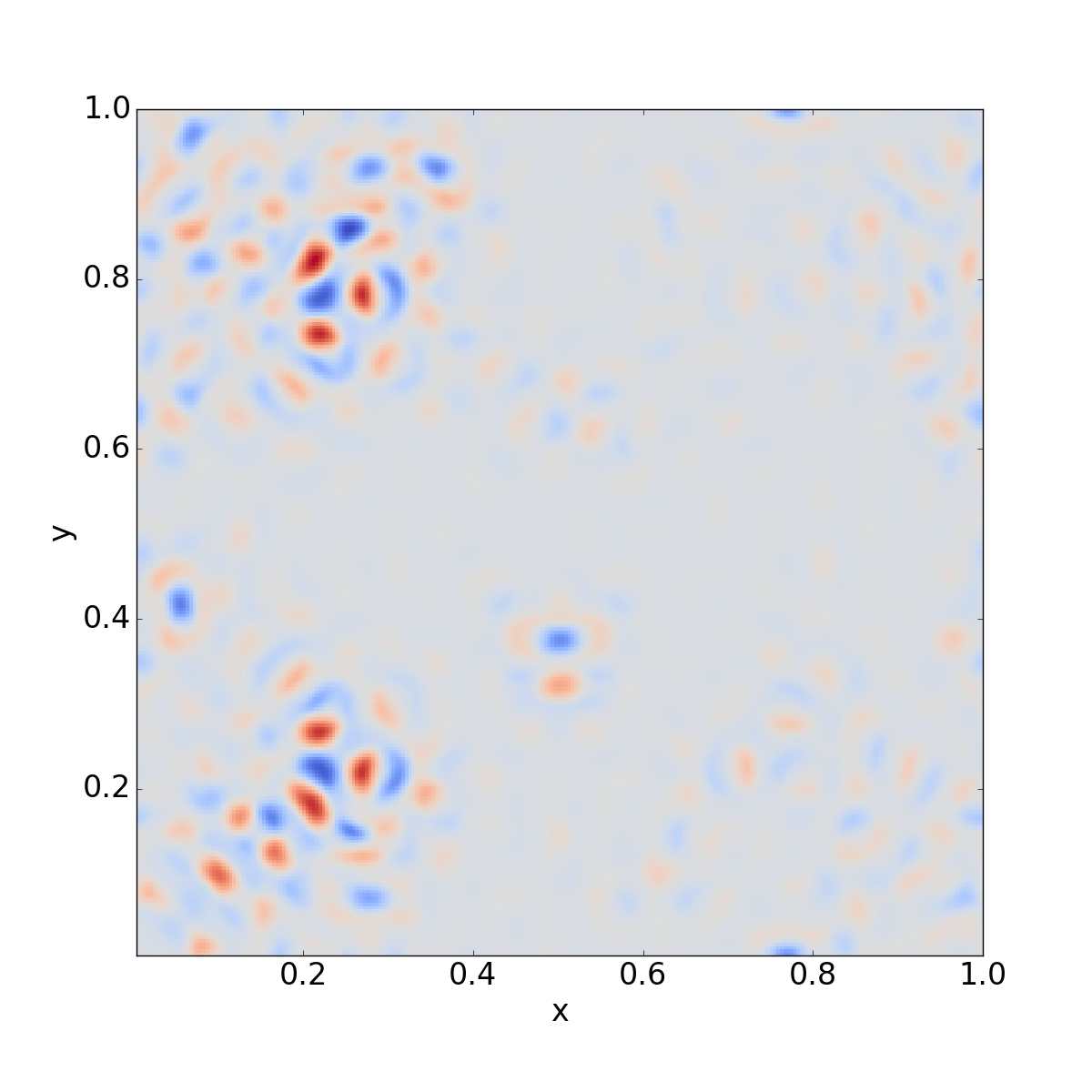}
	}
 \caption{\textbf{The Gray-Scott Equation, Example 2} 
 (a)(d) The true evolution at $T=2500$ using Equation \eqref{eq: gs2 exact}. (b)(e) The learned evolution at $T=2500$ using Equation \eqref{eq: gs2 learned}. (c)(f) The difference between the true evolution and the learned evolution. The overall qualitative structures are similar.}
\label{fig:grayscotts2}
\end{figure}

The last example uses the parameters $f = 0.018$ and $k = 0.051$, which leads to ``U'' shaped patterns. The visual results are given in Figure~\ref{fig:grayscotts3}. The learned equations are:
\begin{equation}
\begin{split}
u_t &= 0.30000\Delta u - 1.00000uv^2 - 1.01800u + 0.01801, \\
v_t &= 0.15000\Delta v + 1.00001uv^2 - 0.56902 v,
\end{split} \label{eq: gs3 learned}
\end{equation}
compared to the exact equations:
\begin{equation}
\begin{split}
u_t &= 0.3\Delta u - uv^2 - 1.018u + 0.018, \\
v_t &= 0.15\Delta v + uv^2 - 0.569v.
\end{split} \label{eq: gs3 exact}
\end{equation}
As before, we compare the learned and true evolutions, by simulating the two systems up to time-stamp $T=1000$, well beyond the learning interval. The location of the ``U'' shaped regions are correct; however, there is some error in their magnitude.

\begin{figure}[b!]
\centering
\subfigure[True evolution of $u$]{\label{fig: gs3_u_realevolution}
	\includegraphics[width = 2 in]{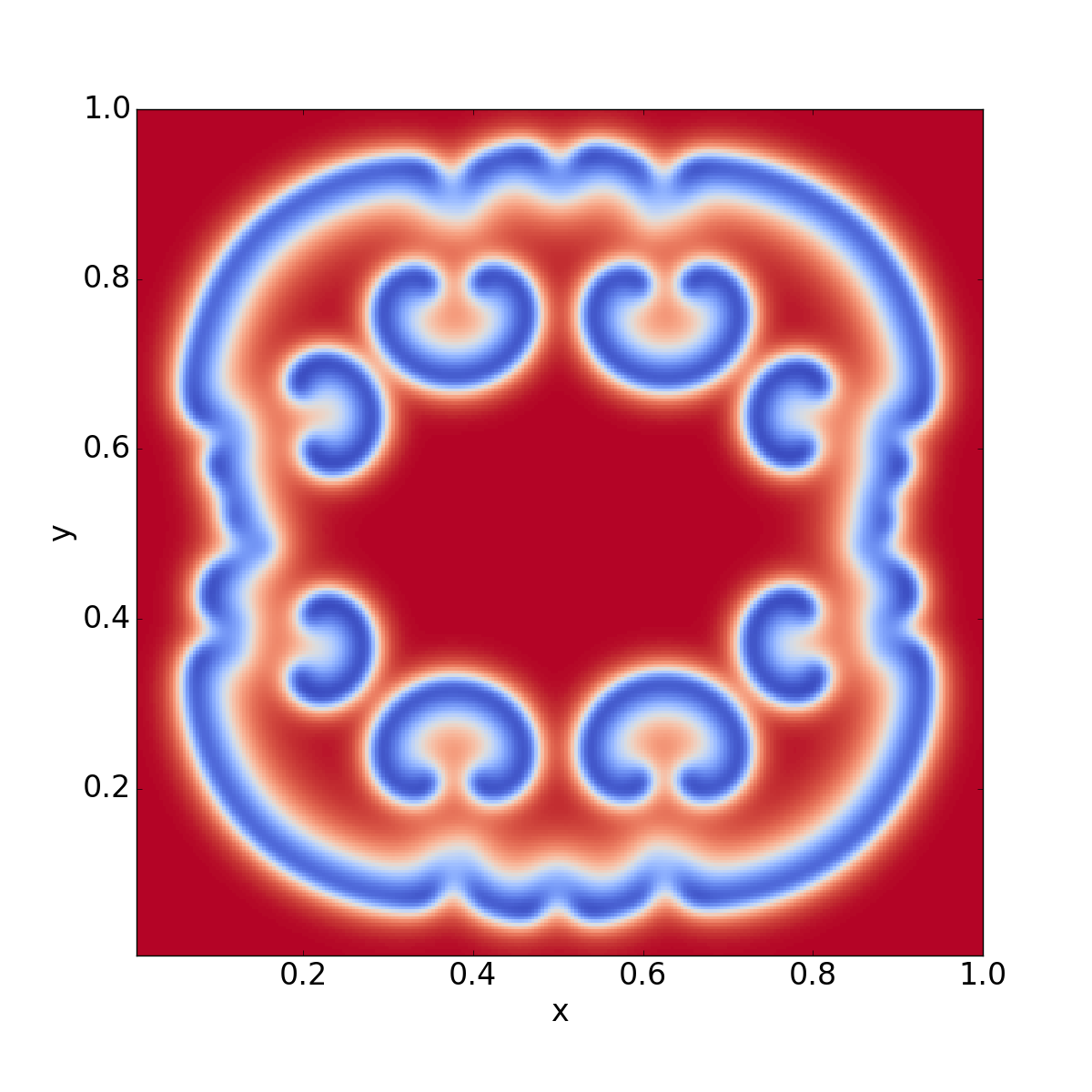}
	}
\subfigure[Learned evolution of $u$]{\label{fig: gs3_u_learnedevolution}	
	\includegraphics[width = 2 in]{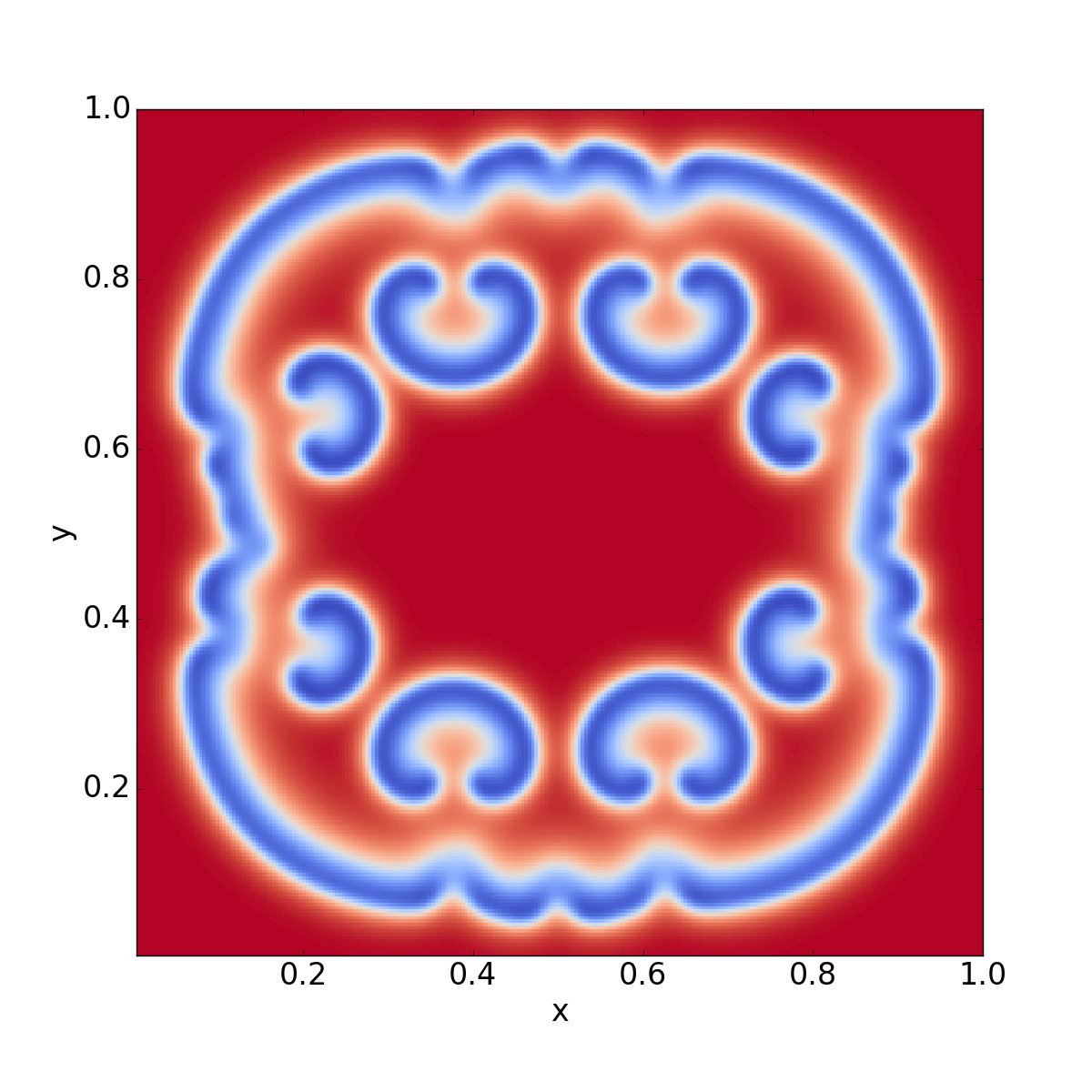}
	}
\subfigure[Difference in $u$]{\label{fig: gs3_u_difference}	
	\includegraphics[width = 2 in]{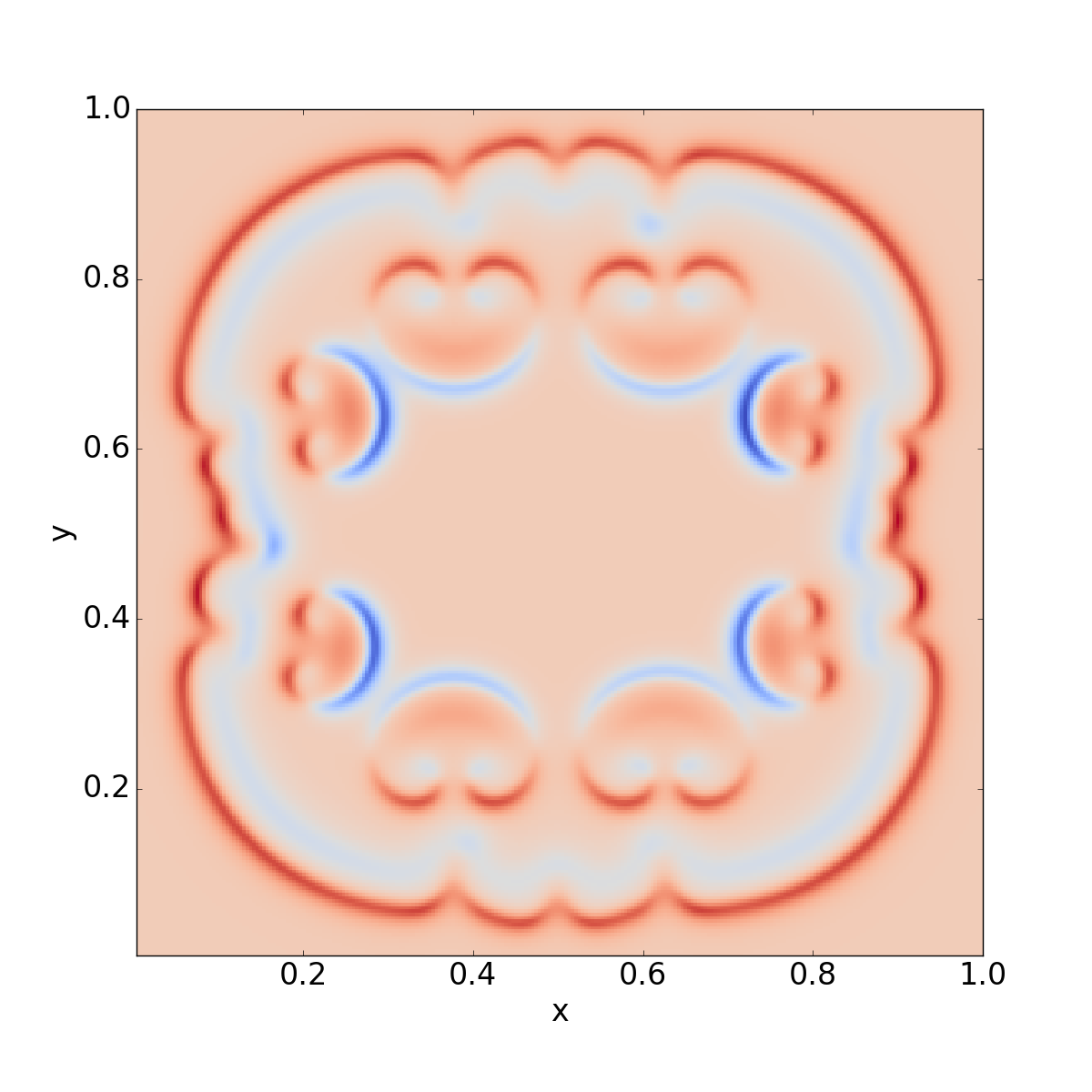}
	}
\subfigure[True evolution of $v$]{\label{fig: gs3_v_realevolution}	
	\includegraphics[width = 2 in]{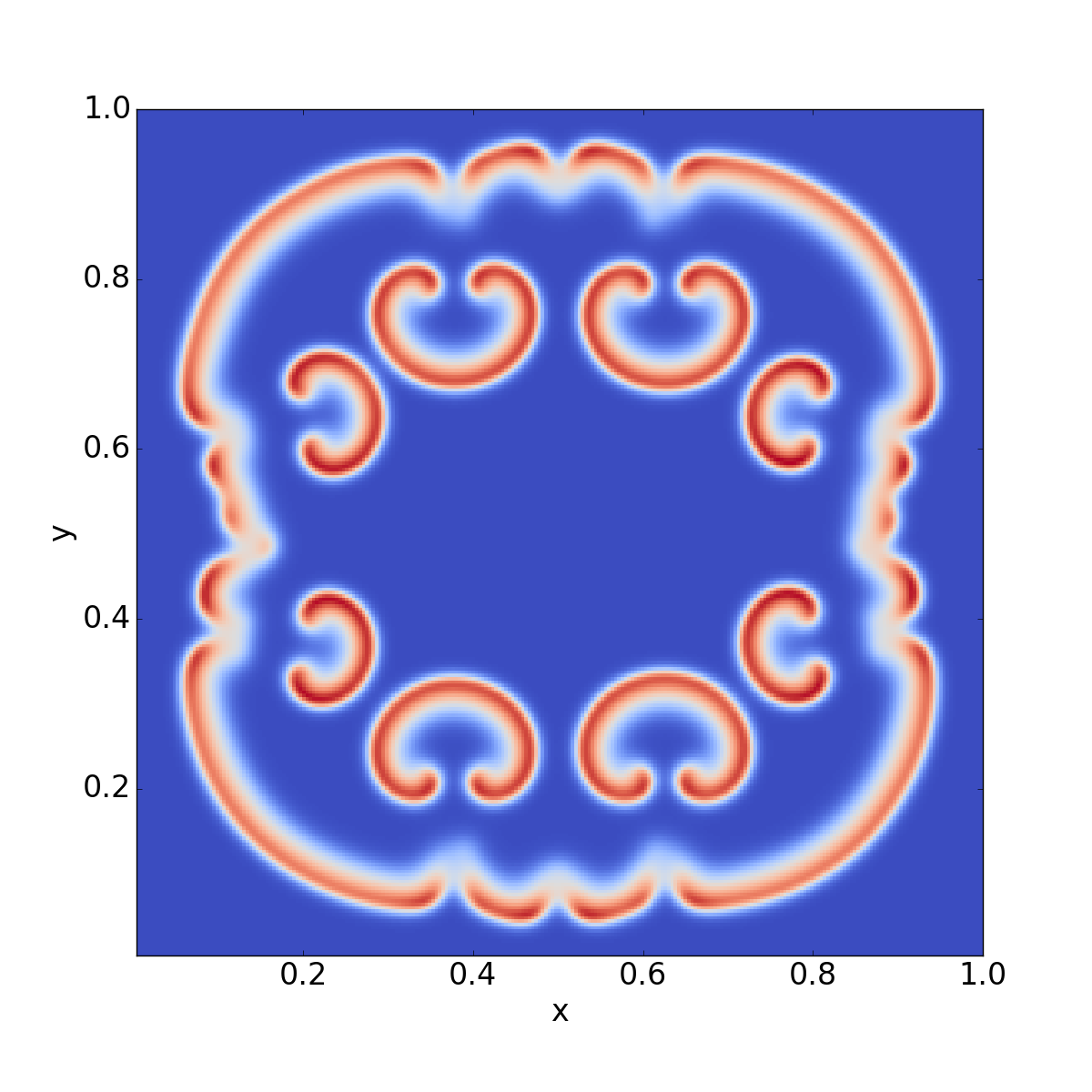}
	}
\subfigure[Learned evolution of $v$]{\label{fig: gs3_v_learnedevolution}	
	\includegraphics[width = 2 in]{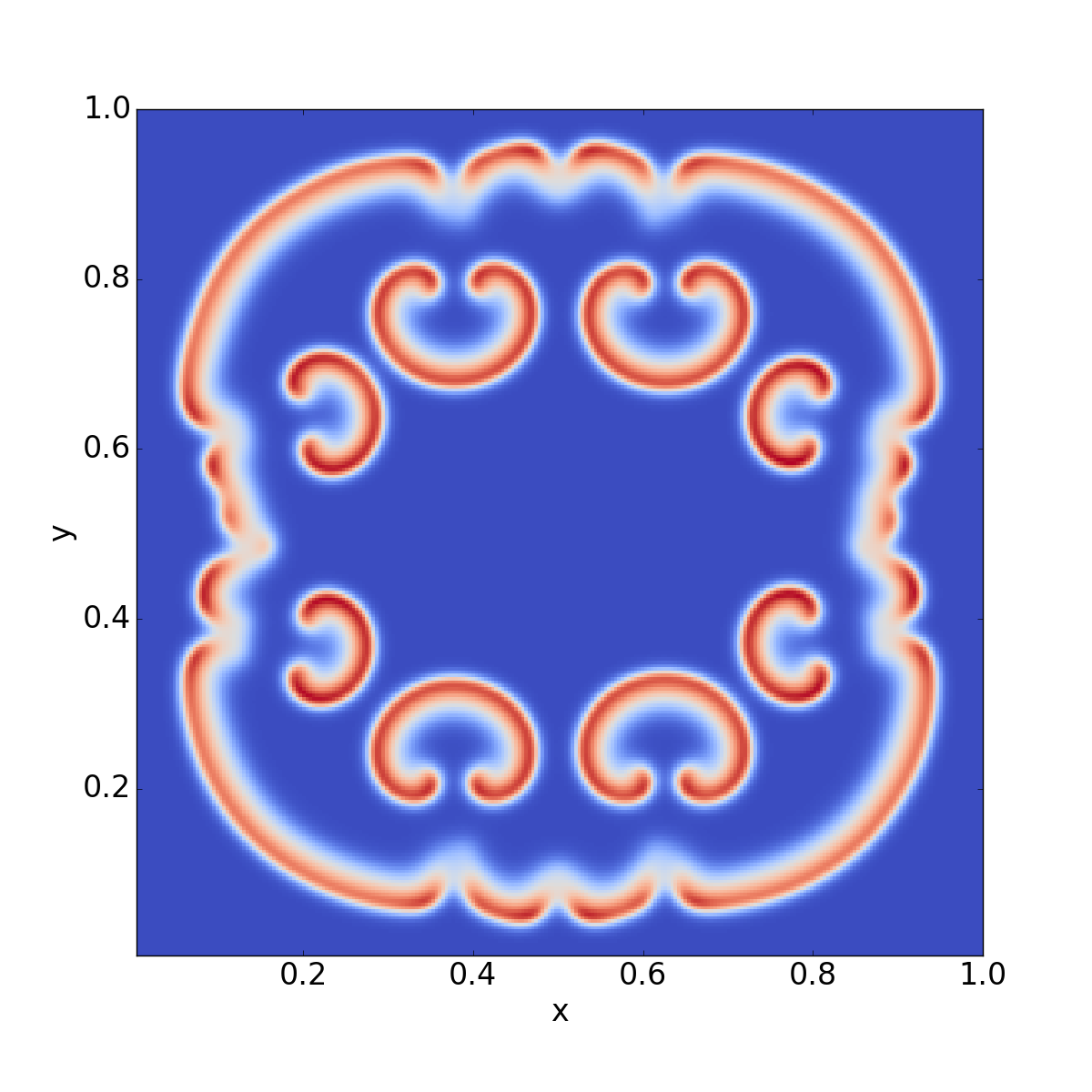}
	}
\subfigure[Difference in $v$]{\label{fig: gs3_v_difference}		
	\includegraphics[width = 2 in]{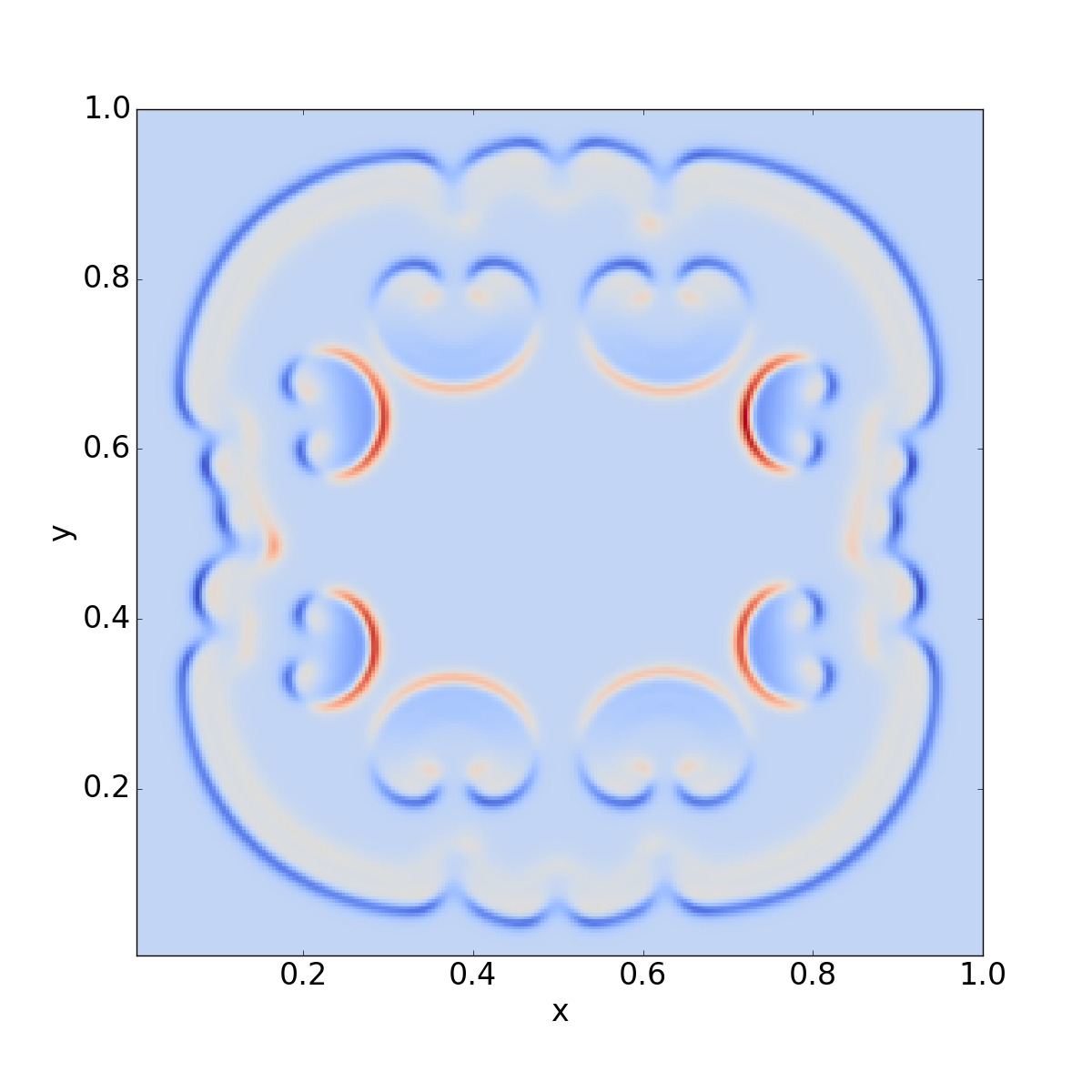}
	}
 \caption{\textbf{The Gray-Scott Equation, Example 3}: (a)(d) The true evolution at $T=1000$ using Equation \eqref{eq: gs3 exact}. (b)(e) The learned evolution at $T=1000$ using Equation \eqref{eq: gs3 learned}. (c)(f) The difference between the true evolution and the learned evolution. The location of the regions are nearly identical. The errors are due to a difference in magnitude.}
\label{fig:grayscotts3}
\end{figure}

\subsection{Discussion} \label{sec:resultsdiscuss}

In all of the examples found in Sections~\ref{sec: lorenz}-\ref{sec: grayscott}, the linear systems are under-determined.  Nevertheless, the model selected and parameters learned via Problem \eqref{eq:lbp} yield relatively accurate results.  The parameters used in the computational experiments in Sections~\ref{sec: lorenz}-\ref{sec: grayscott} are summarized in Tables~\ref{tab: parameter1} and \ref{tab: parameter2}, and the corresponding errors are displayed in Tables \ref{tab: error1} and \ref{tab: error2}.

In Tables \ref{tab: error1} and \ref{tab: error in equation}, we measure the relative error in the learned model by comparing the coefficients:
\begin{align*}
E_{c,\ {\rm LBP}} = \dfrac{\|c_{\rm exact} - c'\|_{\ell^2}}{\|c_{\rm exact}\|_{\ell^2}},
\end{align*}
where $c_{\rm exact}$ is the exact coefficient vector corresponding to the underlying system and $c'$ is the solution of Problem \eqref{eq:lbp}. The relative errors in the coefficients are within the theoretical bounds.  Thus from limited measurements, we are able to extract the governing equations with high-accuracy.

In Table~\ref{tab: error in evolution}, we display the relative error between the learned solution and the exact solution:
\begin{align*}
E_u = \dfrac{||u_{\rm exact}(T) - u(T)||_{\ell^2}}{||u_{\rm exact}(T)||_{\ell^2}},
\end{align*}
where $u_{\rm exact}(T)$ is the true evolution at the final time $T$, and $u(T)$ is the evolution at the final time $T$ with governing equation determined by $c$. The final time $T$ is outside of the interval used to learn the coefficients. In both Burgers' and Gray-Scott's equation, the relative error is within expectation. Note that small errors in the coefficients accumulate rapidly in these evolution equations, since the locations of sharp transitions in the solution will are sensitive to the coefficients.

\begin{table}[b!]
\caption{Parameters used in the computational qualitative experiments in Section~\ref{sec: lorenz}.}
\label{tab: parameter1}
\centering
\subtable[Parameters in matrix constructions.]{
\begin{tabular}{|c|c|c|c|} \hline
\multirow{2}{*}{} & \multicolumn{3}{c|}{The Lorenz 96 Equation} \\ \cline{2-4}
& Example 1 & Example 2 & Example 3 \\ \hline
Block size & $25$ & $25$ & $45$ \\ \hline
Number of bursts & 2 & 4 & 4 \\ \hline
Localization of the dictionary & $10$ & $10$ & 10 \\ \hline
Basis & 3rd order Legendre & 3rd order Legendre & 3rd order Legendre \\ \hline 
Size of the dictionary $n\times N$ & $50\times2040$ & $100\times2024$ & $204\times2024$ \\ \hline 
\end{tabular}} \\
\subtable[Parameters in Problem \eqref{eq:lbp}.]{
\begin{tabular}{|c|c|c|c|c|} \hline
\multicolumn{2}{|c|}{\multirow{2}{*}{}} & \multicolumn{3}{c|}{The Lorenz 96 Equation} \\ \cline{3-5}
\multicolumn{2}{|c|}{} & Example 1 & Example 2 & Example 3 \\ \hline
\multirow{3}{*}{$\sigma$} & var=0.2\% & 0.3515 & 0.53575 & 0.4075 \\ \cline{2-5}
& var=0.1\% & 0.3607 & 0.5074 & 0.7143 \\ \cline{2-5}
& var=0.05\% & 0.3380 & 0.5002 & 0.6888 \\ \hline
\end{tabular}}
\end{table}

\begin{table}[b!]
\caption{Parameters used in the computational qualitative experiments in Sections~\ref{sec: burgers}-\ref{sec: grayscott}.}
\label{tab: parameter2}
\centering
\subtable[Parameters in matrix constructions.]{\label{tab: parameters matrix}
\begin{tabular}{|c|c|c|} \hline
 & The Burgers' Equation & The Gray-Scott Equation \\ \hline
Block size & $7\times7$ & $7\times7$ \\ \hline
Number of bursts & 4 & 3 \\ \hline
Localization of the dictionary & $5\times5$ & $3\times3$ \\ \hline
Basis & 2nd order Legendre & 3rd order Legendre \\ \hline 
Size of the dictionary $n\times N$ & $196\times351$ & $147\times1330$ \\ \hline 
\end{tabular}} \\
\subtable[Parameters in Problem \eqref{eq:lbp}. For the Gray-Scott Examples, the top value is the $u$-component and the bottom value is the $v$-component.]{\label{tab: parameters LBP}
\begin{tabular}{|c|c|c|c|c|} \hline
& \multirow{2}{*}{The Burgers' Equation} & \multicolumn{3}{c|}{The Gray-Scott Equation} \\ \cline{3-5}
& & Example 1 & Example 2 & Example 3 \\ \hline
\multirow{2}{*}{$\sigma$} & \multirow{2}{*}{26.3609} & $5.5707\times 10^{-5}$& $6.3055\times 10^{-5}$ & $6.3321\times 10^{-5}$  \\
& & $5.2655\times 10^{-5}$ & $6.0194\times 10^{-5}$  & $6.0579\times 10^{-5}$ \\ \hline
\end{tabular}}
\end{table}

\begin{table}[b!]
\caption{Errors associated with the computational experiments in Section \ref{sec: lorenz}.}
\label{tab: error1}
\centering
\begin{tabular}{|c|c|c|c|c|} \hline
\multicolumn{2}{|c|}{\multirow{2}{*}{}} & \multicolumn{3}{c|}{The Lorenz 96 Equation} \\ \cline{3-5}
\multicolumn{2}{|c|}{} & Example 1 & Example 2 & Example 3 \\ \hline
\multicolumn{2}{|c|}{Size of the dictionary $n\times N$} & $50\times2040$ & $100\times2024$ & $204\times2024$ \\ \hline 
\multicolumn{2}{|c|}{Sampling rate $n/N\times100\%$} & 2.47\% & 4.94\% & 10.08\% \\ \hline
\multirow{3}{*}{$E_{c,\ {\rm LBP}}$} & var=0.2\% & 0.0276 & 0.0201 & 0.0185 \\ \cline{2-5}
& var=0.1\% & 0.0265 & 0.0197 & 0.0175 \\ \cline{2-5}
& var=0.05\% & 0.0254 & 0.0170 & 0.0165 \\ \hline
\end{tabular}
\end{table}

\begin{table}[b!]
\caption{Errors associated with the computational experiments in Sections \ref{sec: burgers}-\ref{sec: grayscott}. For the Gray-Scott examples, the top value is the $u$-component and the bottom value is the $v$-component.}
\label{tab: error2}
\centering
\subtable[Relative error, $E_{c,\ {\rm LBP}}$]{\label{tab: error in equation}
\begin{tabular}{|c|c|c|c|} \hline
\multirow{2}{*}{ Burgers' Equation} & \multicolumn{3}{c|}{ Gray-Scott Equation} \\ \cline{2-4}
& Example 1 & Example 2 & Example 3 \\ \hline
\multirow{2}{*}{0.0102} & $1.3775\times 10^{-5}$& $1.6550\times 10^{-5}$  & $1.6382\times 10^{-5}$  \\
& $1.6284\times 10^{-5}$  & $1.5525\times 10^{-5}$  & $1.5362\times 10^{-5}$  \\ \hline
\end{tabular}} \\
\subtable[Relative error, $E_u$]{\label{tab: error in evolution}
\begin{tabular}{|c|c|c|c|} \hline
\multirow{2}{*}{ Burgers' Equation} & \multicolumn{3}{c|}{ Gray-Scott Equation} \\ \cline{2-4}
& Example 1 & Example 2 & Example 3 \\ \hline
\multirow{2}{*}{0.0038} & $0.0353$ & $0.0856$& $0.0132$ \\
& $0.1416$  & $0.1221$  & $0.0453$  \\ \hline
\end{tabular}}
\end{table}

In Table \ref{tab: error1}, we display the relative errors $E_{c,\ {\rm LBP}}$ for the Lorenz 96 Equation with different noise levels and dictionary sizes. In all the examples, a large noise level and a small sampling rate lead to a higher relative error. As the sampling rate increases and the noise level decreases as well as the relative error. However, it is worth noting that the support sets are correctly identified.

Based on Theorem~\ref{thm:main}, for sufficient large block-sizes, it is possible to learn the correct coefficients with only one burst and one time-step. In Table \ref{tab: one burst error}, we display the relative errors $E_{c,\ {\rm LBP}}$ for the Burgers' Equation and the Gray-Scott Equation using one burst and varying the block-sizes. The block-sizes are chosen so that the linear systems remain under-determined (see second columns in Table \ref{tab: one burst error}). In both examples, starting with a small block size leads to a high relative error, but as the block-size increases the relative error decreases.

\begin{table}[b!]
\caption{Relative error $E_{c,\ {\rm LBP}}$ with one burst and varying block-sizes. }
\label{tab: one burst error}
\centering
\subtable[Burgers' Equation]{
\begin{tabular}{|c|c|c|} \hline
Block size & Size of the dictionary & $E_{c,\ {\rm LBP}}$ \\ \hline
$7\times7$ & $49\times351$ & 0.3285 \\ \hline
$9\times9$ & $81\times351$ & 0.0212 \\ \hline
$11\times11$ & $121\times351$ & 0.0187 \\ \hline
$15\times15$ & $225\times351$ & 0.0100 \\ \hline
\end{tabular}} \\
\subtable[The Gray-Scott Equation.]{
\begin{tabular}{|c|c|c|c|} \hline
\multirow{2}{*}{Block size} & \multirow{2}{*}{Size of the dictionary} & \multicolumn{2}{c|}{$E_{c,\ {\rm LBP}}$} \\ \cline{3-4}
& & $u$-component & $v$-component \\ \hline
$15\times15$ & $225\times1330$ & $7.1278\times10^{-1}$ & $2.6170\times10^{-1}$ \\ \hline
$21\times21$ & $441\times1330$ & $2.0613\times10^{-4}$ & $2.4336\times10^{-4}$ \\ \hline
$27\times27$ & $729\times1330$ & $1.4427\times10^{-5}$ & $2.4541\times10^{-5}$ \\ \hline
\end{tabular}}
\end{table}

For comparison, we calculate the least-square solution:
$$c_{\rm ls}=\argmin{c} \, ||Ac-V||_2,$$
where $A$ is the dictionary in the monomial basis and $V$ is the velocity matrix. The relative error associated with the least-squares solution is denoted by $E_{c,\, {\rm LS}}$. In Table~\ref{tab: least squares error}, we display $E_{c,\ {\rm LS}}$ for the Burgers' Equation and the Gray-Scott Equation corresponding to the same examples found in Table~\ref{tab: one burst error}. The least-squares solution produces large errors since the resulting coefficient vector is dense (overfitted), leading to meaningless results.

\begin{table}[h!]
\caption{Relative error $E_{c,\ {\rm LS}}$ with one burst and varying block sizes.  }
\label{tab: least squares error}
\centering
\subtable[Burgers' Equation]{
\begin{tabular}{|c|c|c|} \hline
Block size & Size of the dictionary & $E_{c,\ {\rm LS}}$ \\ \hline
$7\times7$ & $49\times351$ & 1.3633 \\ \hline
$9\times9$ & $81\times351$ & 1.7804 \\ \hline
$11\times11$ & $121\times351$ & 2.8562 \\ \hline
\end{tabular}} \\
\subtable[The Gray-Scott Equation.]{
\begin{tabular}{|c|c|c|c|} \hline
\multirow{2}{*}{Block size} & \multirow{2}{*}{Size of the dictionary} & \multicolumn{2}{c|}{$E_{c,\ {\rm LS}}$} \\ \cline{3-4}
& & $u$-component & $v$-component \\ \hline
$15\times15$ & $225\times1330$ & $5.8483$ & $7.7621$ \\ \hline
$21\times21$ & $441\times1330$ & $6.7622$ & $7.5338$ \\ \hline
$27\times27$ & $729\times1330$ & $5.9542$ & $2.9482$ \\ \hline
\end{tabular}}
\end{table}

\section{Conclusion and Future Directions}\label{sec:conclusion}

In this work, we presented an approach for extracting the governing equation from under-sampled measurements when the system has structured dynamics. We showed that permuting i.i.d randomly sampled bursts and restructuring the associated dictionary yields an i.i.d. random sampling of a bounded orthogonal system, thus using a Bernstein-like inequality with a coherence condition, we show that the recovery is exact and stable. In addition, when the noise is sufficiently low, then the support of the coefficients can be recovered exactly, \textit{i.e.} the terms in the governing equation can be exactly identified. The computational examples also highlight ways to extend the learning approach to larger systems and multi-component systems (where the cyclic structural condition must be carefully applied).

The structural assumption is valid for many dynamic processes, for example, when the data comes from a spatially-invariant dynamic system. In the algorithm and results, we made the assumption that one can sample a sub-block of the data to reasonable accuracy, in order to calculate derivatives accurately and so that the dictionary remains a bounded orthogonal system with respect to the given sampling measure. This is a weak assumption since the size of the sub-block is small. Thus, given noisy data, the sub-block could be the result of a pre-processing routine that has de-noised and down-sampled the data (with the removal of outliers). The details on the pre-processing requirements is left for future investigations. 

In future work, we would like to relax the requirements from cyclic structures to more general structures, while maintaining the very low-sampling requirement. 
An open question in the burst framework is how to quantify the trade-off between the burst size and the number of trajectories. This was not discussed in the numerical results but is of intrinsic importance to learning governing equations. In particular, if one has freedom to sample the initial condition but is limited by the number of samples along each trajectory or if one has freedom to sample along trajectories but is limited by the number of trajectories then it would be helpful to have theoretical estimates on the recovery rates. Lastly, the parameter $\sigma>0$ used in the constraint must be estimated from the data. It may be possible to learn $\sigma$ for a given dataset.

\section*{Acknowledgments}
H.S. and L.Z. acknowledge the support of AFOSR, FA9550-17-1-0125 and the support of NSF CAREER grant $\#1752116$. G.T. acknowledges the support of NSERC grant $\#06135$. R.W. acknowledges the support of NSF CAREER grant $\#1255631$.

\section*{Appendix}

Since the data is not independent, we cannot directly apply random matrix theory utilizing the standard Bernstein inequality. There are large deviations theorems for sums of ``locally dependent" centered random variables.  We recall and rephrase Theorem 2.5 from \cite{janson2004large} for our purposes.

\begin{theorem}[\textbf{Theorem 2.5 from }\cite{janson2004large}]\label{thrm:dependent}
Suppose that $Y = \sum\limits_{i=1}^n Y_{i}$ and that all $Y_{i}$ have the same distribution with $Y_{i} - \mathbb{E} Y_{i}  \leq M$ for some $M > 0$.  Suppose further that  $S = n \, {Var}(Y_{i})$, and $\triangle =  \triangle_{0}+ 1 $, where $\triangle_{0}$ is the maximal degree of the dependency graph $\Gamma$ for $\{Y_i\}$. Then
\begin{align}
P( Y - \mathbb{E}Y \geq \tau) \leq \exp\left( - \frac{\tau^2 ( 1- \triangle/(4 n))}{2 \triangle (S + M \tau / 3) } \right).
\end{align}
\end{theorem}

Note that when $\triangle = 1$, this reduces essentially to the standard Bernstein inequality for i.i.d. random variables. For our problem, the dependency graph $\Gamma$ has an edge between two vertices $Y_i$ and $Y_j$ if they have an index in common.

\bibliographystyle{plain}
\bibliography{STWZ_references2}

\end{document}